\documentclass[11pt,reqno,a4paper]{article}
\usepackage{amsmath}
\usepackage{amsbsy}
\usepackage{amsthm}
\usepackage{amssymb}
\usepackage{amscd}
\usepackage{mathabx}
\usepackage{accents}
\usepackage{float}
\usepackage{url}
\usepackage{tikz}
\usetikzlibrary{matrix,arrows,decorations.pathmorphing}
 \usepackage{relsize}
\usepackage{xcolor, mathrsfs}

\usepackage{hyperref}
\hypersetup{
pdftitle={},%
pdfauthor={},%
pdfsubject={},%
pdfkeywords={},%
colorlinks=true,%
linkcolor={black},%
linktoc={},
linktocpage={false},%
pageanchor={},
citecolor={black},
}

\usepackage[english]{babel}
\usepackage[utf8]{inputenc}
\usepackage[margin=2.41cm]{geometry}

\usepackage{hyperref}
\usepackage{url}

\usepackage{cite}

\usepackage{titlesec}
\titleformat{\section}{\large\bfseries\filcenter}{\thesection}{1em}{}
\titleformat{\subsection}{\bfseries}{\thesubsection}{1em}{}

\input{xy}
\xyoption{all}

\newtheorem{theorem}{Theorem}[section]
\newtheorem{theorem-intro}{Theorem}[]

\newtheorem{corollary}[theorem]{Corollary}
\newtheorem{lemma}[theorem]{Lemma}
\newtheorem{proposition}[theorem]{Proposition}
\theoremstyle{remark}

\theoremstyle{definition}
\newtheorem{remark}[theorem]{Remark}
\newtheorem{remarks}[theorem]{Remarks}

\newtheorem{definition}[theorem]{Definition}

\newtheorem{examples}[theorem]{Examples}

\newtheorem{notation}[theorem]{Notation}

\newtheorem{conjecture}[theorem]{Conjecture}
\newtheorem{question}[theorem]{Question}

\numberwithin{equation}{section}
\allowdisplaybreaks[1]

\catcode`@=12

\renewcommand\thanks[1]{%
  \begingroup
  \renewcommand\thefootnote{}\footnote{#1}%
  \addtocounter{footnote}{-1}%
  \endgroup
}

\renewcommand{\tilde}{\widetilde}
\renewcommand{\epsilon}{{\varepsilon}}
\def\bra{{\langle}}
\def\ket{{\rangle}}

\usepackage{bbm}

\newcommand{\ie}{\textit{i.e. }}
\newcommand{\cf}{\textit{cf. }}

\newcommand{\II}{{\mathbb{I}}}
\newcommand{\KK}{{\mathbb{K}}}
\newcommand{\RR}{{\mathbb{R}}}
\newcommand{\CC}{{\mathbb{C}}}
\newcommand{\n}{{\mathtt{n}}}

\newcommand{\Dir}{{\mathsf{D}}}
\newcommand{\G}{\mathsf{G}}
\newcommand{\Id}{{\operatorname{Id}}}

\newcommand{\E}{\mathsf{E}}
\newcommand{\F}{\mathsf{F}}

\newcommand{\M}{\mathsf{M}}
\newcommand{\N}{\mathsf{N}}
\renewcommand{\P}{\mathsf{P}}

\newcommand{\R}{\mathsf{R}}
\renewcommand{\S}{\mathsf{S}}
\newcommand{\T}{\mathsf{T}}

\newcommand{\D}{\mathsf{Dom}}
\newcommand{\Sol}{\mathsf{Sol}}
\newcommand{\Ker}{\mathsf{Ker}}

\newcommand{\bM}{\partial\M}

\newcommand{\ff}{{\mathfrak{f}}}

\newcommand{\fh}{{\mathfrak{h}}}

\newcommand{\fn}{{\mathfrak{n}}}

\newcommand{\vol}{{\textnormal{vol}\,}}
\newcommand{\supp}{{\textnormal{supp\,}}}

\newcommand{\fiber}[2]{\langle  #1\,|\, #2  \rangle}

\usepackage{pict2e}

\makeatletter
\DeclareRobustCommand{\intprod}{%
  \mathbin{\mathpalette\int@prod{(0.1,0)(0.9,0)(0.9,0.8)}}%
}
\DeclareRobustCommand{\intprodr}{%
  \mathbin{\mathpalette\int@prod{(0.1,0.8)(0.1,0)(0.9,0)}}}

\newcommand{\int@prod}[2]{%
  \begingroup
  \sbox\z@{$\m@th#1+$}%
  \setlength\unitlength{\wd\z@}%
  \begin{picture}(1,1)
  \roundcap
  \polyline#2
  \end{picture}%
  \endgroup
}
\makeatother

\begin{document}
\begin{center}
\
\vspace{5mm}

{\Large\bf  PARACAUSAL DEFORMATIONS OF \\[3mm] LORENTZIAN METRICS AND M{\O}LLER ISOMORPHISMS
\\[4mm]IN ALGEBRAIC QUANTUM FIELD THEORY} 

\vspace{5mm}

{\bf by}

\vspace{5mm}
\noindent
{  \bf  Valter  Moretti$^1$, Simone Murro$^2$ and Daniele Volpe$^1$}\\[2mm]
\noindent  $^1$ {\it Dipartimento di Matematica, Universit\`a di Trento and INFN-TIFPA}\\
{\it Via Sommarive 14,} {\it I-38123 Povo, Italy}\\[1mm]
\noindent  $^2$ {\it Dipartimento di Matematica, Universit\`a di Genova and INFN}\\
{\it Via Dodecaneso 35,} {\it I-16146 Genova, Italy}\\[2mm]

Emails: \ {\tt  valter.moretti@unitn.it, murro@dima.unige.it, daniele.volpe@unitn.it}
\\[10mm]
\end{center}

\begin{abstract}
Given a pair of  normally hyperbolic operators over (possibly different) globally hyperbolic spacetimes on a given smooth manifold, 
the existence of a geometric isomorphism, called {\em M\o ller operator}, between the space of solutions is studied. This is achieved by exploiting a new equivalence relation in the space of globally hyperbolic metrics, called {\em paracausal relation}. In particular, it is shown that the M\o ller operator associated to  a pair of paracausally related metrics and normally hyperbolic operators also intertwines the respective  causal propagators of the normally hyperbolic operators and it preserves the natural symplectic forms on the space of  (smooth) initial data. Finally, the M\o ller map is lifted to a $*$-isomorphism between (generally off-shell) $CCR$-algebras.  It is shown that the Wave Front set of a Hadamard bidistribution (and of a Hadamard state in particular)  is preserved by the  pull-back action of this  $*$-isomorphism.
\end{abstract}

\\[0.5mm]

\tableofcontents

\section{Introduction}
Recently a great deal of progress has been made in comparing the spaces of solutions of hyperbolic partial differential equations on (possibly different) Lorentzian manifolds as well as in the comparison of the associated quantum field theories. More precisely, given a pair $\N$ and $\N'$ of Green hyperbolic differential operators on (possibly different) globally hyperbolic spacetimes $(\M,g)$ and $(\M,g')$, a natural issue concerns  the existence of a linear isomorphism $\S : \Sol_\N \to \Sol_{\N'}$ between the linear spaces of the solutions of the equations $\N\psi=0$ and $\N'\psi'=0$.
Such an isomorphism, if it exists, is called a {\em M\o ller map}.
Since the said space of solutions is the first step in the construction of corresponding (algebraic) free quantum field theories, a natural related issue concerns the possibility to promote the M\o ller map $\R$ to a $*$-isomorphism between the associated abstract operator algebras $\mathcal{A}$ and $\mathcal{A}'$ constructed out of $\N$ and $\N'$ respectively on $(\M,g)$ and $(\M,g')$, in terms of corresponding generators given by {\em abstract field operators} $\Phi(ah)$ and $\Phi'(\fh')$ and the associated {\em causal propagators} $\G_\N$, $\G_{\N'}$.
Actually, {\em off-shell} linear QFT can be used to build up a perturbative approach to interacting QFT, a final problem would concern the possibility to extend the M\o ller isomorphism of algebras to an isomorphism of more physically interesting algebras, for instance including Wick powers or time-ordered powers.\\
These problems have been tackled in the past for special cases of metrics $g,g'$ and several types of Green hyperbolic field operators which rule the dynamics of bosonic  fields~\cite{DHP,Moller} or fermionic fields~\cite{DefArg1,MollerMIT}.  In the {\it loc. cit.}, the pairs of Lorentzian metrics $g,g'$  had to satisfy one of the following assumptions: (I) they shared a common foliation of smooth spacelike Cauchy surfaces; (II) they coincided outside a compact set.\medskip

 In this paper, we do not assume either of the above restrictions and instead, we consider a very wide family of pairs $g,g'$ of globally hyperbolic metrics on the same manifold $\M$. As a matter of fact, we consider a new type of relationship between 
globally hyperbolic metrics, which we call {\em paracausal relationship}. To the authors' knowledge this notion  represents a complete novelty on the subject. 
Though the effective definition of paracausal equivalence relation in the set of globally hyperbolic metrics on $\M$ (Definition \ref{def:paracausal def}) is different and more effective for the issues regarding M\o ller maps  raised above, a complete characterization of it can be stated as follows in terms of elementary Lorentzian geometry:

\begin{theorem-intro}[Theorem~\ref{thm:char paracausal}]
 The globally hyperbolic metric $g$ on $\M$ is paracausally related to the globally hyperbolic metric $g'$ on $\M$ if and only if there is a finite sequence $g_0:=g, g_1,\ldots, g_N:=g'$ of globally hyperbolic metrics on $\M$ such that, at each step $g_k,g_{k+1}$, the future open light cones of these metrics have non-empty intersection $V_x^{g_k^+} \cap V_x^{g_{k+1}+} \neq \emptyset$ at every point $x\in \M$.
\end{theorem-intro}

\noindent  The class of paracausally related metrics on a given manifold $\M$ is very large, though some elementary counterexamples of topological nature can be constructed.  A specific study on the properties of this equivalence relation is necessary and it will be done elsewhere.\medskip

Equipped with this notion,  the paper specializes the analytic setup as follows: $\N,\N' : \Gamma(\E) \to \Gamma(\E)$
are  (2nd order) {\em normally hyperbolic} operators  on a real or complex vector bundle $\E$ on $\M$, respectively associated to a pair of globally hyperbolic metrics $g,g'$ on $\M$.

The first important achievement of this work  is the proof of existence of (infinitely many) {\em M\o ller operators}, i.e., isomorphisms $\R: \Gamma(\E) \to \Gamma(\E)$, 
which restrict to M\o ller maps $\S$
between the space of solutions when $g$ and $g'$ are paracausally related.
The overall idea is inspired by the scattering theory in the special case of a pair of globally hyperbolic metrics $g_0,g_1$ over $\M$ such that the light cones of $g_0$ are included in the light cones of $g_1$ (this is the most elementary case of paracausal relation).  We start  with  two ``free theories'',  described  by the space of solutions of normally hyperbolic operators $\N_0$ and $\N_1$ in corresponding  spacetimes $(\M,g_0)$ and $(\M,g_1)$, respectively,  and we intend  to connect them through  an ``interaction spacetime'' $(\M,g_\chi)$ with a ``temporally localized'' interaction defined by interpolating  the two metrics by means of a smoothing function $\chi$. Here we need two M{\o}ller maps: $\Omega_+$ connecting $(\M,g_0)$ and $(\M, g_\chi)$ --
which reduces to the identity in the past when $\chi$ is switched off --
and a second M{\o}ller map connecting $(\M,g_\chi)$ to $(\M,g_1)$ -- which reduces to the identity in the future when $\chi$ constantly takes the value $1$. The ``$S$-matrix'' given by the composition 
$\S :=\Omega_-\Omega_+$
 will be the M\o ller map connecting $\N_0$ and $\N_1$. 

The above construction generalizes to the case of a pair of globally hyperbolic metrics $g,g'$ on $\M$ which are paracausally related and this fact is denoted by $g \simeq g'$.
A summary of the main results obtained is the following  where also a special notion of adjoint operator $\R^{\dagger_{gg'}}$ is used. It will be discussed in details in Section 4.

\begin{theorem-intro}[Theorems ~\ref{remchain}, \ref{thm:pres sympl}, and ~\ref{teoRRR}] 
Let  $\E$ be $\KK$-vector bundle over the smooth manifold $\M$ with a non-degenerate, 
real or Hermitian depending on $\KK$,  fiber metric $\fiber{\cdot}{\cdot}$.
 Consider $g,g' \in \mathcal{GH}_M$ with respectively associated normally hyperbolic formally-selfadjoint operators $\N$, $\N'$.\\
If the metrics are paracausally related $g\simeq g'$, then  it is  possible to define a (non-unique) $\KK$-vector space isomorphism $\R : \Gamma(\E) \to \Gamma(\E)$, called {\bf M\o ller operator} of $g,g'$ (with this order), such that the following facts are true.
\begin{itemize}
\item[(1)] The restrictions to the relevant subspaces of $\Gamma(\E)$ respectively define symplectic M\o ller maps $S^0$ (see Definition~\ref{def:sympl Moll}) which preserve the symplectic forms $\sigma_g^\N$, $\sigma_{g'}^{\N'}$ defined as in Equation~\eqref{def:sympl form}, namely 
$$\sigma^{\N'}_{g'}(\S^0\Psi,\S^0 \Phi) = \sigma^{\N}_{g}(\Psi,\Phi) \quad \mbox{for every $\Psi,\Phi \in \Ker^{g}_{sc}(\N)$.}$$
\item[(2)]   The causal propagators 
$\G_{\N'}$ and $\G_{\N}$, respectively
of $\N'$ and $\N$, satisfy $ \R \G_{\N}  \R^{\dagger_{gg'}} = \G_{\N'}$, where $\R^{\dagger_{gg'}}$ is the adjoint of the M\o ller operator (see Definition~\ref{defadjointdiff}).
\item[(3)]  By denoting $c'$ the smooth function such that $\vol_{g'} = c' \: \vol_{g}$, we have $ c' \N'\R  = \N\:.$
\item[(4)] It holds 
$\R^{\dagger_{gg'}} \N'|_{\Gamma_c(\E)}= \N|_{\Gamma_c(\E)}\: .$
\item[(5)] The maps $\R^{\dagger_{gg'}} : \Gamma_c(\E) \to \Gamma_c(\E)$ and $(\R^{\dagger_{gg'}})^{-1} = (\R^{-1})^{\dagger_{g'g}}: \Gamma_c(\E) \to \Gamma_c(\E)$ are continuous with respect to the natural  topologies of $\Gamma_c(\E)$  in the domain and in the co-domain.
\end{itemize}
\end{theorem-intro}

Theorem 2 permits us  to promote $\R$ to a  $*$-isomorphism of the algebras of field operators $\mathcal{A}$, $\mathcal{A}'$
respectively associated to the paracausally related metrics $g$ and $g'$ (and the associated $\N,\N'$) and
 generated by respective  field operators 
$\Phi(\ff)$ and $\Phi'(\ff')$ with $\ff,\ff'$ compactly supported smooth sections of $\E$. These field operators satisfy respective CCRs
$$[\Phi(\ff), \Phi(\fh)]= i \G_\N(\ff,\fh) \II\:, \quad [\Phi'(\ff'), \Phi'(\fh')]= i \G_{\N'}(\ff',\fh')\II'$$
and the said unital $*$-algebra isomorphism $\mathcal{R}: \mathcal{A}' \to \mathcal{A}$ is determined by the requirement (Proposition~\ref{algebraicMoller})
$$\mathcal{R}(\Phi'(\ff')) = \Phi(\R^{\dagger_{gg'}} \ff)\:.$$
The final important result regards the properties of $\mathcal{R}$ for the algebras of a pair of paracausally related metrics $g,g'$ when it acts on the states $\omega : \mathcal{A}\to \CC$, $\omega' : \mathcal{A}'\to \CC$ of the algebras in terms of pull-back.  
$$\omega' = \omega \circ \mathcal{R}\:.$$
As is known, the most relevant (quasifree)  states in algebraic QFT are {\em Hadamard states} characterized by a certain wavefront set of their two-point function. To this regard, we prove that the pull-back through $\mathcal{R}$ of a Hadamard state $\omega : \mathcal{A}\to \CC$ is a Hadamard state  of the off-shell algebra $\mathcal{A}'$, provided the metrics $g,g'$ be paracausally related. 
The result is extended to a generic bidistribution $\nu$ (corresponding to the two-function of $\omega$, dropping the remaining requirements included in the definition of state).
The proof of the theorem below  is both of geometrical and microlocal analytic nature (see also Theorem \ref{thm:main Had}).

\begin{theorem-intro}[Theorem  \ref{thm:main intro Had}]
Let $\E$ be an $\RR$-vector bundle on a  smooth manifold $\M$ equipped with a  non-degenerate, symmetric, fiberwise  metric  $\fiber{\cdot}{\cdot}$.
 Let  $g,g' \in \mathcal{GH}_\M$, consider the corresponding  formally-selfadjoint normally hyperbolic operators $\N,\N' : \Gamma(\E) \to \Gamma(\E)$ and refer to the associated CCR algebras $\mathcal{A}$ and $\mathcal{A}'$.\\
Let   $\nu   \in \Gamma_c'(\E \boxtimes\E)$ be of Hadamard type and satisfy
$$\nu(x,y) - \nu(y,x) = i\G_{\N}(x,y)\quad mod \quad C^\infty\:,$$
$\G_{\N}(x,y)$ being the distributional Kernel of $\G_{\N}$. \\
 Assuming  $g\simeq g'$,  let us define $$\nu' := \nu \circ \R^{\dagger_{gg'}} \otimes \R^{\dagger_{gg'}}\:,$$
for a M\o ller operator $\R: \Gamma(\E) \to \Gamma(\E)$ of $g,g'$.
Then the following facts are true.
\begin{itemize} 
\item[(i)]   $\nu$ and $\nu'$ are bisolutions mod $C^\infty$ of the field equations defined by $\N$ and $\N'$ respectively,
\item[(ii)] $\nu' \in  \Gamma_c'(\E \boxtimes\E)$,
\item[(iii)] $\nu'(x,y) - \nu'(y,x) = i\G_{\N'}(x,y)$ mod $C^\infty$,
\item[(iv)] $\nu'$ is of Hadamard type.
\end{itemize}
\end{theorem-intro}

As this crucial result concerns off-shell algebras, in principle, it could be exploited in perturbative constructions of interacting theories. Indeed the preservation of
the Hadamard singularity structure plays a crucial role in the development of the perturbative
theory~\cite{DHP}.
Another work will be devoted to this  investigation.\medskip

The work  is organized as follows. Section \ref{sec:Lorentz} contains a recap on the relevant notions of Lorentzian geometry we exploit throughout. In particular, in Section \ref{subsec:convex} we introduce some (apparently new) results about convex interpolations of globally hyperbolic metrics which are preparatory to Section \ref{subsec:paracausal} where we present the definition of paracausal relation and we give some basic results about this equivalence relation.
Section \ref{sec:norm hyp} is completely devoted to recalling some notions and fundamental results about Green hyperbolic operators, normal hyperbolic results and their interplay with convex combinations of Lorentzian metrics. 
Section~\ref{sec:Moller} is the core of the paper. In the Subsections~\ref{subsec:moller 1} and~\ref{subsec:moller 2} M\o ller maps are introduced under the additional assumption that the metrics $g_0,g_1\in\mathcal{GH}_\M$ satisfy $g_0 \preceq g_1$. The latter is removed in Subsection~\ref{subsec:Moller general}, where the analysis is extended to encompass paracausally deformed metrics $g\simeq g'$. In Subsection~\ref{subsec:moller sympl form}, it is shown that the M\o ller map preserves the natural symplectic form on the space of initial data and finally, in Subsection~\ref{subsec:moller oper} the M\o ller operator for paracausally deformed metrics is introduced and analyzed in detail. Section~\ref{sec:Moller AQFT} is devoted to the study of free quantum field theories on globally hyperbolic spacetimes. In particular, in Subsection~\ref{subsec:algebraic moller} we lift the M\o ller operator to a $*$-isomorphism of algebras of observables and in~Subsection~\ref{subsec:Moller Hadamard} we show that the pull-back of any quasifree state along the M\o ller $*$-isomorphism preserves the Hadamard condition. Finally, we conclude our paper with Section~\ref{sec:concl}, where open issues and future prospects are presented.

\subsection*{General notation and conventions}
\begin{itemize}
\item[-] $A \subset B$ permits the case $A=B$, otherwise we write $A \subsetneq B$.
\item[-] The symbol $\KK$ denotes any  element of  $\{\RR,\CC\}$.
\item[-] Tensor fields and sections of $\KK$-vector bundles  on $\M$ are always supposed to be smooth.
\item[-] $(\M,g)$ denotes a $(n+1)$-dimensional spacetime (\cf Definition~\ref{def:spacetime})
and we adopt
the convention that $g$ has the signature $(-,+\dots,+)$.
\item[-] $\sharp:\Gamma(\T^*\M)\to\Gamma(\T\M)$ and its inverse  $\flat:\Gamma(\T\M)\to\Gamma(\T^*\M)$ denote the standard (fiberwise) {\bf musical isomorphisms} (\cf Section~\ref{sec:partial ordering}) referred to a given metric $g$ on $\M$.
\item[-]  $\mathcal{M}_\M$, $\mathcal{T}_\M\subset \mathcal{M}_\M$ and $\mathcal{GH}_\M \subset \mathcal{T}_\M$ denote  respectively the sets of  smooth Lorentzian metrics, {\bf time-oriented} Lorentzian metric and  {\bf globally hyperbolic} metrics on $\M$;

\item[-]$g  \preceq g'$ denotes that $g,g' \in \mathcal{M}_\M$ and the open light cone $V^g_p$ of $g$ is a subset of the open lightcone $V_p^{g'}$ of $g'$ at every point $p\in \M$;
\item[-]$g  \simeq g'$ denotes that $g$ and $g'$ are {\bf paracausally related} (\cf Definition~\ref{def:paracausal def}).

\end{itemize} 

\subsection*{Acknowledgments}
We are grateful to Nicol\`o Drago, Nicolas Ginoux,  and Miguel S\'anchez for helpful discussions related to the topic of this paper. We are grateful to the referees for useful comments on the
manuscript.
This work  was produced within the activities of the INdAM-GNFM.

\subsection*{Funding}
S.M was supported by the DFG research grant MU 4559/1-1 ``Hadamard States in Linearized Quantum Gravity''.
 V.M. and D.V. acknowledge the support of the INFN-TIFPA project ``Bell''.

\section{Convex interpolation and  paracausal deformations of Lorentzian metrics}\label{sec:Lorentz}

The aim of this section is twofold. On the one hand we shall investigate the properties of {\em convex interpolation of Lorentzian metrics}, on the other hand we introduce the notion of  {\em paracausal deformations of globally hyperbolic metrics}. 
As we shall see, these mathematical tools 
rely on a certain preordering relation in the set of Lorentzian metrics on a given manifold, they are quite interesting on their own right and they will be exploited in the second part of this work to construct M\o ller operators and M\o ller $*$-isomorphisms of algebras of quantum fields.

\subsection{Preliminaries on Lorentzian geometry}\label{sec:preliminaries}

The aim of this section is to recall some basic results of Lorentzian geometry which we will need later on. For a more detailed introduction to Lorentzian geometry we refer to~\cite{Ba-lect,Bee,Oneill}.

\subsubsection{Lorentzian manifolds and cones}

Let $\M$ be a smooth connected paracompact Hausdorff manifold and assume that $\M$ is noncompact or its Euler characteristic vanishes. Under these assumptions, $\M$ admits a Lorentzian metric and we denote the space of Lorentzian metrics on $\M$ by  $\mathcal M_\M$ (see {\it e.g.}~\cite{Bee}).
Once that a Lorentzian metric  $g$ is assigned to a smooth manifold $\M\ni p$, we can classify the vectors $v_p \in \T_p\M$ into three different types:\begin{itemize} \item {\bf spacelike} \ie $g(v_p,v_p) > 0 $ or $v_p=0$ , \item {\bf timelike} \ie $g(v_p,v_p) <0$, \item  {\bf lightlike} (also called {\bf null}) \ie $g(v_p,v_p) =0$ {\em and $v_p \neq 0$}.  
\end{itemize}

\begin{remark}
Notice that, with our definition, the tangent vector $0$ is spacelike.
\end{remark}

 As usual, we denote as {\bf causal vectors} any timelike or lightlike vector.
Piecewise smooth curves are classified analogously according to the nature of their tangent vectors. 

 Keeping in mind this classification, the open {\bf lightcone} of $(\M,g)$ at $p\in \M$ is the set $$V^g_p:= \{v_p \in \T_p\M \:|\: g(v_p,v_p) <0\}\:.$$
It is not difficult to see that it is an open convex cone made of two disjoint open convex halves defining the two connected components of $V^g_p$.

The notion of {\em time orientation}  is  defined as in \cite{Ba-lect}: A smooth Lorentzian manifold  $(\M,g)$ is  said to be {\bf time-orientable} if there is a continuous  timelike vector field $X$ on $\M$.

If $(\M,g)$ is time orientable  and  a preferred continuous timelike vector field $X$ has been chosen  as above, the {\bf future lightcone}   $V^{g+}_p \subset V^{g}_p$ at $p\in \M$  is the connected component of  $V^{g}_p$  containing $X_p$. The other connected component $V^{g-}_p$ is the {\bf past lightcone} at $p$.  
 $V^{g+}_p$ and  $V^{g-}_p$ respectively includes  the  {\bf future-directed} and {\bf past-directed} timelike  vectors at $p$. The terminology extends to the causal  (lightlike) vectors which belong to the closures of the said halves.
A classification of (piecewise smooth) causal curves into past-directed and future-directed curves (see \cite{Ba-lect}) arises according to their tangent vectors.  

If $(\M,g)$ is time orientable, 
the continuous choice of one of the two halves of $V_p^g$ for all $p\in \M$ through a continuous timelike vector field as above  defines a {\bf time orientation} of $(\M,g)$. 
$(\M,g)$ with this choice of preferred halves of cones  is said to be {\bf time oriented}. If $(\M,g)$ is connected and time orientable, then it admits exactly two time orientations.

\begin{notation}
In the following, we denote with  $\mathcal{M}_\M$, the set of smooth Lorentzian metrics on the smooth manifold $\M$ and with $\mathcal{T}_\M$  the class of time-oriented Lorentzian metrics on $\M$.
\end{notation}

We have an elementary fact whose proof is immediate if working in a $g$-orthonormal basis.
\begin{proposition}\label{propzero}
Assume that $g \in \mathcal{T}_\M$, $p\in \M$, and $Y_p, Z_p \in V^{g}_p$. 
Then
\begin{itemize}
\item[(i)] $Y_p \in V^{g\mp}_p$ and  $Z_p \in V^{g\pm}_p$ if and only if $g(Y_p,Z_p)>0$,
\item[(ii)] $Y_p, Z_p \in V^{g\pm}_p$ if and only if $g(Y_p,Z_p)<0$.
\end{itemize}
\end{proposition}

\medskip If $g \in  \mathcal{M}_\M$, the associated standard (fiberwise) {\bf musical isomorphism} $\sharp:\Gamma(\T^*\M)\to \Gamma(\T\M)$ is pointwise  defined by  $$g(\sharp\omega_p,v_p)=\omega_p(v_p)\quad \mbox{for every   $v\in\Gamma(\T\M)$ and $\omega \in \Gamma( \T^*\M)$ and $p\in \M$,}$$ and we denote the (fiberwise) {\bf  inverse  musical isomorphism}  by 
$\flat:\Gamma(\T\M)\to\Gamma(\T^*\M)$. 
The notation
$g^\sharp \in \Gamma(\T\M\otimes \T\M)$ indicates  the Lorentzian metric induced on $1$-forms from $\sharp$ as
$$g^\sharp (\omega_{1p},\omega_{2p})= g(\sharp{\omega_1}_p,\sharp{\omega_2}_p)\quad \mbox{for every   $\omega_1,\omega_2 \in\Gamma(\T^*\M)$ and $p\in \M$.}$$
Once that a Lorentzian metric is introduced on $1$-forms, we can distinguish three different type of co-vectors:  $\omega_p \in \T_p^*\M$ is {\bf spacelike}, {\bf timelike}, {\bf null} and {\bf causal} if, respectively, $\sharp \omega_p \in \T_p\M$
is spacelike, timelike or null.    With the definition, we can define the open {\bf  lightcone of $1$-forms} at $p\in \M$  analogously to the case of vectors
$$V^{g^\sharp}_p:= \{\omega_p \in \T^*_p\M \:|\: g^\sharp(\omega_p,\omega_p) <0\}\:.$$
Analogously, if $g\in \mathcal{T}_\M$, the {\bf future} and {\bf past} {\bf  lightcones of $1$-forms} at $p\in \M$ are defined as 
$$V^{g^\sharp\pm}_p:= \{\omega_p \in \T^*_p\M \:|\: \sharp\omega_p \in V^{g\pm}_p \}\:.$$

Let us finally recall that, embedded codimension-$1$ submanifold $\Sigma \subset \M$ of a Lorentzian manifold $(\M,g)$, also called {\bf hypersurfaces}, are classified according to their normal covector $n$: They are {\bf spacelike}, {\bf timelike}, {\bf null} if respectively $n$ is  timelike, spacelike, null everywhere in $\Sigma$.
Notice that an embedded $n-1$ submanifold $\Sigma \subset \M$ is spacelike if and only if its tangent vectors are spacelike in $(\M,g)$. The restriction of $g$ to the tangent vectors to a spacelike hypersurface $\Sigma$ defines a Riemannian metric on it.

\subsubsection{Spacetimes and causality}

\begin{definition}\label{def:spacetime}
A  {\bf spacetime} is a
$(n + 1)$-dimensional   ($n\geq 1$), connected,
	time-oriented, smooth Lorentzian manifold $(\M,g)$
\end{definition}

\begin{remark} Sometimes it is also assumed that $\M$ is orientable and  oriented, but we do not adopt this hypothesis here. However, when we write that {\em $(\M,g)$ is a spacetime} we also mean that a {\em time-orientation} of $(\M,g)$ as Lorentzian manifold has been chosen. In this case, with a little misuse of language, we speak of the {\em time-orientation of the metric} $g$. \end{remark}

 Let now $A \subset \M$ for a spacetime $(\M,g)$. The {\bf causal sets} 
  $J_\pm(A)$ and the {\bf chronological sets} $I_\pm(A)$ are defined according to \cite{Ba-lect}: $J_\pm(A)$ is made of the points  
of $A$ itself  and all $p\in M$  such that there is a smooth future-directed/past-directed  causal  curve 
$\gamma : [a,b] \to \M$
with   $\gamma(a) \in A$ and  $\gamma(b)=p$.
{\em Notice that $J_\pm(A) \supset A$ by definition}, while
$I_\pm(A)$ is made of the points  $p\in M$ such that there is a smooth future-directed/past-directed  timelike curve 
$\gamma : [a,b] \to \M$
with   $\gamma(a) \in A$ and  $\gamma(b)=p$.
As usual we define $J(A) := J_+(A) \cup J_-(A)$.

Let us recall that, on a spacetime $(\M,g)$, a smooth causal curve $\gamma : I \to \M$ with  $I \subset \mathbb{R}$  open interval is said to be {\bf future  inextendible}   \cite{Oneill} if there is no  {\em continuous} curve $\gamma': J \to \M$, defined on  an open interval  $J \subset \mathbb{R}$,  such that $\sup J > \sup I$ and $\gamma'|_I =\gamma$.   A {\bf past inextendible} causal curve is defined analogously. A causal curve is said to be {\bf inextendible} if it is both past and future inextendible.

We eventually  define the {\bf future Cauchy development}  $D_+(A)$  of $A$ to be the set  of points $p\in  \M$
such that  every  past inextendible future-directed smooth causal curve  passing through $p$ meets $A$ in the past.
Similarly, the  {\bf past Cauchy development}  $D_-(A)$ is the set  of points $p\in  \M$
such that  every  future inextendible future-directed  smooth causal curve  passing through $p$ meets $A$ in the future.

\medskip

On a generic Lorentzian manifold, the Cauchy problem for a differential operator is in general
ill-posed: This can be a consequence of the presence of closed timelike curves or the presence of
naked singularities. Therefore, it is convenient to restrict ourselves to the class of {\em globally hyperbolic spacetimes}.

\begin{definition}
	\label{def:globally hyperbolic}
	A  {\bf globally hyperbolic spacetime}  is a spacetime  $(\M,g)$ such that
	\begin{itemize}
		\item[(i)] there are no closed causal curves;
		\item[(ii)] for all points $p,q\in\M$, $J_+(p)\cap J_-(q)$ is compact.
	\end{itemize}
\end{definition}

\begin{notation}
If $\M$ is a smooth connected $(n+1)$-manifold,  $\mathcal{GH}_\M\subset \mathcal{T}_\M$ denotes the class of Lorentzian metrics $g$ such that $(\M,g)$ is globally hyperbolic for a time-orientation. Any  $g \in \mathcal{GH}_\M$ is called {\bf globally hyperbolic} metric on $\M$.
\end{notation}

In his seminal paper \cite{Geroch}, Geroch established the equivalence for a Lorentzian manifold being globally hyperbolic and the existence of a \emph{Cauchy hypersurface}.

\begin{definition}
A subset $\Sigma\subset\M$ of a spacetime $(\M,g)$ is called {\bf Cauchy hypersurface} if it  intersects exactly once any inextendible future-directed smooth timelike curve.  
\end{definition}

In particular, a Cauchy hypersurface is {\bf achronal}: it intersects at most once every future-directed smooth timelike curve.

\begin{theorem}[\protect{\cite[Theorem 11]{Geroch}}]\label{thm:Geroch} A spacetime $(\M,g)$ is globally hyperbolic if and only if it contains a Cauchy hypersurface.
\end{theorem}

It turns out that      Cauchy hypersurfaces of $(\M,g)$ are closed co-dimension $1$ topological submanifolds of $\M$  homeomorphic one to each other.
As a byproduct of Geroch's theorem, it follows that
 the  globally hyperbolic manifold $(\M,g)$ admits a continuous foliation in Cauchy hypersurfaces $\Sigma$, namely
$\M$ is homeomorphic to $\RR \times \Sigma$. 
 The proof of these facts was carried out by finding a {\em Cauchy time function}, i.e., a continuous function $t:\M\rightarrow\mathbb{R}$ which is strictly increasing on any future-directed timelike curve and such that its level sets $t^{-1}(t_0)$, $t_0\in \mathbb{R}$, are Cauchy hypersurfaces homeomorphic to $\Sigma$.
 Geroch's splitting appears at a topological level, and the possibility to smooth them
remained an open folk questions for many years. Only recently, in~\cite{BeSa} Bernal and S\'anchez ``smoothened'' the result of Geroch by introducing the notion of \emph{Cauchy temporal function}.

\begin{theorem}[\protect{\cite[Theorems 1.1 and 1.2]{BeSa}},  \protect{\cite[Theorem 1.2]{BeSa2}}, ]\label{thm: Sanchez}
For every  globally hyperbolic spacetime  $(\M,g)$  there is an isometry $\psi: \M \to \RR \times \Sigma$, where   the latter spacetime  is equipped with the smooth Lorentzian metric 
\begin{equation}  - \beta^2 d \tau\otimes d\tau \oplus h_\tau \:,  \label{GHmetric} \end{equation}
	and the time-orientation induced from $(M,g)$ through $\psi$.
Above $\tau$ is the canonical projection  $$\RR\times \Sigma \ni (t,p) \mapsto t\in  \RR$$ and the following facts are valid:
\begin{itemize}
\item[(i)]  $\nabla \tau := \sharp d\tau$ is  past-directed timelike, 

\item[(ii)]	 $\beta : 	\RR \times \Sigma \to (0,+\infty)$ (called {\bf lapse function}) is a smooth function,

\item[(iii)] $h_t$(called {\bf spatial metric})  is a smooth Riemannian metric on each leaf $\{t\} \times \Sigma$, $t\in \RR$,

\item[(iv)]  every embedded co-dimension-$1$ submanifold  $\{t_0\} \times \Sigma = \tau^{-1}(t_0)$  is a  spacelike (smooth) Cauchy hypersurface.

\end{itemize}
Finally, if $S \subset \M$ is a spacelike Cauchy hypersurface of $(\M,g)$, then we can define an isometry  $\psi : \M \to \RR \times S$, and $\tau$, $\beta$, $h$ as above in order  that  $S = \psi^{-1}(\{0\} \times S)$.
\end{theorem}

The characterization given by Bernal and S\'anchez permits us to give some relevant definitions.

\begin{definition}\label{defTT} Given a spacetime $(\M,g)$,  a smooth surjective  function  $t : \M \to \RR$ 
with $dt$ past-directed timelike
is
\begin{itemize}
\item[(a)] a {\bf Cauchy temporal function}  if 
\begin{itemize}
\item[(i)] $(\M,g)$ is isometric, through some isometry $\psi: \M \to \RR \times \Sigma$, to  a spacetime $(\RR \times \Sigma, h)$ with the time-orientation induced from $(\M,g)$, 
\item[(ii)] $t= \tau \circ \psi$  (where $\tau :  \RR \times \Sigma \ni (t,p) \mapsto t \in \RR$),  
\item[(iii)] $h$  has the form  (\ref{GHmetric}) as in Theorem \ref{thm: Sanchez} satisfying (i)-(iv);
\end{itemize}
\item[(b)] a {\bf smooth Cauchy time function}    if 
\begin{itemize}
\item[(i)]  $(\M,g)$ is isometric, through some isometry $\psi: \M \to \RR \times \Sigma$, to a spacetime  $(\RR \times \Sigma, h)$ with the time-orientation induced from $(\M,g)$,  
\item[(ii)]  $t= \tau \circ \psi$ (where $\tau :  \RR \times \Sigma \ni (t,p) \mapsto t \in \RR$),  
\item[(iii)]   every $\Sigma_{t_0} := t^{-1}(t_0) = \psi^{-1}(\{t_0\} \times \Sigma) $ is a spacelike Cauchy hypersurface of $(\M,g)$ for $t_0\in \RR$.
\end{itemize}
\end{itemize}
\end{definition} 

\begin{remarks}
\begin{itemize}

\noindent  \item[(1)]   An intrinsic way to write (\ref{GHmetric}) for a Cauchy temporal function $t$  without making use to the splitting diffeomorphism $\psi$ is, for $p\in \Sigma_s= t^{-1}(p)$
$$g_{p}(X,Y) = \frac{dt \otimes d t(X,Y)}{g^\sharp(dt,dt)}  + h_s(\pi_{t,g}X , \pi_{t,g} Y)\:, \quad X,Y\in \T_p\M= L(\sharp_g dt)\oplus \Sigma_s$$
where
$$\T_p\M \ni  X \mapsto  \pi_{t,g} X:= X - \frac{   \langle dt, X\rangle   \:\sharp_g dt}{g^\sharp (dt,dt)}\in \T_p\Sigma_s$$
defines the  orthogonal projector onto $\T_p\Sigma_s$ associated to $t$ and $g$, using $\sharp_g dt$ as normal (contravariant) vector to $\Sigma_s$.

 \item[(2)] If an either smooth time or  temporal Cauchy function $t$ exists for $(\M,g)$, the level sets $\Sigma_{t_0}:= t^{-1}(t_0)$ are smooth spacelike  Cauchy  surfaces diffeomorphic to each other and  $(\M,g)$  is  globally hyperbolic.   Theorem \ref{thm: Sanchez} proves that temporal Cauchy functions -- thus also smooth time Cauchy functions -- exist for every globally hyperbolic spacetime. Furthermore, every smooth spacelike Cauchy hypersurface can be embedded in the foliation induced by a suitable temporal Cauchy function.  
\item[(3)]  A Cauchy temporal function is always a Cauchy time function, but even a smooth
time function may not be a temporal one. 

\item[(4)] A Cauchy hypersurface
may meet a causal curve  in more than a point (say, a segment), but
this is not the case  for the spacelike Cauchy hypersurfaces since they are 
{\bf acausal}:   they intersect {\em at most once} every future-directed smooth causal curve,  as easily arises from Theorem \ref{thm: Sanchez}.
\end{itemize}
\end{remarks}

 We shall now give
some notable examples of globally hyperbolic spacetimes to acquaint the reader with some concrete cases.
\begin{examples}\label{ex:globhyp} We shall list a few globally hyperbolic spacetimes which appear commonly in general relativity and quantum field
theory over curved backgrounds. As one can infer per direct inspection, they all fulfill
Theorem~\ref{thm: Sanchez}:
\begin{itemize}
	\item the prototype example is Minkowski spacetime which isometric to    $\mathbb{R}^{n+1}$ 
with Cartesian coordinates $(t,x^1,\ldots, x^n)$ and equipped with the Minkowski 
 metric 
	$$-dt\otimes dt +\sum_{i=1}^{n}{dx^i}\otimes dx^i\:;$$
	\item de Sitter spacetime, that is the maximally symmetric solution of Einstein's equations with a
positive cosmological constant $\Lambda$. As a manifold it is topologically $\mathbb R \times\mathbb S^3$ and the metric reads:
$$g = -dt^2 + \frac{ 3}{\Lambda}\cosh^2 \left(\sqrt{\frac{\Lambda}{3} t}\right) [ d\chi^2 + \sin^2 \chi (d\theta^2 + \sin^2 \theta d\varphi^2)] $$
where $t \in \RR$ while $(\chi,\theta,\varphi)$ are the standard coordinates on $\mathbb{S}^3$;
\item the Friedmann-Robertson-Walker spacetime, i.e., an isotropic and homogeneous manifold which is
topologically $\RR\times \Sigma$ and
$$g = -dt^2 + a(t) \left[\frac{dr^2}{1-kr^2}  + r^2 (d\theta^2 + \sin^2 \theta d\varphi^2)\right] $$
where $k$ can be either 0 or $\pm 1$ and function $a(t)$ is smooth and positive valued;
	\item The external Schwarzschild spacetime, i.e., a stationary spherically symmetric solution of vacuum Einstein's equations which is topologically $\mathbb{R}^2\times \mathbb{S}^2$ with metric
$$g = - \left(1- \frac{2M}{r}\right) dt^2 + \left(1- \frac{2M}{r}\right)^{-1} dr^2  + r^2 (d\theta^2 + \sin^2 \theta d\varphi^2) \,.$$
Here $M>0$ is interpreted as the mass of the spherically symmetric source (a blackhole, a star,...) and
the domain of definition of the coordinates is $t \in \RR$, $r \in (2M, +\infty)$ and $(\theta, \varphi) \in \mathbb{S}^2$;
\item Finally, given any $n$-dimensional complete Riemannian manifold $(\Sigma,h)$, an open interval $I\subseteq\mathbb{R}$ and a smooth  function $f:I\rightarrow (0,+\infty)$, the Lorentzian \textit{warped product} defined topologically by $I\times\Sigma$ with metric $g=-dt^2+f(t)h$ is a globally hyperbolic spacetime.
\end{itemize}
\end{examples}

\subsection{Convex interpolation  of Lorentzian metrics}\label{subsec:convex}
Due to motivations that will be clear later in the paper  and related to the construction of M\o ller operators, we are now interested in the structure of the set $\mathcal{M}_\M$ of Lorentzian metrics on a given manifold $\M$. In particular, we are interested in the following problem:
\begin{center} \em Are there some natural operations which can be used to produce (globally hyperbolic) Lorentzian metrics starting from (globally hyperbolic) Lorentzian metrics?\end{center}
Given two globally hyperbolic metrics $g,g'$, a linear  combination of them is in general not a Lorentzian metric and, when it is, it fails to be globally hyperbolic in general. However, as shown in~\cite[Appendix B]{Romeo}, if $g$ and $g'$ coincide outside a compact set, then there exists a sequence of 5 globally hyperbolic metrics, such that for each neighbouring pair all pointwise convex combinations are
globally hyperbolic metrics.  Therefore, this section aims to provide sufficient conditions for  some kind of linear combination of globally hyperbolic metrics to be a globally hyperbolic Lorentzian metric. We shall see that convex combinations are an interesting case of study under suitable conditions.  We point out the recent work \cite{Sanchez22} where, in addition to several  related issues,  the convex structure of the space of globally hyperbolic metrics on a given manifold is addressed with a number of results.

\subsubsection{A preorder relation of Lorentzian metrics}
\label{sec:partial ordering}

\begin{definition}\label{def:part ord} 
Let $g,g' \in \mathcal{M}_\M$ and denote
$$g  \preceq g'\quad \mbox{ \quad  iff \quad $V^g_p \subset V^{g'}_p$  for all $p\in \M$.}$$
We say that  $g,g' \in \mathcal{M}_\M$ are {\bf $\preceq$-comparable} if either $g \preceq g'$ or $g'\preceq g$ (see e.g. Figure~\ref{fig:comp}).
\end{definition}
\begin{figure}[h!]
\center
\includegraphics[scale=0.34]{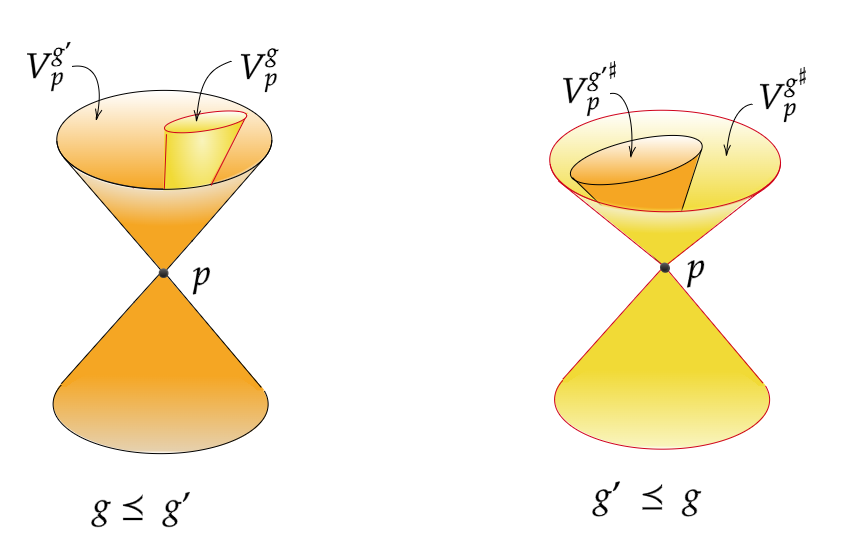}
\caption{Lorentzian metrics $\preceq$-comparable}\label{fig:comp}
\end{figure}
\begin{remarks}
\noindent
\begin{itemize}
\item[(1)] Let us remark that the definition above can be generalized by considering the so-called {\em causal diffeomorphisms}, namely a time-orientation preserving diffeomorphism $\varphi:\M\to\N$ such that the open light cone $V_p^g$ of $g$ is included in the open light cone $V_p^{\varphi g'}$ of $\varphi^* g'$ for every $p\in\M$. For further details and properties we refer to~\cite{causal-diff1,causal-diff2,causal-diff3}.
\item[(2)]
The preorder relation introduce in Definition~\ref{def:part ord} has a corresponding for the associated metrics in the cotangent space: If $g,g' \in \mathcal{M}_\M$,   $$g^\sharp  \preceq g'^\sharp\quad \mbox{ \quad  iff \quad $V^{g^\sharp}_p \subset V^{g'^\sharp}_p$  for all $p\in \M$.}$$
\end{itemize}
\end{remarks}

We observe that if $g\preceq g'$ for $g,g' \in \mathcal{T}_\M$ and the two metrics share the same time-orientation  -- i.e.,  there is a continuous vector field on $\M$ which is timelike for both metrics and defines the same time-orientation for both of them --  then 
$V^{g+}_p \subset V^{g'+}_p$ 
and $V^{g-}_p \subset V^{g'-}_p$ 
for every $p \in \M$. Similar inclusions hold when considering the closures of the considered half cones.
As a consequence, we have both inclusions 
 with obvious notations
\begin{equation*}
I_\pm^{g}(A) \subset I_\pm^{g'}(A) \:, \quad J_\pm^{g}(A)  \subset J_\pm^{g'}(A) \quad \mbox{for every $A\subset \M$.}\label{eq:inclusJ}
\end{equation*}

The  relation  $\preceq$ in $ \mathcal{M}_\M$  has several consequences whose most elementary ones are established in the following proposition.  We remind the reader that a closed set $A\subset \M$, with $(\M,g)$ time-oriented, is {\bf past compact} if $J_-(p) \cap A$ is compact or empty for every $p\in \M$. The definition of {\bf future compact} is analogous, just replacing $J_-$ for $J_+$.
\medskip

Then next lemma is just a routine computation, so we leave the proof to the reader.
\begin{lemma}\label{prop:metrics0} Let $\M$ be a smooth $(n+1)$-dimensional manifold and $g,g'\in \mathcal{M}_\M$. The following facts are valid for the preordering relation $\preceq$ in $\mathcal{M}_\M$. 
\begin{itemize}
\item[(1)] For $p\in \M$ and
$v\in \T_p\M$, if
$g\preceq g'$ then  
\begin{itemize}
\item[(i)]  $g(v,v)=0$ implies  $g'(v,v)\leq 0$.
\item[(ii)] $g'(v,v)> 0$  implies $g(v,v)> 0$.
\item[(iii)] $g'(v,v)= 0$  implies $g(v,v)\geq  0$.
\end{itemize}
\item[(2)]  If $g \preceq g'$ with $g \in \mathcal{T}_\M$  and $g' \in \mathcal{GH}_\M$, then $(\M,g)$ is globally hyperbolic as well  when, {\it e.g.}, equipped with the same orientation and time-orientation of $(\M,g')$ and 
\begin{itemize}
\item[(i)] a spacelike Cauchy hypersurface for $(\M,g')$ is also a spacelike Cauchy hypersurface for $(\M,g)$;
\item[(ii)] a smooth Cauchy time function for $(\M,g')$ is also a smooth Cauchy time function for $(\M,g)$;
\item[(iii)]  a closed set $A \subset \M$ is past/future compact in $(\M,g)$ if it is respectively past/future compact in $(\M,g')$.
\end{itemize}
\item[(3)]  $g\preceq g'$ if and only if $g'^\sharp \preceq g^\sharp$. 

\item[(4)] If $g,g' \in \mathcal{T}_\M$, $g\preceq g'$ and $p\in \M$, then 
$V^{g+}_p \subset V^{g'+}_p$  if and only if $V^{g'^\sharp+}_p \subset V^{g^\sharp+}_p$ (see e.g. Figure~\ref{fig:Inclusion}).
\end{itemize}
\end{lemma}
\begin{figure}[h!]
\vspace{-8mm}
\centering
\includegraphics[scale=0.34]{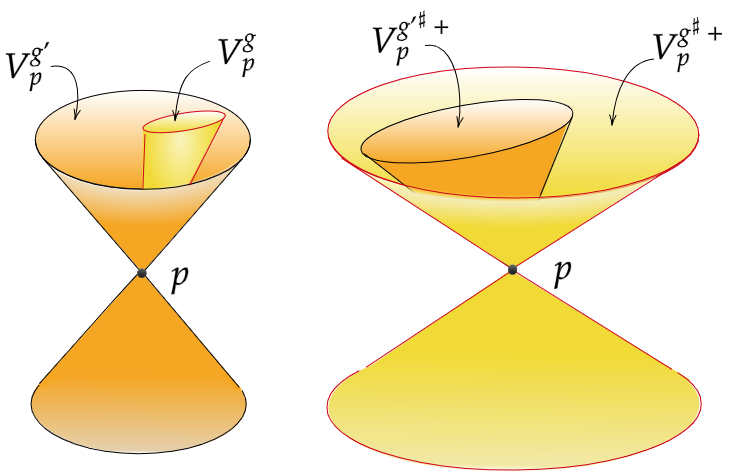}
\caption{Inclusion-of-cones relations}\label{fig:Inclusion}
\end{figure}

Using the lemma above, we can immediately conclude the following.
\begin{proposition}\label{remchichi0chi1} If $g\in \mathcal{M}_\M$  and $\mu : \M \to (0,+\infty)$ is smooth,  then
 \begin{itemize} 
\item[(a)] $\mu  g$ and $\mu^{-1}g$ 
  are Lorentzian,
\item[(b)]   $\mu g \preceq g \preceq \mu g$,
 \item[(c)]  $\mu^{-1}g \preceq g \preceq \mu^{-1}g$.
\item[(d)] $\mu g$ and $\mu^{-1}g$ are  globally hyperbolic if $g$ is and the spacelike Cauchy hypersurfaces of $g$ are also 
spacelike Cauchy hypersurfaces for  $\mu g$ and $\mu^{-1}g$.
\end{itemize}
\end{proposition}

\subsubsection{Properties of convex combinations of Lorentzian metrics}
A more interesting set of properties arises when focusing on smooth {\em convex combinations} of Lorentzian metrics.  This is the first main result of this section.

\begin{theorem}\label{thm:metrics1} Let $\M$ be a smooth $(n+1)$-dimensional manifold, $g,g'\in \mathcal{M}_\M$, and 
consider  a  smooth function
 $\chi : \M \to [0,1]$.
  If $g\preceq g'$, the following facts are valid
\begin{itemize}
\item[(1)] $(1-\chi) g +  \chi g' $ is a metric of Lorentzian type;

\item[(2)]  $g \preceq (1-\chi) g +  \chi g'   \preceq g'$;
\item[(3)]   if $g_{\chi}^\sharp :=  (1-\chi)  g^\sharp + \chi g'^\sharp$, then $g_{\chi}^\sharp:= (g_{\chi})^\sharp$ for a (unique)  metric $g_{\chi}$ of Lorentzian type;

\item[(4)] 
$g\preceq g_{\chi} \preceq g'$;
\item[(5)]  If $g'$ is globally hyperbolic and $g$ time-orientable, then $(1-\chi) g +  \chi g' $  and $g_{\chi}$ are  globally hyperbolic.

\end{itemize} 
\end{theorem}

\begin{proof}  
(1) It is sufficient to prove the thesis point by point.
Let $q,q'$ be quadratic forms in a real $n+1$ dimensional linear space  $V$ of signature $(-, +, \ldots, +)$ such that $q'(x)\le 0$ implies $q(x)\le 0$. We prove that the strict convex combination $q''= cq + (1-c) q'$ for $c\in (0,1)$ has signature $(-, +, \ldots, +)$. Indeed, there is a   $1$-dimensional subspace  $L$ on which $q'(x)<0$ if $x\neq 0$. So $q(x)\le 0$ on $L$ and hence $q''(x)< 0$ on $L$ for $x\neq 0$. There is also  a $n$-dimensional subspace $H$  on which $q(x)>0$ if $x\neq 0$. Then $q'(x)>0$ on $H$ for $x\neq 0$ and hence $q''(x)>0$ on $H$ if $x\neq 0$.   By construction,  $L \cap  H= \{0\}$ necessarily, so that $V= L\oplus H$. 
The bilinear form $Q'': V \times V \to \mathbb{R}$ associated to  $q''$, in a basis of $V$ made of 
$0\neq e_0 \in L$ and $\{e_k\}_{k=1,\ldots,n} \in H$ with $Q''(e_k,e_h)= \delta_{kh}$, is represented by the $(n+1)\times (n+1)$ matrix $\begin{bmatrix}
q''(e_0) & c^t  \\
c & I
\end{bmatrix}$. Since the determinant is $q''(e_0) - c^tc <0$ and  $n$ eigenvalues are $+1$, its  signature is $(-,+,\ldots, +)$.

(2) Suppose that $g(v,v) <0$, then  $g'(v,v) <0$ because $g\preceq g'$ and thus $(1-\chi) g(v,v) + \chi g'(v,v)<0$ because $\chi, 1-\chi \geq 0$. We have obtained that $g \preceq (1-\chi) g + \chi g'$. Let us pass to the remaining inequality. If 
$(1-\chi) g(v,v) +\chi  g'(v,v)<0$ then $g'(v,v)<0$ or $g(v,v)<0$, in this second case also  $g'(v,v)<0$  because $g\preceq g'$. In both cases 
$g'(v,v) <0$. Summing up, $(1-\chi) g + \chi  g' \preceq g'$, concluding the proof of (2).\\
(3) $g^\sharp$ and $g'^\sharp$ are Lorentzian metrics on $\T^*\M$ and $g'^\sharp \preceq g^\sharp$ due to Lemma ~\ref{prop:metrics0}, we can recast the same argument used to establish (1) with trivial re-arrangements, obtaining that $g_\chi^\sharp$ is Lorentzian and $g'^\sharp \preceq g^\sharp_\chi =(1-\chi') g'^\sharp +  \chi' g^\sharp  \preceq g^\sharp$ with
$\chi' := 1-\chi$.
Notice that $g_\chi (v,v) := g_\chi^\sharp(\flat v, \flat v)$ defines a Lorentzian metric as well, since it has the same signature of $h$ by construction, and 
$g^\sharp = h$ trivially (and it is the unique metric with this property since $\flat$ is an isomorphism).\\
(4) It immediately arises from Lemma~\ref{prop:metrics0} by using $g'^\sharp \preceq g^\sharp_\chi =(1-\chi') g'^\sharp +  \chi' g^\sharp  \preceq g^\sharp$ with
$\chi' := 1-\chi$.\\
(5) A smooth timelike vector field of $(\M,g)$ is also timelike  for $(1-\chi) g+ \chi g'$ and $g_\chi$ for (2) and (4) respectively.  Hence these metrics are time-orientable and the thesis follows from Lemma~\ref{prop:metrics0} point (2).
\end{proof}

\subsection{Paracausal deformation of Lorentzian metrics}\label{subsec:paracausal}

The aim of this section is to provide a new definition that shall encode  the idea to deform a Lorentzian metric equipped with a time-orientation  into another Lorentzian metric with a corresponding time-orientation, taking advantage of  a procedure consisting of a finite number of steps. At each step, the light cones of the final  metric $g_k$ is related to the initial  one $g_{k-1}$ through an inclusion relation, either $g_{k-1} \preceq g_k$ or $g_k \preceq g_{k-1}$ preserving the time-orientation at each step, i.e., the future cone of $g_k$, respectively, includes or is included in the future cone of $g_{k-1}$.

\subsubsection{Paracausal relation}

\begin{definition} \label{def:paracausal def} Consider a pair of globally hyperbolic spacetimes  on the same manifold $\M$ with 
corresponding metrics $g,g' \in \mathcal{GH}_\M$ and corresponding time-orientations.
We say that  $g$ is  {\bf paracausally related} to $g$ -- and we denote it by $g \simeq g'$ -- or equivalently  $g'$ is a {\bf paracausal deformation} of $g$,  if there is a  finite sequence 
$g_0= g, g_1,\ldots, g_N = g'\in \mathcal{GH}_\M$ with corresponding time-orientations, 
such that for all $p\in \M$ either $$V^{g_k +}_p \subset V^{g_{k+1} +}_p  \qquad \text{ or } \qquad V^{g_{k+1} +}_p \subset V^{g_{k} +}_p$$ where the choice may depend on $k=0,\ldots, N-1$.
\begin{figure}[h!]
\center
\includegraphics[scale=0.34]{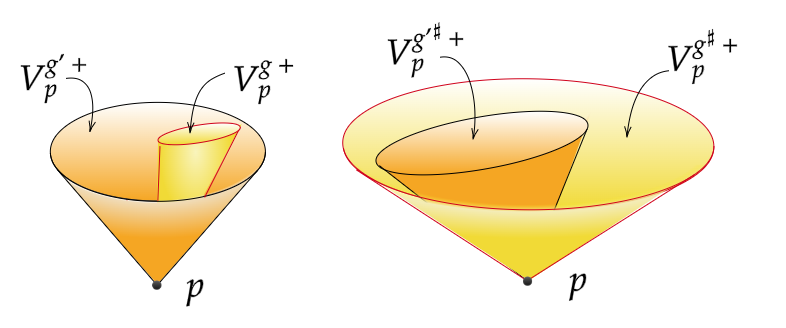}
\end{figure}
\end{definition}

\begin{remarks}
\begin{itemize}
\noindent \item[(1)]  Let us remark that our notion of paracausally deformation implies in particular that $g_k$ and $g_{k+1}$ are always  $\preceq$-comparable.
\item[(2)] Evidently, to be paracausally related is an {\em equivalence relation}  in $\mathcal{GH}_\M$.
\item[(3)] We stress that  paracausal deformations explicitly consider the  time-orientations of the used sequences of globally hyperbolic spacetimes. So, even if we say that ``metrics are paracausally related'',  the relation actually involves {\em the metrics equipped with corresponding time-orientations}.
\item[(3)] We shall show below a characterization of the paracausal relationship which seems more natural from a geometric and physical viewpoint. However, the definition above {\em as it stands} is much more directly suitable for the applications to 
M\o ller operators we shall introduce in the second part of this work. 
\end{itemize}
\end{remarks}

\begin{examples}\label{examples}
\begin{itemize}
\noindent \item[(1)] There two elementary  cases of paracausally related (globally hyperbolic) metrics $g_0$, $g_1$ on $\M$ which are not directly $\preceq$-comparable:
\begin{enumerate}
\item There is a globally hyperbolic metric $g$ on $\M$ such that, simultaneously  $g \preceq g_0$ and  $g \preceq g_1$ and
the future  lightcones are correspondingly included.
\item There is a globally hyperbolic metric $g$ on $\M$ such that, simultaneously  $g_0 \preceq g$ and  $g_1\preceq g$
and
the future  lightcones are correspondingly included.
\end{enumerate}
In both cases, 
 the existence of sequence $g_0,g,g_1$ proves that $g_0 \simeq g_1$.
\item[(2)] Let us give an elementary concrete example of paracausally related metrics. Consider the following smooth manifold $\RR^n$ endowed with the Minkowski metrics
$$\eta_0=-dt\otimes dt+ \sum_{i=1}^n dx^i \otimes dx^i \qquad \eta_1 = -d\tau \otimes d\tau  + \sum_{i=1}^n dy^i\otimes dy^i$$
where $(t,x_1,\dots,x_n)$ and $(\tau, y_1,\dots,y_n)$ are two different systems of Cartesian  coordinates on $\RR^{n+1}$. Here $t$ and $\tau$ are  Cauchy temporal functions associated to the respective Lorentzian metric and defining the time-orientation of the two metrics: $dt$ and $d\tau$ are assumed  to be past directed for the respective metric.
More precisely, we assume that the two coordinate systems are related by means of a physically non-trivial  permutation that  interchanges space and time, as in Figure~\ref{fig:Mink}, $\tau = x_1$, $y_1= t$, and $y_k=x_k$ for $k>1$.
\begin{figure}[h!]
\vspace{-12mm}
\centering
\includegraphics[scale=0.34]{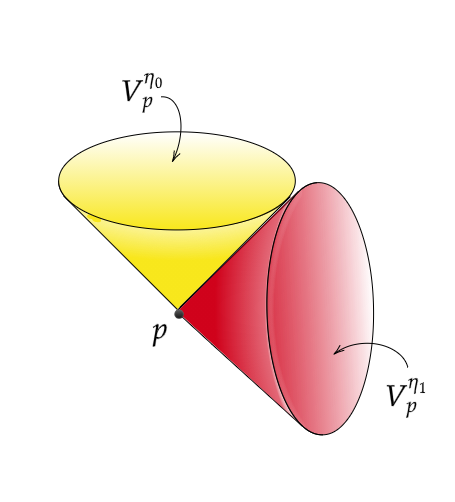}
\caption{Future  light cones of different Minkowski metrics on $\RR^{n+1}$.}\label{fig:Mink}
\end{figure}
 It is not difficult to see that even if $\eta_0\neq \eta_1$ evidently, we have  $\eta_0\simeq  \eta_1$: they  are paracausally related by the sequence of metrics $\eta_0,\overline g_1,\overline{g}_2,\eta_1$ whose future cones are given as in Figure~\ref{fig:paraMink}.
 \begin{figure}[h!]
 \vspace{-8mm}
\centering
\includegraphics[scale=0.34]{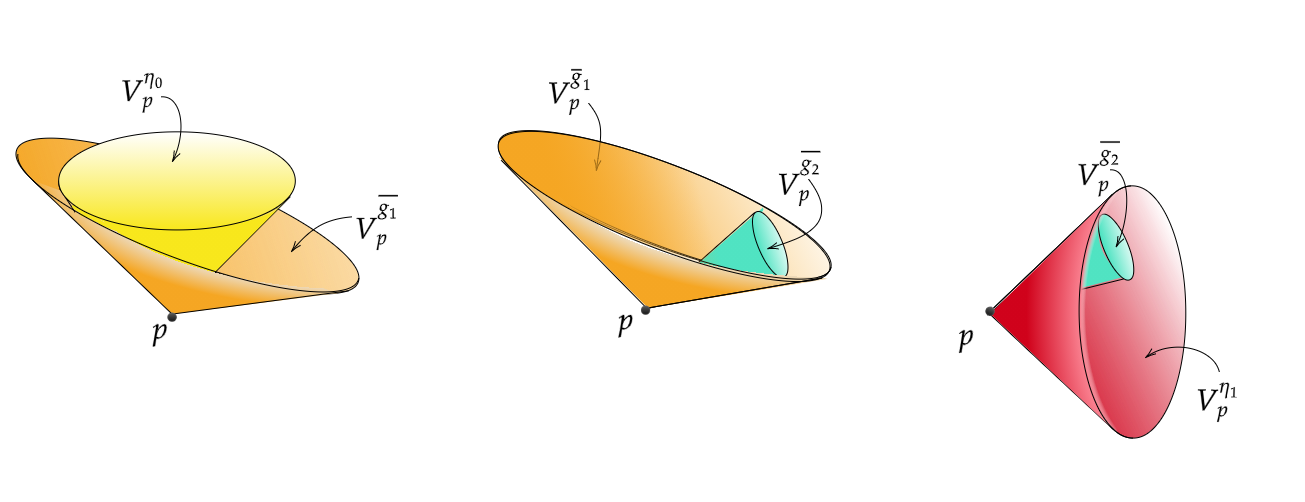}
\caption{Auxiliary future light cones to prove  $\eta_0 \simeq \eta_1$}\label{fig:paraMink}
\end{figure}
It is evident that by  further implementing the procedure, it is possible to reverse the time-orientation of  $(M,\eta_0)$ through a sequence of paracausal deformations leaving the final metric identical to the initial one.
\item[(3)] We pass to present a case where a pair of globally hyperbolic metrics are {\em not} paracausally related.
Consider the $2D$ Minkowski cylinder $\M$ obtained by identifying $x$ and $x+L$ in $\RR^2$ with coordinates $x,y$.
The first globally hyperbolic spacetime is $(\M, \eta_1)$ where $\eta_1 = -dy\otimes dy+ dx\otimes dx$, taking the identification into account, and with time-orientation defined by assuming that $\partial_y$ is future-directed. The second  globally hyperbolic spacetime is $(\M, \eta_2)$ where again $\eta_2 = -dy\otimes dy + dx\otimes dx$, taking the identification into account, but  with the opposite time-orientation, i.e.,  defined by $-\partial_y$.  See also Figure~\ref{fig:2-cil}.
\begin{figure}[h!]
\centering
\includegraphics[scale=0.3]{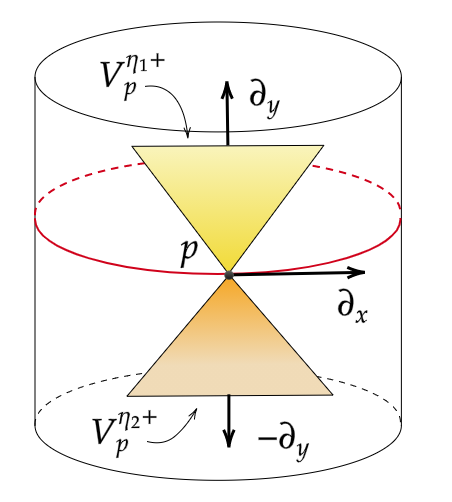}
\caption{2-D Minkowski cylinder.}\label{fig:2-cil}
\end{figure}

These two metrics are not paracausally related. Any attempt to use the procedure as in the previous example to rotate the former into the latter faces the insurmountable obstruction that one of the auxiliary metrics would have Cauchy hypersurfaces given by the $x$-constant lines. This Lorentzian manifold is not globally hyperbolic because it admits closed temporal curves as in Figure~\ref{fig:not_paradef}.
 \begin{figure}[h!]
\centering
\includegraphics[scale=0.3]{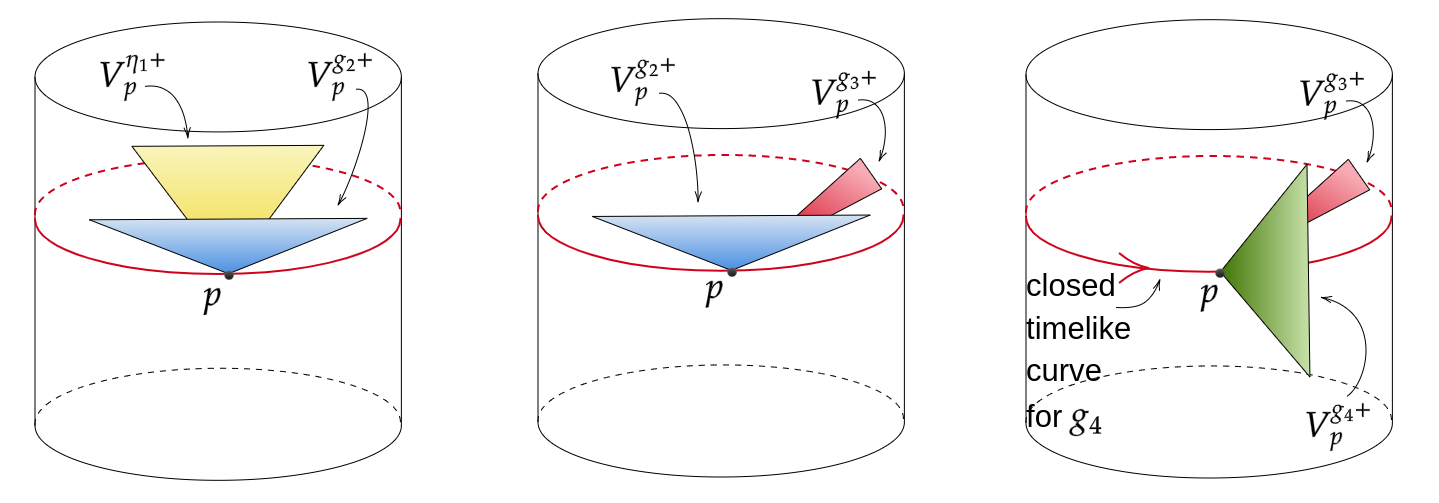}
\caption{Sequence of metrics where $g_4$ is not globally hyperbolic}\label{fig:not_paradef}
\end{figure}

Notice that this obstruction does not take place without the identification $x\equiv x+L$. 
\end{itemize}
\end{examples}

\subsubsection{Characterization of paracausal deformation in terms of future cones}

There is a  natural situation where two globally hyperbolic metrics $g$ and $g'$ on $\M$ are paracausally related.   
The generalization of the following result leads to a natural characterization of the paracausal relationship.

\begin{proposition}\label{teoPC} Let  $(\M,g)$ and $(\M,g')$ be globally hyperbolic spacetimes on the same manifold $\M$. 
If $V_x^{g+} \cap V_x^{g'+} \neq \emptyset $ for every $x\in \M$, then 
the metrics  $g$ and $g'$
are paracausally related.
\end{proposition}

\begin{proof} To prove the assertion it is sufficient to prove the existence of a Lorentzian metric $h \in \mathcal{T}_\M$ such that $h \preceq g$ and $h \preceq g'$. In this case, $h$ would be globally hyperbolic according to (2) in Lemma \ref{prop:metrics0} and  the same argument as in (1) Examples \ref{examples} would prove the thesis.

Let us start by proving that  a smooth  vector field $X$ on $\M$ exists such that $X_p \in V_p^{g+} \cap V_p^{g'+}$ for all $p\in \M$. Let us define the smooth functions  $$G: \T\M \ni (p,v) \mapsto  g_p(v,v)\in \RR\:, \quad G_Y: \T\M\ni (p,v) \mapsto g_p(v,Y)\in \RR\:,$$ where $Y$ is a smooth timelike future oriented vector field for $g$.
By construction   (with obvious notation)  $\cup_{p\in \M} V_p^{g+} = G^{-1}(-\infty,0) \cap G_Y^{-1}(-\infty,0) \subset \T\M$ is an open set. 
With the same argument,
we have that also $\cup_{p\in \M} V_p^{g'+} \subset \T\M$ is open. Finally, $\cup_{p\in \M} V_p^{g+}  \cap \cup_{p\in \M} V_p^{g'+} = \cup_{p\in \M} V_p^{g+}  \cap V_p^{g'+}$ is therefore open, non-empty by hypothesis,  and projects onto the whole $\M$ by construction.

As a consequence,
 given a local trivialization patch $\T U$ around $p\in U$, where $(U, \psi)$ is a local chart on $\M$ (with $\dim(\M)=n+1$),  the set $\left(\cup_{p\in \M} V_p^{g+}  \cap V_p^{g'+}\right) \cap \T U$
is  diffeomorphic to an open  subset $A\subset V \times \RR^{n+1}$ with  $V := \psi(U)\subset \RR^{n+1}$ and 
$\pi_1(A)= V$ ($\pi_1: \RR^{n+1} \times \RR^{n+1} \to \RR^{n+1}$ being the standard projection onto the former factor). 
 Working in coordinates, it is then trivially possible to pick out   a smooth local section $X^{(U)}$ of $\T U$ such that $X^{(U)}_q \in V_q^{g+} \cap V_q^{g'+}$ if $q\in U$. To conclude, consider 
  a partition of the unity $\{\chi_i\}_{i\in I}$ of $\M$ subordinated to a locally finite covering $\{U_i\}_{i\in I}$ of domains of local charts of $\M$ and let  $X_p^{(U_i)} \in V_p^{g+} \cap V_p^{g'+}$  be constructed as above when  $p\in U_i$ for every $i \in I$.  The smooth vector field constructed as a  locally finite convex linear combination $X:= \sum_{i\in I} \chi_i X^{(U_i)}$ satisfies $X_p \in V_p^{g+} \cap V_p^{g'+}$ for every $p\in \M$
because  the cones $V_q^{g+}$,   $V_q^{g'+}$  are convex sets in a vector space and thus their intersection is also convex. $X$ is the vector field we were  searching  for.
 
As the second step we construct a Lorentzian metric $h$,  whose future cones  $V_p^{h_+}$  satisfy  $X_p \in V_p^{h+} \subset V_p^{g+} \cap V_p^{g'+}$ for every $p\in \M$.
Notice that it means $h \in \mathcal{T}_\M$ since $X$ is future directed for $h$ (and also for the two metrics $g$ and $g'$) and thus it defines a time-orientation for $(\M,h)$. Since $h\preceq g, g'$, this would conclude our proof. 

Let us construct $h$ taking advantage of the vector field $X$. Consider $p\in M$ and define a $g$-pseudo orthonormal basis $e_0,\ldots, e_n$ where $e_0 = \frac{X_p}{\sqrt{-g(X_p,X_p)}}$ and the remaining vectors are spacelike. If $v,v'\in \T_p\M$,
$$g(v,v') = -g(e_0,v)g(e_0,v') + \sum_{k=1}^n g(e_k,v)g(e_k,v')\:.$$ If $a\in (0,1)$, the new Lorentzian scalar product in $\T_p\M\ni v,v'$
\begin{equation} g^a(v,v') :=  -ag(e_0,v)g(e_0,v') + \sum_{k=1}^n g(e_k,v)g(e_k,v') = g(v,v') + (a-1) \frac{g(X_p,v)g(X_p,v')}{g(X_p,X_p)} \label{Newmet}\end{equation}
trivially satisfies (the closure being  taken in $\T_p\M \setminus\{0\}$)
$$X_p \in V_p^{g^a+} \subsetneq  \overline{V_p^{g^a+}} \subsetneq V_p^{g+} \quad \mbox{for $a\in (0,1)$.}$$
The strong inclusions are due to the fact that the lightlike boundary of $V_p^{g^a+}$ is made of timelike vectors of $g$ as it arises from the definition of $g^a$.  
Now note that  $\partial V_p^{g^a+} $ becomes more  and more concentrated  around the set $\{\lambda X_p \:|\: \lambda >0 \}$ as $a$ approaches $0$ from above.  (In particular,  the limit and degenerate case $g_p^{a=0}(v,v)=0$ implies $v$ is parallel to  $X_p$.)
Since $X_p \in V_p^{g'+}$ which is also an open convex cone as $ V_p^{g^a+}$,  there must exist $a_p\in (0, 1)$ such that  $V_p^{g^{a_p}+}\subset V_p^{g'+}$.  This property is locally uniform in $a$ as established in the following technical lemma:
\medskip

\textbf{Lemma.}\footnote{As noticed by the referee, a different strategy for proving this lemma would be showing that the function $M \ni p \mapsto a(p) = \sup\{a \in (0, 1): {V_p^{g^a}}^+\subset {V_p^{g'}}^+ \}$ is continuous. In that case, one canalternatively  define $a_U:= \min_{p \in \overline{U} a(p)}$. However the proof of continuity is not technically  easy.} {\em Within the hypotheses of the proposition, if $x \in \M$, there is a coordinate patch with domain  $V \ni x$,  an open set $U\ni x$
with compact closure $\overline{U} \subset V$, and  a constant $a_U\in (0,1)$ such that $V_p^{g^{a_U}+}\subset V_p^{g'+}$ for every $p\in U$.} \medskip

\noindent \textit{Proof.}  If $x\in \M$, there is a coordinate patch with domain $V \ni x$ and coordinates $V \ni p \mapsto \varphi(p) = (x^0(p),\ldots, x^n(p)) \in \RR^{n+1}$ such that $U \ni x$ for some open subset $U \subset V$ such that $\overline{U}$ is compact. We will  henceforth deal with $U$ and the coordinates $(x^0,\ldots, x^n)$ restricted to  thereon. We will also take advantage of the  compact set  $K := \varphi(\overline{U}) \subset \RR^{n+1}$ and identify $\T\overline{U}$ with $K \times \RR^{n+1}$ using the coordinates.  Finally, we will equip both $K$ and $\RR^{n+1}$ (representing $\T_p \M$ at each $p\in \overline{U}\equiv K$) with the standard Euclidean metric of $\RR^{n+1}$ whose norm will be denoted by $||\cdot ||$. 

Let us  start the proof by  proving that the family of cones $V_p^{g'+}$ of $g'$ has a minimal  width $m>0$ when  $p$ ranges in $K$.
We henceforth view the above future-directed timelike vector field $X$ and $g'$ as geometric objects on $K$ using the coordinate system. In particular, 
if $p \in K$, let us indicate by $v_p \in \RR^{n+1}$ the unique  future-directed timelike  vector parallel to $X_p$ (now viewed as a vector in $\RR^{n+1}$) such that $||v_p||=1$.
Consider the set made of future-directed elements of $T\M$
$$C := \{(p,u)  \in K \times \mathbb{S}^{n}  \:|\:  g_p'(u,v_p) \leq 0 \:, g_p'(u,u) =0\}$$
(above $\mathbb{S}^n := \{z\in \RR^{n+1}\:|\: ||z||= 1\}$)
and the continuous function 
$$W : C \ni (p,u) \mapsto ||u - v_p|| \geq 0\:,$$
which computes the width of $\partial{V_p^{g'+}}$  (that is of $V_p^{g'+}$ itself)  around $X_p$ along  the direction $u$ by  using the Euclidean distance induced by $||\cdot||$.
Observe that $C$ is  compact since it is the  intersection of preimages of a pair of  closed sets along two corresponding  continuous maps and it is included in a compact set.  
Since this map is continuous and $C$ is compact,  there exists $$m := \min_C W>0\:.$$ In particular,  $m>0$, otherwise $u = v_q$ for some $(q,u) \in C$ and this is not possible since it would imply $g'(v_q,v_q)= g'(u,u)=0$, but $v_q$ is timelike ($g'_q(v_q,v_q) <0$) since  it does not vanish ($||v_q||=1$) and it is proportional to the timelike vector  $X_q$.

An analogous width-cone  function can be defined for the cones of $g^a$ (including the degenerated case $a=0$) on a set $C'$ which also embodies  the dependence on $a$:
$$C' := \{(a, p,u)  \in [0,1/2]  \times K \times \mathbb{S}^{n}  \:|\: g^a_p(u,v_p) \leq 0 \:, g^a_p(u,u) =0\}\:.$$
We also define the continuous function
$$W' : C' \ni (a, p,u) \mapsto ||u- v_p|| \geq 0\:.$$
Observe that $C'$ is again compact since it is the  intersection of preimages of two closed sets along a pair of corresponding  continuous maps of $(a,p,u)$ and $C'$ is included in a compact set.

We want to prove that there exists  $a^{m} \in [0,1/2]$ such that $W'(a^m,p,u) < m$ for all $(p,u) \in C$. If this were not the case, then for every 
$a_n := 1/n$ there would be  a pair $(p_n,u_n) \in C$ such that $W'(a_n,p_n,u_n) \geq  m$. Since $C'$ is a compact metric space, we could  extract a subsequence
of triples $(a_{n_k},p_{n_k}, u_{n_k}) \to (0, p_\infty, u_\infty) \in [0,1/2] \times C$ for $k\to +\infty$
and some $(p_\infty, u_\infty)\in C$.  By continuity $0=  g_{p_n}^{a_n}(u_n,u_n) \to g^0_{p_\infty} (u_\infty,u_\infty)$
where $||u_\infty||=1$. From (\ref{Newmet}), $g^0_{p_\infty} (u_\infty,u_\infty)=0$ would entail  that $u_\infty$ is parallel to $v_{p_\infty}$ and thus $W'(0, p_\infty, u_{p_\infty})= ||u_\infty- v_{p_\infty}|| =0$. That is in contradiction with the requirement  $W'(a_n,p_n,u_n) \geq  m>0$ for every $n=1,2,\ldots$ in view of the continuity of $W'$.

We have therefore  established that there exists  $a^{m} \in [0,1/2]$ such that $W'(a^m,p,u) < m$ for all $(p,u) \in C$.
From the definition of $W$ and $W'$, we have also obtained that  
$V_p^{g^{a^m+}} \subset V_p^{g'+}$ for all $p\in K$.  It is enough to conclude that 
$V_p^{g^{a_U}+} \subset V_p^{g'+}$ for all $p \in U$ as wanted simply by taking $a_U:=a^m$. This concludes our claim.
\hfill$\Box$\medskip

Let us go on with the main proof. For every $U$ as in the previous lemma,  define  the  constant  function $a(p) = a_U$  for $p \in U$. Since this  can be done in a neighborhood of every point $p\in \M$, using a partition of the unity $\{\chi_i\}_{i\in I}$ subordinated to a  locally finite covering of charts $\{U_i\}_{i\in I}$, we can construct the metric $h$, where now every $a_i := a_{U_i} : U_i  \to (0,1)$ is a constant in $U$ and   thus it is a smooth function therein.
$$h_p(v,v')  =  \sum_i \chi_i(p)   g_p^{a_i(p)}(v,v') = \sum_i \chi_i(p)  \left( g_p(v,v') + (a_i(p)-1) \frac{g_p(X_p,v)g_p(X_p,v')}{g_p(X_p,X_p)}\right) $$ $$=
g_p(v,v') + \left(\sum_i \chi_i (p)a_i(p)-1\right) \frac{g_p(X_p,v)g_p(X_p,v')}{g_p(X_p,X_p)} $$
Since $\sum_i \chi_i(p) a_i(p) \in (0,1)$, this metric  is still  Lorentzian and  of the form (\ref{Newmet}) point by point,  where now
$a(p) = \sum_i \chi_i(p) a_i(p)$.
By construction $X_p \in V_p^{h+} \subset V_p^{g+}$ for every $p\in \M$, just because it happens point by point with the above choice of $a(p)$. In particular, we can endow $h$ with the  time-orientation induced  by $X$ as it happens for
$g$, $g'$ and all local metrics $g^{a_i}$.
 Finally, $V_p^{h+} \subset V_p^{g'+}$ because, if $h_p(v,v)<0$, at least one of the values $g^{a_{i_0}(p)}(v,v)$ appearing in $ \sum_i \chi_i(p)   g_p^{a_i(p)}(v,v) $ must be negative and thus,  if  $v$ is  future-directed,
$v \in V_+^{g^{a_i(p)}+} \subset V_p^{g'+}$. The proof is over because $h$ satisfies all requirements
$X_p \in V_p^{h+} \subset V_p^{g+} \cap V_p^{g'+}$ for every $p\in \M$.
\end{proof}

As an immediate  byproduct, it is easy to see that for any globally hyperbolic metric $g$, there exists a paracausal deformation $g'$ of $g$ which is ultrastatic.
\begin{corollary}\label{cor:ultr para}
Let $(\M,g)$ be a globally hyperbolic spacetime. Then there exists a paracausal deformation $g'$ of $g$ such that $(\M,g')$ is an ultrastatic spacetime.
\end{corollary}
\begin{proof}
Let $t$ be a Cauchy temporal function for the globally hyperbolic spacetime $(\M,g)$ 
so that $\M$ is isometric to $\RR \times \Sigma$ with metric $-dt^2 + h_t$. We indicate by $\partial_t$ the tangent vector to the submanifold $\RR$.
Let $h$ be a complete Riemannian metric on $\Sigma$. Then the ultrastatic metric $g':= -dt^2 + h$ is globally hyperbolic
\cite{Sanchez97}
 and the vector $\partial_t$ is contained in the intersection of $V_p^{g +}$ and $V_p^{g' +}$ for any $p\in\M$.  Proposition  \ref{teoPC} ends the proof.
\end{proof}

The  result established in Proposition~\ref{teoPC} leads to a crucial  characterization of paracausally related metrics, which represent the second main result of this section.

\begin{theorem}\label{thm:char paracausal}
Let $\M$ be a smooth manifold. Two metrics $g,g' \in \mathcal{GH}_\M$  are paracausally related if  and only if there exists a finite  sequence of globally hyperbolic  metrics $g_1=g, g_2\ldots, g_n = g' $ on $\M$ such that all pairs of consecutive metrics $g_k, g_{k+1}$ satisfy $V_x^{g_k+} \cap V_x^{g_{k+1}+} \neq \emptyset $ for every $x\in \M$.
\end{theorem}

\begin{proof} If $g,g'$ are paracausally related,  then a sequence of metrics as in Definition \ref{def:paracausal def}  trivially satisfies the condition in the thesis. If that condition is {\em vice versa} satisfied, then the metrics of each pair $g_k, g_{k+1}$ of the sequence are paracausally related  in view of Proposition  \ref{teoPC}. Since paracausal relation is transitive, $g$ and $g'$ are paracausally related. 
\end{proof}

\subsubsection{Paracausal deformation and Cauchy  temporal functions}

We now study the interplay of  the notion of Cauchy temporal function and the one of paracausal deformation.

A first result regards metrics that share a common foliation of  Cauchy surfaces.
We need a  preliminary technical result.
\begin{lemma}\label{lem:alpha glob hyp}
Let $(\M,g)$ be a globally hyperbolic spacetime, $t: \M \to \RR$ a Cauchy temporal function according to Definition \ref{defTT} and $\psi:\M\to\RR\times\Sigma$ a diffeomorphism mapping isometrically $(\M,g)$ to $(\RR\times\Sigma,-\beta^2d\tau\otimes d\tau\oplus h_{\tau})$. Finally let  $(\M,g_{\alpha})$ be a time oriented spacetime with time orientation such that $dt$ is past directed.\\
 If $\psi$ maps $(\M,g_\alpha)$ isometrically to $(\mathbb{R}\times\Sigma,-
d\tau \otimes d\tau\oplus \alpha^{2}(\tau)\beta^{-2}h_\tau)$ with $\alpha\in C^\infty(\mathbb{R},(0,\infty))$, then $(\M,g_{\alpha})$ is globally hyperbolic.
\end{lemma}

\begin{proof} 
 We will henceforth omit to write the isometry $\psi$ and consider, without loss of generality, $\M= \RR \times \Sigma$, $t=\tau$, $g=-\beta^2dt\otimes dt\oplus h_{t}$ and $g_{\alpha}=-
	dt \otimes dt\oplus \alpha^{2}(t)\beta^{-2}h_t$.\\
On account of~Theorem~\ref{thm:Geroch}, it is enough  to prove that $\Sigma$,  viewed as the  $t=0$ slice of the temporal function $t$, is a spacelike Cauchy hypersurface for 
$g_\alpha$. Evidently $\Sigma$ is a spacelike hypersurface  for $g_\alpha$ so that it suffices to prove that
it meets exactly once every inextendible future directed $g_\alpha$-timelike  curve $\gamma : I\ni s  \mapsto \gamma(s) \in \M$.\\
Since $\frac{dt}{ds}= g_\alpha(\partial_t, \dot{\gamma}) <0$ by hypothesis, that  $\gamma$ can be re-parametrized by $t$ itself as $\gamma' : J \ni t \mapsto \gamma'(t) \in M$ for some open interval $J\subset \RR$. There must exist a finite  $a>0$ such that 
$(-a,a) \cap  J \neq \emptyset$. Since $\gamma'|_{(-a,a) \cap  J }$ is inextendible in the spacetime $(-a,a) \times \Sigma$ (otherwise it would not be inextendible  in the whole spacetime), to conclude it is sufficient to prove that $(-a,a) \times \Sigma$  equipped with the metric $g_\alpha$ and the time-orientation 
induced by $dt$ admits $\Sigma$ as a Cauchy surface. Indeed, in that case,  $\gamma'$ must meet $\Sigma$ exactly once in $(-a,a) \times \Sigma$  and thus $\Sigma$ would be a Cauchy hypersurface for $(\RR \times \Sigma, g_\alpha)$.
Moreover, notice that it cannot meet $\Sigma = t^{-1}(0)$ again outside   $(-a,a) \times \Sigma$ because $\gamma'$ is parametrized by $t$.
Global hyperbolicity of $((-a,a) \times \Sigma, g_\alpha)$ can be proved as follows.
If $a>0$, there exists a positive constant $\alpha_0$ such 
that $\alpha(t)\geq\alpha_0>0$ for all $t\in[-a,a]$. We therefore have $g_\alpha\preceq 
g_{\alpha_0}$ on $(a,b)\times\Sigma$.
In particular,  with the time-orientation declared in the hypothesis, every future-directed causal tangent vector for $g_\alpha$ is  a future-directed causal vector for 
$g_{\alpha_0}$.
Therefore, according to (2) in Lemma \ref{prop:metrics0},  it suffices to show that $g_{\alpha_0}$ is globally hyperbolic on 
$(-a,a)\times\Sigma$ and that $\Sigma$ is a Cauchy hypersurface for $g_{\alpha_0}$. To this end, consider an inextendible  future-directed timelike curve $\gamma=(\gamma^0,\hat{\gamma})$ in 
$((-a,a)\times\Sigma,g_{\alpha_0})$. The curve 
$\tilde{\gamma}:=(\alpha_0^{-1}\gamma_0,\hat{\gamma})$ is future directed  timelike w.r.t. $g$ and still inextendible , therefore it meets $\Sigma = t^{-1}(0)$ exactly once, but $\tilde{\gamma}$ and $\gamma$ intersect in $t=0$. Thus $\gamma$ intersects $\Sigma$ once.
This shows $g_{\alpha_0}$ and therefore $g_\alpha$ to be globally hyperbolic on 
$(-a,a)\times\Sigma$. 
\end{proof}

We can now state and prove a first result concerning Cauchy surfaces and the paracausal relation\footnote{The following proof is actually extracted by a result due to M. S\'anchez  who, with Theorem 3.4 of  \cite{Sanchez22},  improved a similar statement in a previous version of this work where we also assumed that the Cauchy surfaces were compact. We are grateful to M. S\'anchez for providing this improved version of our result.}. 

\begin{proposition}\label{prop:paracausal same Cauchy temporal} 
Let  $(\M,g)$ and $(\M,g')$ be  globally hyperbolic spacetimes on $\M$ which share a  Cauchy temporal function $t: \M \to \RR$ according to  Definition \ref{defTT}. Then $g\simeq g'$. 
\end{proposition}
\begin{proof} As before,  we will  henceforth omit to write the isometries identifying the various spacetimes. However we may have two different isometries from $\M$ to $\RR \times \Sigma$ for $g$ and $g'$.
 Proposition \ref{remchichi0chi1} yields  $g \preceq \hat g \preceq g$,  $g' \preceq \hat g' \preceq g'$ if
$$ \hat g:=\beta_0^{-2} g = -dt\otimes dt + \beta_0^{-2} h_{t}  \qquad \text{and} \qquad  \hat g':=\beta_1^{-2} g' = -dt\otimes dt + \beta^{-2}_1 h'_{t}\,,$$
where $\beta^2_0,\beta^2_1$ are the lapse function we choose in accordance with Theorem~\ref{thm: Sanchez}.  
The metrics $\hat{g}$ and $\hat{g}'$ are globally hyperbolic for Lemma \ref{lem:alpha glob hyp} (with $\alpha=1$). The proof ends if proving that $\hat g$ and $\hat g'$ are paracausally related.  Referring to the splitting of $\M$ as $\RR \times \Sigma$ induced by  the Cauchy temporal function $t$,  define the globally hyperbolic metric $-dt\otimes dt  + h$, where $h$ is a complete Riemannian metric on $\Sigma$ (see, e.g., \cite{Sanchez97}). For every $\lambda \in (0,1)$, direct inspection proves that, 
$$g_\lambda :=\lambda (-dt\otimes dt+ h) + (1-\lambda)\hat g  = -dt\otimes dt + \lambda h + (1-\lambda)  \beta_0^{-2}  h_t \preceq -dt\otimes dt + \lambda h$$
and
$$g_\lambda = \lambda (-dt\otimes dt + h) + (1-\lambda) \hat g  = -dt\otimes dt + \lambda h + (1-\lambda)  \beta_0^{-2}  h_t \preceq -dt\otimes dt + (1-\lambda)  \beta_0^{-2}  h_t\:.$$
Since $\lambda h$ is complete,  from the former line we conclude that the  metric $g_\lambda$ is globally hyperbolic due to (2) Lemma \ref{prop:metrics0} and that it is paracausally related to $dt\otimes dt + \lambda h$.
From the latter, since $-dt\otimes dt + (1-\lambda)  \beta_0^{-2}  h_t$ is globally hyperbolic in view of Lemma \ref{lem:alpha glob hyp}, we have that this metric and $g_\lambda$
are paracausally related. Since $(1-\lambda)\in (0,1)$, the cones of  $-dt\otimes dt + (1-\lambda)  \beta_0^{-2}  h_t$ include  the cones of 
$-dt\otimes dt +  \beta_0^{-2}  h_t = \hat g$ so that these metrics are  paracausally related as well.
Transitivity implies that  $-dt\otimes dt + \lambda h$ and $\hat g$ are paracausally related.
The same argument proves that $-dt\otimes dt + \lambda h$ and $\hat g'$ are paracausally related so that 
$\hat g \simeq \hat g'$ and 
the thesis holds.
\end{proof}

Now we prove another non trivial result about paracausally related metrics for Cauchy compact spacetimes.

\begin{proposition}\label{thm:paracausal}
Let  $(\M,g)$ and $(\M,g')$ be  spacetimes such  that $g,g'\in \mathcal{GH}_\M$.
Suppose that $g$ admits a Cauchy temporal function $t: \M \to \RR$ whose spacelike  Cauchy hypersurfaces are compact and are also $g'$-spacelike, then $g\simeq g'$ up to a change of the temporal orientation of $g'$.
\end{proposition}
\begin{proof}
First of all, by defining the $g$-normal $n_g=\frac{\sharp_g dt}{\sqrt{-g^{\sharp}(dt,dt)}}$, any vector field $X$ can be written as $X=X_n n_g+\pi_g(X)$, where $X_n=g(n_g,X)$ and $\pi_g(X)=Id-g(n_g,X)n_g$ projects on the Cauchy surface.\\
The metric tensor $g=- \frac{dt^2}{g^{\sharp}(dt,dt)}+h(\pi(\cdot),\pi(\cdot))$ under the action of the diffeomorphism $\psi_g$ gets recast in the orthogonal form $g_{ort}=-\beta dt^2+h_t$. This metric is obviously, by \ref{remchichi0chi1}, paracausally related to the conformal metric $g_c=- dt^2 +\frac{1}{\beta_t}h_t$, which is, in turn, paracausally related to the globally hyperbolic metric $\tilde{g}=-dt^2+h$ with $h$ a complete Riemannian metric on the slice and if we choose coherently the time orientation. The last statement is a consequence of \ref{teoPC} which is proved exactly as $\ref{cor:ultr para}$.\\
Then we want look at the metric $g'$ after the action of the isometric diffeomorphism $\psi_g$ and define the $\tilde{g}'=\psi_g^{*}g'$.
 The proof ends if we are able to find a globally hyperbolic metric $g''$ such that $\tilde{g}\simeq g''\simeq \tilde{g}'$.\\
If we choose a function $\alpha\in C^{\infty}(\RR,(0,\infty))$, then, by lemma \ref{lem:alpha glob hyp} the metric tensor $g_{\alpha}=-dt^2+\alpha(t)h$ is globally hyperbolic and, by \ref{teoPC}, paracausally related to $\tilde{g}$. We want to tune the fuction $\alpha$ in order to have that the cones of $g_{\alpha}$ intersect the cones of $\tilde{g}'$.\\
First we define pointwise $n'$ the smooth vector field $g'$-normal to the Cauchy hypersurfaces of $g$ and decompose it with respect to the splitting of the tangent space induced by the metric $g_{\alpha}$ through its normal $n_{\alpha}$. We get $n'=Zn_{\alpha}+W$ where $Z=g_{\alpha}(n_{\alpha},n')$ and $W=\pi_{g_{\alpha}}(n')$.\\
Since the Cauchy hypersurfaces of $g$ and $g_{\alpha}$ are spacelike also for $g'$, we have that $Z\neq0$. The cones of the two metrics intersct if $\alpha$ is such that $n'$ is $g_{\alpha}$-timelike i.e. iff
$$||n'||_{g_{\alpha}}=-|Z|^2+\alpha(t)||W||^2_{h}<0 \iff \frac{1}{\alpha(t)}>\frac{||W||^2_{h}}{|Z|^2}.$$
The manifold $\RR\times\Sigma$ can be covered by the time-strips $\mathcal{TS}_n=\{[-n,n]\times\Sigma\}_{n\in\mathbb{N}}$ which are obviously compact since $\Sigma$ is compact by hypotesis.\\
 This means that for all $n\in\mathbb{N}$ the smooth function $f:\M\to\mathbb{R^+}$ defined by $f:=\frac{||W||^2_{h}}{|Z|^2}$ attains a maximum $M_n$ and a minimum $m_n$ when restricted to the strip $\mathcal{TS}_n$. So we construct the required function $\frac{1}{\alpha(t)}:\mathbb{R}\to\mathbb{R}^+$ such that
 $$\frac{1}{\alpha(t)}=M_n +1\quad t\in[-n-1,-n)\cup (n,n+1].$$
 This function is only piecewise constant.
 The maximum has been increased by one to avoid the possibility that this function gets null: it could happen if the normal $n'$ and $n_{\alpha}$ get aligned in the first time-strip and then depart.\\
 The last thing to do is to smoothen the function $\alpha(t)$, something which can of course be done by standard gluing arguments.\\
 Now that we know that the cones of $g_{\alpha}$ and $\tilde{g}'$ intersect, if the temporal orientation of $g'$ is such that $V^+_{g_{\alpha}}\cap V^+_{\tilde{g}'}\neq \emptyset$ we define $g'':=g_{\alpha}\simeq \tilde{g}'$ and the proof is concluded.\\
  If $V^+_{g_{\alpha}}\cap V^+_{\tilde{g}'}= \emptyset$ the metric $g''$, and therefore the metric $\tilde{g}$, is paracausally related to $\tilde{g}'$ with opposite time orientation.
 \end{proof}

\section{Normally hyperbolic operators and their properties}\label{sec:norm hyp}
One of the main goals  of this paper is to realize a geometric map to compare the space of solutions of {\em normally hyperbolic} operators defined on possibly different globally hyperbolic manifolds. 
Before starting to introduce  our theory, we remind some general definitions and we fix the notation that will be used from now on.
 Let $\E$ be a vector bundle (always on $\KK$  and of finite rank in this paper) over a spacetime $(\M,g)$, whose generic fiber (a  $\KK$ vector space isomorphic to a canonical fiber) is  denoted by $\E_p$ for $p\in\M$. 
$\Gamma(\E)$ is the $\KK$-space of smooth sections  $\E$. 
$\Gamma(\E)$ has a number of useful subspaces we list below.
\begin{itemize}
\item[(i)]  $\Gamma_c(\E) \subset \Gamma(\E)$
is the subspace of compactly supported smooth sections.
\item[(ii)]  $\Gamma_{pc}(\E)$ and $\Gamma_{fc}(\E)$ denote the subspaces of $\Gamma(\E)$ whose elements have  respectively past compact support and future compact support.
 
\item[(iii)] If $(\M,g)$ is globally hyperbolic,  $\Gamma_{sc}(\E) \subset \Gamma(\E)$ is the 
subspace  of {\bf spatially compact} sections: the smooth sections whose support intersects every spacelike  Cauchy hypersurface in a compact set.
\end{itemize}
These spaces are equipped with natural topologies as discussed in  \cite{Ba-lect}.
In case there are several metrics on a common spacetime $\M$ basis of $\E$, the used metric $g$ will be indicated as well, for instance $ \Gamma^g_{pc}(\E)$, if  the  nature of the space of sections depends on the chosen metric (this is not the case for $\Gamma_c(\E)$).

   As usual, $\E  \boxtimes \E'$ denotes the {\bf external  tensor product} of the two $\KK$-vector bundles over $\M$. This $\KK$-vector  bundle  has basis  $\M\times\M\ni (p,q)$
 and fibers given by the pointwise tensor product $\E_p\otimes \E'_q$ of the  fibers of the two bundles.
Referring to $\Gamma(\E  \boxtimes \E')$, if $\ff\in \E$ and $\ff'\in \E'$, then $\ff\otimes \ff' \in \Gamma(\E  \boxtimes \E')$ denotes the section defined by 
$(\ff\otimes \ff')(p,q):= \ff(p) \otimes \ff'(q)$ where the tensor product on the right-hand side is the one of the fibers and  $(p,q) \in \M\times \M$.

\subsection{Normally hyperbolic  operators}

 As for Section~\ref{sec:preliminaries}, if $g\in \mathcal{M}_\M$ we defined $g^\sharp$ as the induced metric on the cotangent bundle. If $(\M,g)$ is globally hyperbolic, by fixing a Cauchy temporal function $t:\M\to\RR$ such that $g=-\beta^ 2dt\otimes dt + h_t$, we have 
$$g^\sharp= -\beta^{-2} \partial_t \otimes\partial_t+ h^\sharp_t\,.$$

\begin{definition}\label{def:Norm} A linear second order differential operator $\N:\Gamma(\E)\to\Gamma(\E)$ is  {\bf normally hyperbolic} if its principal symbol $\sigma_\N$ satisfies
	\begin{equation*}
 \sigma_\N(\xi)=-g^\sharp(\xi,\xi)\,\Id_\E 
	\end{equation*}
	for all $\xi\in\T^*\M$, where $\Id_\E$ is the identity automorphism of $\E$.
\end{definition}
Referring to a foliation of $(\M,g)$ as in Definition \ref{defTT},
in local coordinates $(t,x)$ on $\M$ adapted to the foliation so that  $x=(x_1,\dots,x_n)$ are local coordinates on $\Sigma_t$, and using a  local trivialization of $\E$, any normally hyperbolic operator $\N$ in a point $p\in\M$ reads as
$$\N= \frac{1}{\beta^2} \partial_t^2 - \sum_{i,j=1}^n {h^\sharp_t}_{ij} \partial_{x_i} \partial_{x_j} + A_0(t,x) \partial_t +\sum_{j=1} A_j(t,x)\partial_{x_j} + B(t,x)$$
where $A_0,A_j$ and $B$  are linear maps  $\E_{(t,x)} \to \E_{(t,x)}$ depending smoothly on $(t,x)$.
\begin{examples}
In the class of normally hyperbolic operators we can find many operators of interest in quantum field theory:
\begin{itemize}
\item Fix $\E$ be the trivial real bundle, i.e. $\E=\M\times \RR$, so that the space of smooth sections of $\E$ can be identified with the ring of smooth functions on $\M$.  The Klein-Gordon operator  $\N=\Box + m^2$ is normally hyperbolic, where $\Box$ is the d'Alembert operator and $m$ is a mass-term.
\item Let now $\E=\Lambda^k \T^*\M$ be the bundle of $k$-forms and  $d$ (resp $\delta$) the exterior derivative (resp. the codifferential). The operator $\N:=d\delta + \delta d +m^2$ is normally hyperbolic and it is used to describe the dynamics of Proca fields, for further details we refer to \cite[Example 2.17]{CQF1}.
\item Let $\S\M$ be a spinor bundle over a globally hyperbolic spin manifold $\M_g$ and let $\nabla$ be a spin connection. By denoting with $\gamma:\T\M\to \text{End}(\S\M)$ the Clifford multiplication, the classical Dirac operator reads as $\Dir=\gamma\circ \nabla:\Gamma(\S\M)\to\Gamma(\S\M)$, see \cite{DiracAPS,DiracMIT,MollerMIT} for further details. By Lichnerowicz-Weitzenb\"ock formula we get the spinorial wave operator
$\N = \Dir^2 = \nabla^\dagger\nabla + \frac{1}{4}\text{Scal}_g\,,$
where Scal$_g$ is the scalar curvature.
\end{itemize}
\end{examples}

It is well-known that, once the Cauchy data are suitably assigned, the Cauchy problem for $\N$ turns out to be well-posed, see {\it e.g.}~\cite{Ba-lect,Ginoux-Murro-20}.
\begin{theorem}\label{thm:main nh}
Let $\E$ be a  vector bundle (of finite rank) over  a globally hyperbolic manifold $(\M,g)$, let $\N$ be a normally hyperbolic operator with a $\N$-compatible connection $\nabla$ (see (\ref{concomp}) below) and $\Sigma_0$ a (smooth) spacelike Cauchy hypersurface of $(\M,g)$.
	Then the Cauchy problem for  $\N$  is well-posed, \ie for any $\ff\in\Gamma_c(\E)$, $\fh_1,\fh_2\in\Gamma_c(\E|_{\Sigma_0})$  there exists a unique solution $\Psi\in\Gamma_{sc}(\E)$  to the initial value problem
	\begin{equation*} 
	\begin{cases}{}
	{\N}\Psi=\ff   \\
	\Psi|_{\Sigma_0} = \fh_1   \\
	(\nabla_{\fn} \Psi)|_{\Sigma_0}=\fh_2 \\
	\end{cases} 
	\end{equation*}
 being $\fn$ the future directed timelike unit normal field along $\Sigma_0$, and it
	 depends continuously on the data $(\ff,\fh_1,\fh_2)$ w.r.to the standard topologies of smooth sections  and satisfies \begin{equation}\label{finprop}\mbox{supp}(\Psi) \subset  J(\mbox{supp}(\ff))
\cup J(\mbox{supp}(\fh_1)) \cup J(\mbox{supp}(\fh_2))\:.\end{equation}
\end{theorem}

As a consequence of the well-posedness of the Cauchy problem of normally hyperbolic operators
with  ``finite propagation of the solutions'' stated in (\ref{finprop}),
one may establish  the existence of Green operators. In order to recall this result, we need first a preparatory definition.

\begin{definition}\label{def:GreenHyp} 
A  linear differential operator $\P:\Gamma(\E)\to\Gamma(\E)$ is called  {\bf Green hyperbolic} 
if  
\begin{itemize}
\item[(1)] there exist linear  maps, dubbed {\bf advanced Green operator}
	$\G^+\colon \Gamma_{pc}(\E)  \to  \Gamma(\E)$  and {\bf retarded Green operator} $\G^-\colon  \Gamma_{fc}(\E)  \to  \Gamma(\E)$,  satisfying
	\begin{itemize}
		\item[(i.a)] $\G^+ \circ \P\, \ff  = \P \circ \G^+ \ff=\ff$ for all $ \ff \in  \Gamma_{pc}(\E)$ ,
		\item[(ii.a)]  $\supp(\G^+ \ff ) \subset J^+ (\supp \ff )$ for all $\ff \in  {\Gamma}_{pc} (\E)$;
		\item[(i.b)]  $\G^- \circ \P\, \ff  = \P\circ \G^- \ff=\ff$ for all $ \ff \in  \Gamma_{fc}(\E)$,
		\item[(ii.b)]  $\supp(\G^- \ff ) \subset J^- (\supp \ff )$ for all $\ff \in  {\Gamma}_{fc} (\E)$;
	\end{itemize}
\item[(2)] the {\bf formally dual operator} $\P^*$
 admits advanced and retarded Green operators as well.

\end{itemize}
\end{definition}
For sake of completeness, let us recall that the formally dual operator $\P^*:\Gamma(\E^*)\to\Gamma(\E^*)$  is  the unique linear differential operator
acting on the smooth sections of the dual bundle $\E^*$ satisfying 
$$\int_\M \langle\ff',\P \ff\rangle \:  \mbox{vol}_g = \int_\M \langle{\P^*} \ff', \ff\rangle\:  \mbox{vol}_g \quad $$
for every $\ff\in \Gamma_c(\E)$ and $\ff'\in \Gamma_c(\E^*)$
(which is equivalent to saying $\ff\in \Gamma(\E)$ and $\ff'\in \Gamma(\E^*)$
 such that $\mbox{supp}(\ff) \cap \mbox{supp}(\ff)'$ is compact),
$ \mbox{vol}_g$ being the volume form induced by $g$ on $\M$.
\begin{remarks}
\begin{itemize}
\noindent \item[(1)] The Green operators we define below are the extensions  to $\Gamma_{pc/fc}(\E)$ of the analogs defined in \cite{Ba} and indicated by $\overline{G}_\pm$ therein.
\item[(2)] It is possible to prove that the Green operators are unique for a Green hyperbolic  operator  (cf. \cite[Corollary 3.12]{Ba}).
Furthermore as a consequence of  \cite[Lemma 3.21]{Ba}, it arises that if $\ff' \in \Gamma_c(\E^*)$ and  $\ff \in \Gamma_{pc}(\E)$ or 
$\ff \in \Gamma_{fc}(\E)$ respectively,
\begin{equation}\label{eqGstar}
\int_\M \langle\G_{P^*}^{-}\ff', \ff\rangle \: \mbox{vol}_g = \int_\M \langle\ff', \G_{P}^+\ff\rangle \: \mbox{vol}_g \:,\quad
\int_\M \langle\G_{P^*}^{+}\ff',\ff\rangle \: \mbox{vol}_g = \int_\M \langle\ff', \G_{P}^- \ff\rangle \: \mbox{vol}_g \:,
 \end{equation}
where 
$\G_{P}^{\pm}$ indicate the Green operators of $\P$ and $\G_{P^*}^{\pm}$ indicate the Green operators of $\P^*$.
\end{itemize}

\end{remarks}

\begin{proposition}\label{propprodgreen}  If $\P$ is a Green hyperbolic operator on a vector bundle $\E$ over the globally hyperbolic spacetime $(\M,g)$ and $\rho : \M \to (0,+\infty)$ is smooth, then $\rho \P$ is Green hyperbolic as well and $\G^\pm_{\rho \P} = \G_\P^\pm \rho^{-1}$.
\end{proposition}

\begin{proof} The thesis immediately follows form the fact that $\G_\P^\pm \rho^{-1}$ and $\rho^{-1} \G_{\P^*}^\pm$ 
satisfy the properties of the Green operators for $\rho\P$ and $(\rho \P)^* = \P^* \rho$ respectively.
\end{proof}

\begin{proposition}[\protect{\cite[Corollary 3.4.3]{Ba-lect}}]\label{prop:Green}
A normally hyperbolic operator $\N$ on a  vector bundle $\E$ (of finite rank) on a globally hyperbolic manifold $(\M,g)$ is Green hyperbolic.
\end{proposition}

Given a Green hyperbolic operator with Green operators $\G^\pm$, a relevant operator constructed out of $\G^\pm$ is the so-called {\bf causal propagator},
\begin{equation*}
\G := \G^+|_{\Gamma_c(\E)} -  \G^-|_{\Gamma_c(\E)}   : \Gamma_c(\E) \to \Gamma(\E)\:.
\end{equation*}
It satisfies remarkable properties we are going to discuss  (see {\it e.g.}~\cite[Theorem 3.6.21]{Ba-lect}).

\begin{theorem} Let  $\G$ be the causal propagator of a Green hyperbolic differential operator $\P: \Gamma(\E) \to \Gamma(\E)$  on the vector bundle $\E$ over a globally hyperbolic spacetime $(\M,g)$. The  following sequence is exact
$$\{0\} \to \Gamma_c(\E) \stackrel{\P}{\rightarrow} \Gamma_c(\E) \stackrel{\G}{\rightarrow}  \Gamma_{sc}(\E)  \stackrel{\P}{\rightarrow}
 \Gamma_{sc}(\E)   \to \{0\}\:.$$
\end{theorem}

\begin{proof}
Injectivity of  $\Gamma_c(\E) \stackrel{\P}{\rightarrow} \Gamma_c(\E)$ easily arises from the well-posedness of the Cauchy problem stated in Theorem \ref{thm:main nh}.  Let us pass to the other parts of the sequence.
First of all notice that $\G^\pm(\Gamma_c(\E)) \subset \Gamma_{sc}(\E)$ since $\mbox{supp}(\G^\pm(\ff))  \subset J_\pm(\mbox{supp}(\ff))$ and the first  assertion then follows form known facts of globally hyperbolic spacetimes.
Let us prove that $ \Gamma_c(\E) \stackrel{\G}{\rightarrow}  \Gamma_{sc}(\E) $ is surjective when the image is restricted to the kernel of 
$ \Gamma_{sc}(\E)  \stackrel{\P}{\rightarrow}
 \Gamma_{sc}(\E)$.
Suppose that $\P \Psi =0$ for $\Psi \in \Gamma_{sc}(\E)$.
 If $t$ is a smooth Cauchy time function of $(\M,g)$ and $\chi : \M \to [0,1]$ is smooth, vanishes for $t< t_0$ and is constantly $1$ for $t>t_1$, then
$$\ff_\Psi := \P (\chi \Psi)  \in \Gamma_c(\E)$$
is such that $\Psi = \G \ff_\Psi$. Notice that $\mbox{supp}(\ff_\Psi)$ is included between the Cauchy hypersufaces $t^{-1}(t_0)$ and $t^{-1}(t_1)$. Indeed, $$\G \ff_\Psi = \G^+ \P(\chi \Psi) -  \G^- \P(\chi \Psi)  = \G^+ \P(\chi \Psi) +  \G^- \P((1-\chi) \Psi) = \chi \Psi + (1-\chi) \Psi = \Psi\:.$$
It is obvious that that $\ff_\Psi$ can be  changed by adding a section of  the form  $\P\fh$  with $\fh \in \Gamma_c(\E)$  preserving the property $\G \ff_\Psi = \Psi$. This exhaust the kernel  of
$ \Gamma_c(\E) \stackrel{\G}{\rightarrow}  \Gamma_{sc}(\E)$
as asserted in the thesis. Indeed, if $\G\ff=0$, then $\G^+\ff = \G^-\ff$. From the properties of the supports of $\G^\pm\ff$, we conclude that $\G^\pm \ff = \fh_\pm\in \Gamma_c(\E) \subset \Gamma_{pc}(\E) \cap \Gamma_{fc}(\E)$. Hence $\ff = \P\G^\pm \ff  = \P \fh_\pm$.
To conclude, we prove that $  \Gamma_{sc}(\E)  \stackrel{\P}{\rightarrow}
 \Gamma_{sc}(\E) $ is surjective.  If $\ff \in  \Gamma_{sc}(\E)$, with $\chi$ as above, $$\ff = \chi \ff + (1-\chi) \ff = \P \G^+ (\chi \ff)
+ \P\G^-((1-\chi) \ff) = \P [\G^+ (\chi \ff) + \G^-((1-\chi)\ff)]$$
and $\G^+ (\chi \ff) + \G^-((1-\chi)\ff) \in \Gamma_{sc}(\E)$.
\end{proof}

\subsection{Formally selfadjoint normally hyperbolic operators and their symplectic form}
Let $\E$ be a $\KK$-vector bundle on a globally hyperbolic  spacetime $(\M,g)$.
As shown in~\cite[Lemma 1.5.5]{wave}, for any normally hyperbolic operator $\N: \Gamma(\E) \to \Gamma(\E)$ there exists a unique covariant derivative $\nabla$ on $\E$ such that \begin{equation}\label{concomp}\N=-\mbox{tr}_g(\nabla\nabla)+ c\end{equation} for some some zero-order differential operator  $c: \Gamma(\E) \to \Gamma(\E)$. 
In the formula above the left $\nabla$ is actually the connection induced on $\T^*\M \otimes \E$ by the {\em Levi-Civita connection} associated to $g$  and the original connection $\nabla$ (the one appearing as the right $\nabla$) given on $\E$
Adopting the terminology of~\cite{Ba-lect}, we shall refer to $\nabla$ as the {\bf $\N$-compatible connection} on $\E$.

We stress that, if we  suppose that   $\E$ is equipped  with a smooth assignment of a Hermitian fiber metric
$$\fiber{\cdot}{\cdot}_p :\E_p\times\E_p\to \KK\,.$$
then  the above $\nabla$ is  $g$-metric but not necessarily metric with respect to $\fiber{\cdot}{\cdot}$.

The physical relevance of the fiber metric is that it permits to equip $\Ker_{sc}(\N)$ with a symplectic form with important physical properties in the formulation of QFT in curved spacetime.  This symplectic form can be derived using the Green identity for a normally hyperbolic operator $\N$ and its formal adjoint. For sake of completeness let us remind the definition of formal adjoint.

\begin{definition}\label{defadjointdiff}  The {\bf formal adjoint} of a differential operator $\P:\Gamma(\E) \to \Gamma(\E)$  is the unique differential operator $\P^\dagger : \Gamma(\E) \to \Gamma(\E)$ satisfying
$$\int_\M \fiber{\ff'}{\P \ff} \:  \mbox{vol}_g = \int_\M \fiber{\P^\dagger \ff'}{\ff}\:  \mbox{vol}_g \quad $$
for every $\ff, \ff'\in \Gamma_c(\E)$ (which is equivalent to saying $\ff, \ff'\in \Gamma(\E)$ such that $\mbox{supp}(\ff) \cap \mbox{supp}(\ff)'$ is compact). If $\P = \P^\dagger$ then $\N$ is said to be (formally) {\bf selfadjoint}.
\end{definition}

\begin{remark} If $\P : \Gamma(\E) \to \Gamma(\E)$ is normally hyperbolic on the bundle $\E$ over $(\M,g)$, equipped with a non-degenerate, Hermitian fiber metric $\fiber{\cdot}{\cdot}$,
$\P$ is  Green hyperbolic as said above. In this case  $P^\dagger$ has the same  principal symbol as $\P$ and thus it is Green hyperbolic as well.
Taking advantage of the natural (antilinear if $\KK=\mathbb{C}$) isomorphism $\Gamma(\E) \to \Gamma(\E^*)$ induced by 
$\fiber{\cdot}{\cdot}$ and (\ref{eqGstar}), it is not difficult to prove  that, if $\ff' \in \Gamma_c(\E)$ and  $\ff \in \Gamma_{pc}(\E)$ or 
$\ff \in \Gamma_{fc}(\E)$ respectively,
\begin{equation}\label{eqGstar2}
\int_\M \fiber{\G_{P^\dagger}^{-}\ff'}{ \ff} \: \mbox{vol}_g = \int_\M \fiber{\ff'}{ \G_P^+\ff} \: \mbox{vol}_g \:,\quad
\int_\M \fiber{\G_{P^\dagger}^{+}\ff'}{\ff} \: \mbox{vol}_g = \int_\M \fiber{\ff'}{\G_P^- \ff} \: \mbox{vol}_g \:.
 \end{equation}
where 
 $\G_{P}^{\pm}$ indicate the Green operators of $\P$ and 
$\G_{P^\dagger}^{\pm}$ indicate the Green operators of $\P^\dagger$. 
\end{remark}

Let us pass to introduce a Green-like identity where we explicitly exploit the $\N$-compatible connection $\nabla$.

\begin{lemma}[Green identity]\label{lem:green id}
Let $\E$ be an non-degenerate, Hermitian  $\KK$ vector bundle over a $(n+1)$-dimensional spacetime $(\M,g)$ and denote the fiber metric $\fiber{\cdot}{\cdot}$. Moreover, let $\N: \Gamma(\E) \to \Gamma(\E)$ be a normally hyperbolic  operator with $\N$-compatible connection $\nabla$.
Let $\M_0\subset \M$ be a 
submanifold with continuous piecewise smooth boundary.
Then for every $\Phi,\Psi \in \Gamma_c(\E))$
\begin{equation}\label{eq:Green Id}
\int_{\M_0} \left( \fiber{\Psi}{\N \Phi}-  \fiber{\N \Psi}{\Phi}\right) \vol_g  =
\int_{\bM_0}\Xi^\N_{\partial \M_0} (\Psi,\Phi) \,,
\end{equation}
where $\Xi_{\partial \M_0}^\N$ is the $n$-form in $\partial \M_0$
\begin{equation*}
\Xi^\N_{\partial \M_0} (\Psi,\Phi)  :=\imath^*_{\partial \M_0} \left[\sharp\Big(\fiber{\Psi}{\nabla\Phi}-\fiber{\nabla\Psi}{\Phi}\Big)  \intprod \vol_g\right]
\end{equation*}
 $\imath_{\partial \M_0} : \partial \M_0 \to \M$ being the inclusion embedding.
If the  normal vectors to $\partial\M_0$ are either spacelike or timelike (up to zero-measure sets), then 
\begin{equation}\Xi^\N_{\partial \M_0} (\Psi,\Phi)  = \Big(\fiber{\Psi}{\nabla_\n\Phi}-\fiber{\nabla_\n\Psi}{\Phi}\Big) \vol_{\partial \M_0}
\label{formXi}\end{equation}
where $\n$ is the outward unit normal vector to $\bM_0$ and  $\vol_{\partial \M_0} = \n\intprod \vol_g$ is the volume form of $\partial \M_0$ induced by $g$.
\end{lemma}
\begin{proof}
Consider  the $n$-form in $\M$
$$Z :=	\sharp\Big(\fiber{\Psi}{\nabla\Phi}-\fiber{\nabla\Psi}{\Phi}\Big)  \intprod \vol_g\:.$$
If the  normal vectors to $\partial\M_0$ are either spacelike or timelike, 
some computations with the exterior differential of forms yields (\ref{formXi}).
In all cases it is easy to prove that
$$
dZ =\Big(\fiber{\Psi}{g^{ij}\nabla_i\nabla_j\Phi} - \fiber{g^{ij}\nabla_i\nabla_j\Psi}{\Phi}\Big)\vol_g\\
=\Big(\fiber{\Psi}{\N\Phi} - \fiber{\N\Psi}{\Phi}\Big)\vol_g\:.$$
At this juncture,  Stokes' theorem  for $(n+1)$-forms,
$$\int_{\partial \M_0} Z = \int_{\M_0} dZ\:,$$
 produces \eqref{eq:Green Id}.
\end{proof}

We have the following crucial result when applying the previous lemma to the theory on globally hyperbolic spacetimes.

\begin{proposition}\label{lem:indip Sigma}
Let $\Sigma \subset \M$ be a smooth
spacelike Cauchy hypersurface with its future-oriented unit normal vector field $\n$
in the globally hyperbolic spacetime $(\M,g)$
 and its induced volume element $\vol_{\Sigma}$. Furthermore, let $\N$ be a formally  self-adjoint normally hyperbolic operator. Then
\begin{equation}\label{def:sympl form}
\sigma^\N_{(\M,g)} : \Ker_{sc}(\N)\times \Ker_{sc}(\N) \to \CC \quad\mbox{such that}\quad  \sigma^\N_{(\M,g)}(\Psi,\Phi)= \int_\Sigma \Xi^\N_{\Sigma} (\Psi,\Phi) 
\end{equation}
where $\Xi^\N_\Sigma$ is defined in Equation~(\ref{formXi}),
yields a non-degenerate symplectic form  (Hermitian if $\KK=\CC$) which does not depend on the choice of $\Sigma$.
\end{proposition}
\begin{proof}  
First note that, referring to a spacelike Cauchy hypersurface $\Sigma$,  $\text{supp}( \Psi) \cap \Sigma$ is compact since $\text{supp}( \Psi)$ is spacelike compact,
so that the integral is well-defined.
 The fact that $\sigma_\N$ is not degenerate can be proved as follows.
If $\sigma^\N_{(\M,g)}(\Psi,\Phi)$ =0 for all $\Phi \in \Gamma_{sc}(\E)$, from the  definition of $\sigma_\N$ and non-degenerateness of $\fiber{\cdot}{\cdot}_p$ (passing to local trivializations referred to local coordinates on $\Sigma$ re-writing 
$\fiber{\cdot}{\cdot}_p$ in terms of the pairing with $\E_p^*$), we have that the Cauchy data of $\Psi$ vanishes on every local chart on $\Sigma$ and thus they vanish on $\Sigma$. According to Theorem \ref{thm:main nh}, $\Psi =0$.  The other entry can be worked out similarly.\\
Let 
$\Psi,\Phi\in\Ker_{sc}(\N)$ and $\Sigma_t'$ and $\Sigma_{t''}$ be a pair of  smooth spacelike Cauchy
hypersurfaces  associated to a smooth time Cauchy function $t$ with $t''>t'$.
Let us focus on the submanifold with boundary  $\M_0=t^{-1}((t',t''))$. Its boundary is $\bM_0= \Sigma_{t'} \cup \Sigma_{t''}$.
The supports of $\Psi$ and $\Phi$  between the two Cauchy surfaces are included in the sets of type $J_+$ of the compact supports of the  Cauchy data on $\Sigma_{t'}$  of $\Psi$ and $\Phi$  respectively, and these portions of causal sets are compact as $(\M,g)$ is globally hyperbolic  (see {\it e.g.}~\cite[Proposition 1.2.56]{Ba-lect}) we end up with a pair of functions in $\Gamma_c(\E)$ and we can apply  the Green identity (see Lemma~\ref{lem:green id}) to $\M_0$. 
Using a smoothly vanishing function as a factor, we can make smoothly vanishing $\Psi$ and $\Phi$ before $\Sigma_{t'}$ and after $\Sigma_{t''}$ without touching them between the two Cauchy surfaces. 
As  a matter of fact the resulting sections constructed out $\Psi$ and $\Phi$ by this way  are smooth, compactly supported and coincide with $\Psi$ and $\Phi$ between the two Cauchy surfaces. We can therefore apply Lemma  \ref{lem:green id},
obtaining
$$ \int_{\M_0}(\fiber{\Psi}{ \N\Phi} - \fiber{ \N \Psi}{\Phi} )\vol_{g} = \int_{\Sigma'}\Xi^\N_{\Sigma'} - \int_\Sigma\Xi^\N_\Sigma \:. $$
 Since $\N$ is assumed to be self-adjoint, $\fiber{\Psi}{ \N\Phi} - \fiber{ \N \Psi}{\Phi} =
\fiber{\Psi}{ \N\Phi} - \fiber{ \Psi}{\N^\dagger \Phi} =
0$. Therefore we can conclude that 
 $\int_{\Sigma'}\Xi^\N_{\Sigma'} = \int_\Sigma\Xi^\N_\Sigma $. Finally consider the case of two spacelike Cauchy functions $\Sigma$ and 
$\Sigma'$ belonging to different foliations induced by different smooth Cauchy time functions (notice that  a spacelike Cauchy hypersurface always belong to a foliation generated by a suitable smooth  Cauchy time (actually temporal)  function for Theorem \ref{thm: Sanchez}).
We sketch a proof of the identity
$$ \int_{\Sigma'}\Xi^\N_{\Sigma'} = \int_\Sigma\Xi^\N_\Sigma \:. $$
Let $K\subset \Sigma$ a compact set including the Cauchy data of $\Psi$ and $\Phi$. If $t$ is the smooth Cauchy time function such that $\Sigma_{t_1} = \Sigma'$, let $T = \max_{K} t$.  If $t_1 < T$ we can always take $t_2>T$ and to consider 
the symplectic form evaluated on $\Sigma_{t_2}$. In view  of the previous part of our proof
the symplectic form on $\Sigma_{t_1}$ and $\Sigma_{t_2}$ coincide, so that 
our thesis can be re-written
$$ \int_{\Sigma_2}\Xi^\N_{\Sigma'} = \int_\Sigma\Xi^\N_\Sigma \:. $$
As $t_2 >  \max_{K} t$, we conclude that $\Sigma_{t_2}$ does not intersect $\Sigma$ in the set $K$. 
Therefore we can define the solid set $L_K$ made of the portion of $J_+(K)$ between $\Sigma$ and $\Sigma_{t_2}$.  $L$ is compact  (see {\it e.g.}~\cite[Proposition 1.2.56]{Ba-lect}) and is a ``truncated cone'' whose ``lateral surface'' is part of the boundary of $J_+(K)$ and whose ``non-parallel bases'' are parts of $\Sigma_2$ and $\Sigma$. We can include $L$ in the interior of a larger manifold with boundary $\M_0$ whose part of the boundary are portions of $\Sigma$ and $\Sigma_{t_2}$ including the support of the Cauchy data of $\Psi$ and $\Phi$. Notice that $\M_0$ includes the supports of $\Psi$ and $\Phi$ between the two Cauchy surfaces according to Theorem \ref{thm:main nh} and these supports do not touch the ``lateral surface'' of $\M_0$. We can now apply the Green identity~\ref{lem:green id} to $\M_0$ proving the thesis.
\end{proof}
There is a nice interplay of the causal propagator $\G$ of $\N : \Gamma(\E) \to \Gamma(\E)$ as above and the symplectic form $\sigma^\N_{(\M,g)}$.

\begin{proposition}\label{propCPsim} With the same hypotheses as of Proposition \ref{lem:indip Sigma}, if $\ff, \fh \in \Gamma_c(\E)$ and $\Psi_\ff := \G \ff$, $\Psi_\fh := \G \fh$, it holds
\begin{equation*}
\sigma^\N_{(\M,g)}(\Psi_\ff,\Psi_\fh) = \int_\M \fiber{\ff}{\G \fh} \: \vol_g\:.
\end{equation*}
\begin{proof}  If $\ff,\fh \in \Gamma_c(\E),$
 consider a smooth Cauchy time function $t$ and fix $t_0< t_1$ such that the supports of $\ff$ and $\fh$ are included in the interior of 
the submanifold with boundary $\M_0$ contained between the spacelike Cauchy hypersurfaces $\Sigma_{t_0}:=t^{-1}(t_0)$ and $\Sigma_{t_1}:=t^{-1}(t_1)$. It holds
$$\int_{\M} \fiber{\Psi_\ff}{\fh} \: \vol_g = \int_{\M_0} \fiber{\Psi_\ff}{\fh} \: \vol_g = \int_{\M_0} \fiber{\Psi_\ff}{\N \G^+\fh} \: \vol_g$$
Since $\N\Psi_\ff=0$, we have found that
$$ \int_{\M} \fiber{\ff}{\Psi_\fh} \: \vol_g =  \int_{\M_0} \left(\fiber{\Psi_\ff}{\N\G^+\fh} -  \fiber{\N\Psi_\ff}{\G^+\Psi_\fh}\right)\: \vol_g\:.$$
Applying Lemma \ref{lem:green id}, we find
$$ \int_{\M} \fiber{\ff}{\Psi_\fh} \: \vol_g = \int_{\bM_0}\Xi^\N_{\partial \M_0} (\Psi_\ff,\G^+\fh)
= \int_{\Sigma_{t_1}}\Xi^\N_{\partial \M_0} (\Psi_\ff,\G^+\fh)\:,$$
where we noticed that $\G^+\fh$ vanishes on the remaining part of the boundary $\Sigma_{t_0}$.  On the other hand, we can replace 
$\G^+\fh$ for $\G^+\fh- \G^-\fh = \G\fh$ in the last integral, since $\G^-\fh$ gives no contribution to the integral on $\Sigma_{t_1}$.
In summary,
$$ \int_{\M} \fiber{\ff}{\G\fh} \: \vol_g  =  \int_{\M} \fiber{\ff}{\Psi_\fh} \: \vol_g 
=  \int_{\Sigma_{t_1}}\Xi^\N_{\Sigma_{t_1}} (\Psi_\ff,\G\fh) = \sigma^{\N}_{(\M,g)}(\Psi_\ff, \Psi_\fh)\:.$$
and this is the thesis.
\end{proof}
\end{proposition}

\subsection{Convex combinations of normally hyperbolic operators}
Let now $\N_0$, $\N_1$ be normally hyperbolic operators with respect to different Lorentzian metrics $g_0$ and $g_1$ (the former time-orientable and the latter globally hyperbolic) on the same manifold $\M$ and assume that they are acting on the smooth sections of the same vector bundle $\E$. It turns out, that a positive (and convex) combination $(1-\chi) \N_0 + \chi \N_1$ is also (a) normally  hyperbolic with respect to the naturally associated  metric $g_{\chi}$ -- the unique Lorentzian metric  in $\T\M$ whose associated metric in  $\T^*\M$ is $(1-\chi) g_0^\sharp + \chi g_1^\sharp$
according to Theorem~\ref{thm:metrics1} --  and (b)  Green hyperbolic with respect to $g_1$, everything provided that $g_0\preceq g_1$. This is the main result of this section.

\begin{theorem}\label{lem:DefNorm1} 
Let $\E$ be a $\KK$-vector bundle over a smooth manifold $\M$, let be $g_0,g_1\in\mathcal{GM}_\M$  with $g_0\preceq g_1$, and let  $\N_0, \N_1:\Gamma(\E)\to\Gamma(\E)$	
be normally hyperbolic operator with respect to $g_0$ and $g_1$ respectively.
If $\chi\in C^{\infty}(\M, [0,1])$, define  $g_{\chi}$  as
the unique Lorentzian metric whose associated metric in $\T^*\M$ is $(1-\chi) g_0^\sharp + \chi g_1^\sharp$
according to Theorem~\ref{thm:metrics1}.
Then the  second order differential operator defined by
	\begin{equation}\label{def:Nchi}
	\N_{\chi}:=(1-\chi)\N_0 +\chi  \N_1 :\Gamma(\E)\rightarrow\Gamma(\E)
	\end{equation}
satisfies the following properties:
\begin{itemize}
\item[(1)] It is normally and 
Green hyperbolic over  $(\M,g_{\chi})$;
\item[(2)] 
It is Green hyperbolic over  $(\M,g_1)$ and, with obvious notation,
$$\Gamma^{g_1}_{pc}(\E) \subset \Gamma^{g_\chi}_{pc}(\E)\:, \quad \Gamma^{g_1}_{fc}(\E) \subset \Gamma^{g_\chi}_{fc}(\E)\:,$$
$$\G_{\N_\chi}^{g_1+} = \G_{\N_\chi}^{g_\chi+}|_{\Gamma^{g_1}_{pc}(\E)}\:, \quad \G_{\N_\chi}^{g_1-} = \G_{\N_\chi}^{g_\chi-}
 |_{\Gamma^{g_1}_{fc}(\E)}\:.$$
\end{itemize}
  In particular, (2) is  true for $\N_0$ by choosing $\chi=0$.
\end{theorem} 

\begin{proof}
	(1) Since $\N_0$ is a normally hyperbolic operator for $(\M, g_0)$ and $\N_1$ is a normally hyperbolic operator for $(\M, g_1)$, by linearity 
	
	\begin{equation*}
	\sigma_2(\N_{\chi},\xi)=(1-\chi)\sigma_2(\N_{0},\xi)+\chi\sigma_2(\N_{1},\xi).
	\end{equation*}
	In particular, we have that $N_\chi$ is normally hyperbolic with respect to $g_\chi$:
	\begin{equation*}
	\sigma_2(\N_{\chi},\xi)=-(1-\chi)g_0^\sharp(\xi,\xi)\Id_\E-\chi g_1^\sharp(\xi,\xi)\Id_\E=-g^\sharp_{\chi}(\xi,\xi)\Id_\E\,.
	\end{equation*}
	By Theorem \ref{thm:metrics1}, the metric $g_\chi$ is globally hyperbolic and, on account of Proposition~\ref{prop:Green} $\N_\chi$ is Green-hyperbolic over $(\M, g_\chi)$. \\
Regarding (2), and referring to the existence of Green operators of $\N_\chi$ in $(\M,g_1)$ we can proceed as follows.   Observe that,
since  $g_\chi \preceq g_1$, we have 
 $J_\pm^{g_\chi}(A) \subset J_\pm^{g_1}(A)$  and, with obvious notation,
$\Gamma^{g_1}_{pc}(\E) \subset \Gamma^{g_\chi}_{pc}(\E)$ together with $\Gamma^{g_1}_{fc}(\E) \subset \Gamma^{g_\chi}_{fc}(\E)$,
in view of (iii) (2) Lemma \ref{prop:metrics0}.  As a consequence, the Green operators of $\N_\chi$ with respect to $(\M,g_\chi)$
are also Green operators with respect to $(\M,g_1)$. Finally we pass to the existence of  the Green operators of $\N_\chi^*$  -- {\em where  $^*$ is here referred to the volume form of $g_1$ and not $g_\chi$} -- in $(\M, g_1)$.  Since $\N^*_\chi$ has the same principal symbol $g^\sharp_\chi(\xi,\xi)\Id_\E$ as $\N_\chi$ it  is normally hyperbolic in $(\M,g_\chi)$ and hence Green hyperbolic thereon. With the same argument used above, we see that the Green operators of $\N^*_\chi$ (with $^*$ always referred to $g_1$) in $(\M,g_\chi)$ are also Green operators in $(\M,g_1)$.
\end{proof}

\begin{remark} We stress that, when $g_0\preceq g_1$ are globally hyperbolic, 
$\N_{\chi}$ and $\N_0$ are therefore Green-hyperbolic second-order differential operators on $(\M,g_1)$ though  they are {\em not}  normally hyperbolic thereon. These are examples of {\em second-order} linear differential operators which are Green hyperbolic but {\em not} normally hyperbolic in a given globally hyperbolic spacetime.
\end{remark}

\section{M{\o}ller maps and operators for normally hyperbolic operators}\label{sec:Moller}

We are  in the position to introduce the notion of  so-called {\em M\o ller map}, which we shall later specialize to the case of a {\em M\o ller operator}, namely a (geometric) map which compares the space of solutions of different normally hyperbolic operators. The novelty of this approach consists in defining the notion of  M\o ller map in a more general fashion. More in detail, in~\cite{FPMoller,Moller,DefArg1,MollerMIT} the M\o ller operator was constructed once that a foliation of $\M$ in Cauchy hypersurfaces was assigned and referring to the family of the  metrics which are decomposed as in (\ref{GHmetric}) with respect to {\em that} foliation. Here we shall see, that the construction  of a M\o ller map still requires the choice of a foliation (associated to some smooth Cauchy time function), but the involved metrics do not have any particular relationship with the choice of the  foliation.  Instead they should  enjoy  some interplay  concerning their light-cone structures which generalizes $g\preceq g'$ in the sense of paracausal deformations.

\subsection{General approach to construct M{\o}ller maps when $g_0 \preceq g_1$} \label{subsec:moller 1}
Let us consider a globally hyperbolic spacetime $(\M,g)$ equipped with a  vector bundle
$\E \to \M$ as before.   If $\P : \Gamma(\E) \to \Gamma(\E)$ is a linear differential operator, a family of physically relevant solutions of the inhomogeneous equation $\P\ff= \fh$ is the  linear vector space of spacelike compact smooth solutions with compactly supported source:
\begin{equation*}
\Sol^g_{sc, c}(\P) :=  \{\ff \in \Gamma^g_{sc}(\E) \:|\: \P\ff \in \Gamma_c(\E)\}\:.
\end{equation*}
Its  subspace corresponding  to the  solutions of the homogeneous equation $\P\ff= 0$ is denoted by
\begin{eqnarray*}
\Ker^g_{sc}(\P) :=\{\ff \in \Gamma^g_{sc}(\E) \:|\: \P\ff =0\}
\end{eqnarray*}
and it will play a pivotal role in the formulation of linear QFT.

We now specialize the operators  $\P$ to 2nd-order normally-hyperbolic linear operators 
$\N_1,\N_0, \N_\chi$ (\ref{def:Nchi})
over $\Gamma(\E)$  associated to globally hyperbolic metrics $g_0\preceq g_1$ and $g_\chi$ on the  common spacetime manifold $\M$. Our goal is to construct several families of {\em M{\o}ller maps}, namely linear operators such that
\begin{itemize}
 \item[(a)] they are linear space isomorphisms between  $\Sol^{g_0}_{sc, c}(\N_0)$,  $\Sol^{g_1}_{sc, c}(\N_1)$,  $\Sol^{g_\chi}_{sc, c}(\N_\chi)$; 
 \item[(b)] they  restrict to isomorphisms to the subspaces 
$\Ker^{g_0}_{sc}(\N_0)$,  $\Ker^{g_1}_{sc}(\N_1)$,  $\Ker^{g_\chi}_{sc}(\N_\chi)$.
\end{itemize}
For later convenience, we shall additionally require that the M\o ller maps preserve also the symplectic forms, which are of interest in applications to linear QFT.\medskip

The overall idea is inspired by the scattering theory.  We start  with  two ``free theories'',  described  by the space of solutions of normally hyperbolic operators $\N_0$ and $\N_1$ in corresponding  spacetimes $(\M,g_0)$ and $(\M,g_1)$, respectively,  and we intend  to connect them through  an ``interaction spacetime'' $(\M,g_\chi)$ with a ``temporally localized'' interaction defined by interpolating  the two metrics by means of a smoothing function $\chi$. Here we need two M{\o}ller maps: $\Omega_+$ connecting $(\M,g_0)$ and $(\M, g_\chi)$ --
which reduces to the identity in the past when $\chi$ is switched off --
and a second M{\o}ller map connecting $(\M,g_\chi)$ to $(\M,g_1)$ -- which reduces to the identity in the future when $\chi$ constantly takes the value $1$. The ``$S$-matrix'' given by the composition 
$\S :=\Omega_-\Omega_+$
 will be the M\o ller map connecting $\N_0$ and $\N_1$.

\subsection{M{\o}ller maps for metrics satisfying $g_0\preceq g_1$}\label{subsec:moller 2}
The first step   consists of comparing  $\N_0$ and $\N_1$ with $\N_\chi$ separately to construct the  M{\o}ller map.
As usual, we denote with $\E$ the $\KK$-vector bundle over a spacetime $(\M,g)$.

We first   start with operators denoted by $\R_\pm$ defined on the whole space of smooth sections $\Gamma(\E)$ which is in common for the three metrics on $\M$ and next we will  restrict these operators to the special spaces of solutions with spatially compact support and compactly supported sources, proving that  these restrictions $\Omega_\pm$ are still linear space isomorphisms.

\begin{proposition} \label{propositionA} Let $g_0,g_1\in\mathcal{GM}_\M$ be such that  $g_0\preceq g_1$ and
$V^{g_0+}_x\subset 
V^{g_1+}_x$ for all $x\in M$.
Let $\E$ be a vector bundle over $\M$ and  $\N_0, \N_1:\Gamma(\E)\to\Gamma(\E)$	
be normally hyperbolic operators associated to  $g_0$ and $g_1$ respectively.  Choose
\begin{itemize}
\item[(a)]  a smooth  Cauchy time $g_1$-function $t: \M \to \mathbb{R}$ and $\chi \in C^\infty(\M; [0,1])$ such that $\chi(p)=0$ if $t(p) <t_0$ and $\chi(p)=1$ if $t(p)>t_1$ for given $t_0< t_1$;
\item[(b)]
a pair of smooth functions $\rho, \rho' : \M \to (0,+\infty)$  such that $\rho(p) =1$ for $t(p) < t_0$ and $\rho'(p) = \rho(p) =1$ if $t(p)>t_1$.
(Notice that $\rho=\rho'=1$ constantly is allowed.)
\end{itemize}
The following facts are true.
\begin{itemize}
\item[(1)] The operators
\begin{align}
  \R_+&=\Id - \G_{\rho \N_\chi}^+(\rho \N_\chi -\N_0): \Gamma(\E) \to \Gamma(\E)  \label{defR+} \\
  \R_-&=\Id - \G_{\rho' \N_1}^-  (\rho'\N_1 - \rho\N_\chi) : \Gamma(\E) \to \Gamma(\E) \label{defR-}
  \end{align}
are linear space isomorphisms, whose inverses are given by
\begin{align}\label{inverse}
\R_+^{-1}&=\Id+\G_{\N_0}^+(\rho \N_\chi-\N_0):\Gamma(\E)\to \Gamma(\E)\\
  \R_-^{-1}&=\Id+\G_{\rho \N_\chi}^-(\rho'\N_1-\rho\N_\chi):\Gamma(\E)\to\Gamma(\E) \label{inverse2}\:.
   \end{align}
\item[(2)] It holds 
\begin{equation} \label{eqRNN}\rho \N_\chi \R_+ = \N_0\: \qquad \text{ and }  \:\qquad \rho' \N_1 \R_- = \rho \N_\chi\,.
 \end{equation}
\item[(3)]  If $\ff \in \Gamma(\E)$, then 
\begin{align}
(\R_+\ff)(p) &= \ff(p) \quad \mbox{for} \quad t(p) < t_0,   \label{pastequal}\\
(\R_-\ff)(p) &= \ff(p) \quad \mbox{for} \quad t(p) > t_1 \label{futureequal}\:.
   \end{align}
\end{itemize}
\end{proposition}

\begin{proof}  Observe that $\rho \N_\chi$ and $\rho' \N_1$ are Green hyperbolic  with respect to $g_\chi$ (as in Theorem \ref{lem:DefNorm1}) and $g_1$ respectively
according to Theorem \ref{lem:DefNorm1} and \ref{propprodgreen}, and thus they are with respect to $g_1$.
 Moreover $\G_{\rho \N_\chi}^\pm =\G_{\N_\chi}^\pm \rho^{-1} $ and $\G_{\rho' \N_1}^\pm =\G_{\N_1}^\pm \rho'^{-1} $.\\
 (1)  If $\ff \in \Gamma(\E)$,
in view of the hypotheses $((\rho \N_\chi-\N_0)\ff)(p) =0$ and $((\N_1-\N_\chi)\ff)(p)=0$ is respectively $t(p) <t_0$ and $t(p)>t_1$ where $t^{-1}(t_0)$ and $t^{-1}(t_1)$ are spacelike Cauchy hypersurfaces in common for the metrics $g_0,g_\chi, g_1$.
Therefore 
the operators $\R_-$ and $\R_+$ are linear and well defined on the domain $\Gamma(\E)$ because $(\rho \N_\chi-\N_0)\ff\in \Gamma^{g_1}_{pc}(\E) \subset \Gamma^{g_\chi}_{pc}(\E) \subset \D(\G_{\rho \N_\chi}^+)$ and  $(\rho' \N_1-\rho \N_\chi)\ff\in \Gamma^{g_1}_{fc}(\E) \subset \D(\G_{\rho' \N_1}^-)$. A similar argument holds for $\R_\pm^{-1}$.
 To prove bijectivity of $\R_\pm$  it suffices to establish  that $\R_-^{-1}$ in (\ref{inverse2})  is a two-sided  inverse of $\R_-$ and that  $\R_+^{-1}$ in (\ref{inverse}) is a two-sided inverse of $\R_+$ on  $\Gamma(\E)$:
  \begin{equation}
  \R_-\circ \R_-^{-1}=\R_-^{-1}\circ\R_-=\Id \quad \mbox{and}\quad 
  \R_+\circ \R_+^{-1}=\R_+^{-1}\circ\R_+=\Id.\nonumber
  \end{equation}
The proof of the well definiteness of $\R_-^{-1}$ and $\R_+^{-1}$ on $\Gamma(\E)$ is analogous to the previous one for $\R_-$.
  We prove that $\R_-$ defined as in (\ref{inverse2})  inverts $\R_-$ from the right by direct computation:
  \begin{equation}
  \begin{gathered}\nonumber
  \R_{-}\circ \R_-^{-1}=(\Id - \G_{\rho' \N_1}^-  (\rho'\N_1 - \rho\N_\chi))\circ(\Id+\G_{\rho \N_\chi}^-(\rho'\N_1-\rho\N_\chi))=\\
  =\Id-\G_{\rho' \N_1}^-  (\rho'\N_1 - \rho\N_\chi) +\G_{\rho \N_\chi}^-(\rho'\N_1-\rho\N_\chi)
- \G_{\rho' \N_1}^-  (\rho'\N_1 - \rho\N_\chi)  \G_{\rho \N_\chi}^-(\rho'\N_1-\rho\N_\chi).
  \end{gathered}
  \end{equation}
  Now, by exploiting the identity
  \begin{equation}\nonumber 
  \begin{gathered}
  \G_{\rho' \N_1}^-  (\rho'\N_1 - \rho\N_\chi)  \G_{\rho \N_\chi}^-=\G_{\rho \N_\chi}^--\G_{\rho' \N_1}^- :\Gamma_{fc}^{g_\chi}(\E)\cap \Gamma_{fc}^{g_1}(\E)\rightarrow\Gamma(\E),
  \end{gathered}
  \end{equation}
  we can prove our claim
  \begin{equation}\nonumber
  \begin{gathered}
  \R_{-}\circ \R_-^{-1}=\Id-\G_{\rho' \N_1}^-  (\rho'\N_1 - \rho\N_\chi) +\G_{\rho \N_\chi}^-(\rho'\N_1-\rho\N_\chi)
- (\G_{\rho \N_\chi}^--\G_{\rho' \N_1}^-)(\rho'\N_1-\rho\N_\chi)  = \Id\:.
  \end{gathered}
  \end{equation}
  The proof that $\R_-^{-1}$ is also a left inverse is  the same with obvious changes  and analogous calculations show that $\R_+^{-1}$ is a left and right inverse of $\R^{+}$.\\
(2) Taking advantage of (ia)-(iib) in  Definition \ref{def:GreenHyp}  and the  definition of $\N_\chi$ and the one of $\R_\pm$, a direct computation establishes (\ref{eqRNN}).\\
(3)  Let us prove (\ref{pastequal}).  Consider a compactly supported  smooth section $\fh$ whose support is included in the set $t^{-1}((-\infty, t_0))$. Taking advantage of the former in (\ref{eqGstar}), we obtain
$$\int_\M \langle \fh, \G_{\rho \N_\chi}^+ (\rho \N_\chi-\N_0) \ff \rangle \: \mbox{vol}_{g_\chi} = 
\int_\M \langle \G_{(\rho \N_\chi)^*}^- \fh,   (\rho \N_\chi-\N_0) \ff \rangle \: \mbox{vol}_{g_\chi} =0$$
since $\mbox{supp}( \G_{(\rho \N_\chi)^*}^-\fh) \subset J^{g_\chi}_-(\mbox{supp}(\fh))$ from Definition \ref{def:GreenHyp}  and thus that support  does not meet $\mbox{supp}((\rho \N_\chi-\N_0) \ff)$ because 
$((\rho \N_\chi-\N_0) \ff)(p)$ vanishes if $t(p)< t_0$.   As  $\fh$ is an arbitrary smooth section compactly  supported in $t^{-1}((-\infty,t_0))$,
$$\int_\M \langle \fh, \G_{\rho \N_\chi}^+ (\rho \N_\chi-\N_0) \ff \rangle \: \mbox{vol}_{g_\chi} = 0$$
 entails  that $\G_{\rho \N_\chi}^+ (\rho \N_\chi-\N_0) \ff =0$ on $t^{-1}((-\infty,t_0))$. Eventually, the very definition (\ref{defR+}) of $\G_{\rho \N_\chi}^+$ implies 
 (\ref{pastequal}). The proof of (\ref{futureequal}) is strictly analogous, so we leave it to the reader.
\end{proof}
We can now pass to the second step, namely we perform  restrictions of $\R_\pm$ to the relevant subspaces of solutions. 
\begin{proposition}\label{propositionB}
With the same hypotheses as in Proposition \ref{propositionA} (in particular $\chi(p)=0$ if $t(p) <t_0$ and $\chi(p)=1$ if $t(p)>t_1$ for given $t_0< t_1$), we have
\begin{equation}\label{former} \R_+(\Sol^{g_0}_{sc,c}(\N_0)) = \Sol^{g_\chi}_{sc,c}(\N_\chi)\quad \mbox{and}\quad \R_-(\Sol^{g_\chi}_{sc,c}(\N_\chi)) = \Sol^{g_1}_{sc,c}(\N_1) \end{equation}
and 
\begin{equation} \label{latter}\R_+(\Ker^{g_0}_{sc}(\N_0)) = \Ker^{g_\chi}_{sc}(\N_\chi)\quad \mbox{and}\quad \R_-(\Ker^{g_\chi}_{sc}(\N_\chi))= \Ker^{g_1}_{sc}(\N_1) \:.\end{equation}
As a consequence, 
 the restrictions
  \begin{align*} 
  \Omega_+ &:= \R_+|_{\Sol^{g_0}_{sc,c}(\N_0)}
: \Sol^{g_0}_{sc,c}(\N_0) \to \Sol^{g_\chi}_{sc,c}(\N_\chi)\:, \quad \Omega^0_+ :
 \R_+|_{\Ker^{g_0}_{sc}(\N_0)}
: \Ker^{g_0}_{sc}(\N_0) \to \Ker^{g_\chi}_{sc}(\N_\chi)\:,\\ 
\Omega_- &:= \R_-|_{\Sol^{g_\chi}_{sc,c}(\N_\chi)}
: \Sol^{g_\chi}_{sc,c}(\N_\chi) \to \Sol^{g_1}_{sc,c}(\N_1)\:, \quad \Omega^0_- :
 \R_-|_{\Ker^{g_\chi}_{sc}(\N_\chi)}
: \Ker^{g_\chi}_{sc}(\N_\chi) \to \Ker^{g_1}_{sc}(\N_1)\:,
\end{align*}
 define linear space isomorphisms such that
\begin{equation} \rho \N_\chi \Omega_+ = \N_0\:, \quad \rho'\N_1 \Omega_- = \rho \N_\chi \label{eqONN}
\end{equation}
and,  for $\ff$ in the respective domains,
\begin{align}
(\Omega_+\ff)(p) &= \ff(p) \:, \quad (\Omega^0_+\ff)(p) = \ff(p) \quad \mbox{for} \quad t(p) < t_0,   \label{pastequal2}\\
(\Omega_-\ff)(p) &= \ff(p) \:, \quad  (\Omega^0_-\ff)(p) = \ff(p)  \quad  \mbox{for} \quad t(p) > t_1 \label{futureequal2}\:.
   \end{align}
\end{proposition}

Before we prove our claim, we need a preparatory lemma.
\begin{lemma}\label{lemma} Let $\P : \Gamma(\E) \to \Gamma(\E)$ be a 2nd order normally hyperbolic differential operator on the vector bundle $\E\to \M$ on the globally hyperbolic spacetime $(\M,g)$.  Let $\Psi \in \Gamma(\E)$ be such that  $\P \Psi \in \Gamma_c(\E)$.
Then  the following facts are equivalent.
\begin{itemize}
\item[(a)] $\Psi \in \Gamma^g_{sc}(\E)$;
\item[(b)] there is  a  spacelike Cauchy hypersurface of $(\M,g)$ such that $\Psi$ have compactly supported Cauchy data thereon.
\end{itemize}
\end{lemma}
\begin{proof}  If $\Psi \in \Gamma^g_{sc}(\E)$ then, by definition, (b) is true. Suppose that (b) is true for $\Sigma_0$. According to Theorem \ref{thm:main nh}, 
$\Psi$ is the unique solution of the Cauchy problem  whose equation is $\P\Psi = \ff$, where $\ff \in \Gamma_c(\E)$. As a consequence 
the support of $\Psi$ completely lies in  $J(\mbox{supp}(\ff))
\cup J(\mbox{supp}(\fh_0)) \cup J(\mbox{supp}(\fh_1)) \subset J(K)$ where
$\fh_0$ and $\fh_1$ are the Cauchy data of $\Psi$ on $\Sigma_0$ and $K:= \mbox{supp}(\ff) \cup \mbox{supp}(\fh_0) \cup J(\mbox{supp}(\fh_1)$. In particular 
 $K$ is compact. In view of well known properties of globally hyperbolic spacetimes (see {\it e.g.}~\cite[Proposition 1.2.56]{Ba-lect}), since $K$ is compact 
 $J(K) \cap \Sigma$ is compact  for every  Cauchy hypersurface $\Sigma$ of $(\M,g)$ so that $\Psi \in \Gamma_{sc}^g(\E)$.
\end{proof}

\begin{proof}[Proof of Proposition~\ref{propositionB}]  $\R_{\pm}$ and $\R_{\pm}^{-1}$ are bijective on $\Gamma(\E)$. As a consequence (\ref{former})  and  thesis for $\Omega_\pm$, including (\ref{eqONN}) which is a specialization of (\ref{eqRNN}), immediately arise when proving that \begin{equation}\label{sample}\R_+(\Sol^{g_0}_{sc,c}(\N_0)) \subset  \Sol^{g_\chi}_{sc,c}(\N_\chi)\:, \quad 
\R^{-1}_+(\Sol^{g_\chi}_{sc,c}(\N_\chi)) \subset  \Sol^{g_0}_{sc,c}(\N_0)\end{equation} and 
$$\R_-(\Sol^{g_\chi}_{sc,c}(\N_\chi)) \subset \Sol^{g_1}_{sc,c}(\N_1)\:, \quad \R^{-1}_-(\Sol^{g_1}_{sc,c}(\N_1)) \subset \Sol^{g_\chi}_{sc,c}(\N_\chi)$$
The identities in  (\ref{latter})  and the  thesis for $\Omega^0_\pm$ immediately arise from the bijectivity of the linear maps  $\Omega_\pm$ and (\ref{eqONN}) where we know that $\rho,\rho' >0$.
To conclude, let us establish the first inclusion  in (\ref{sample}), the remaining three inclusions have a strictly analogous proof.  Suppose that $\ff \in \Sol^{g_0}_{sc,c}(\N_0)$. Hence $\rho \N_\chi \R_+ \ff = \N_0 \ff \in \Gamma^{g_0}_c(\E) = \Gamma^{g_\chi}_c(\E)$ and 
$\N_\chi \R_+ \ff = \rho^{-1}\N_0 \ff \in  \Gamma^{g_\chi}_c(\E)$. Next pass to consider the Cauchy hypersurfaces of $t$ 
which are in common with the three considered metrics $g_0,g_1,g_\chi$ and 
choose $t' < t_0$. (3) in Proposition \ref{propositionA} yields
$(\R_+\ff)(t',x) = \ff(t',x)$ where $x\in \Sigma_{t'}$. The Cauchy data of $\ff$ on $\Sigma_{t'}$ have compact support because 
$\ff \in \Sol^{g_0}_{sc,c}(\N_0)$. On the ground of Lemma~\ref{lemma}, noticing that $\N_\chi$ is normally hyperbolic in $(\M,g_\chi)$, referring to the Cauchy problem on $\Sigma_{t'}$ for the equation $\N_\chi \R_+ \ff =\rho^{-1}\N_0\ff \in  \Gamma^{g_\chi}_c(\E)$ in the spacetime $(\M, g_\chi)$, we conclude that $\R_+\ff \in \Gamma_{sc,c}^{g_\chi}(\E)$ because its Cauchy data on $\Sigma_{t'}$ (now interpreted as a Cauchy hypersurface for $g_\chi)$ have compact support as they coincide with the ones of $\ff$ itself.
\end{proof}

\subsection{General  M{\o}ller maps  for paracausally  related metrics}\label{subsec:Moller general}

We are now in a position to state a result  regarding the existence of M{\o}ller maps between two normally hyperbolic operators $\N_0$ and $\N_1$ on respective globally hyperbolic spacetimes over the same manifold (and vector bundle)
 whose metrics are  $\preceq$ comparable. The final goal is to extend the results to pairs of paracausally related metrics.

\begin{proposition}\label{theoremS1}
Let $g_0,g_1\in\mathcal{GM}_\M$ be such that  either $g_0\preceq g_1$ or $g_1\preceq g_0$  with, respectively, either $V^{g_0+}_x\subset 
V^{g_1+}_x$ for all $x\in M$ or $V^{g_1+}_x\subset 
V^{g_0+}_x$ for all $x\in M$.
Let $\E$ be a vector bundle over $\M$ and  $\N_0, \N_1:\Gamma(\E)\to\Gamma(\E)$	
be normally hyperbolic operators associated  to $g_0$ and $g_1$ respectively.
 There exist (infinitely many) vector space isomorphisms,
$$\S  :\Sol^{g_0}_{sc,c}(\N_0)\to \Sol^{g_1}_{sc,c}(\N_1)$$ 
such that, for some smooth function 
$\mu: \M \to (0,+\infty)$ depending on $\S$ (which can be chosen $\mu=1$),
\begin{itemize}
\item[(1)] referring to the said domains,
$$\mu  \N_1\S = \N_0 \quad  \mbox{and}\quad \mu^{-1}\N_0\S^{-1} = \N_1 $$

\item[(2)]  the restriction  $\S^0 := \S|_{\Ker^{g_0}_{sc}(\N_0)}$  defines a vector space isomorphism
$$\S^0 : \Ker^{g_0}_{sc}(\N_0) \to \Ker^{g_1}_{sc}(\N_1)\:.$$
\end{itemize}

\end{proposition}
\begin{proof} First consider the case $g_0\preceq g_1$.
Referring to a smooth Cauchy time function $t$ of $(\M,g_1)$ and a smoothing function $\chi$, $\S:= \Omega_-\Omega_+$ constructed as in Proposition \ref{propositionB} satisfies all the requirements trivially for $\mu:= \rho'$.
 The previous result is also valid for $g_1\preceq g_0$. It is sufficient to construct  $\Omega_\pm$ 
as in  Proposition \ref{propositionB}, {\em but using $g_1$ as the initial metric and $g_0$ as the final one}, and eventually 
 defining $\mu := \rho^{-1}$, $\S :=  (\Omega_-\Omega_+)^{-1} = \Omega_+^{-1}\Omega_-^{-1}$, and  $\S^0 :=  (\Omega^0_-\Omega^0_+)^{-1} = (\Omega^0_+)^{-1} (\Omega^0_-)^{-1}$.
\end{proof}

We  can pass to the generic case $g\simeq g'$, obtaining the first main result of this work.

\begin{theorem}\label{remchain}
Let $(\M,g)$ and $(\M,g')$ be globally hyperbolic spacetimes,   $\E$ a vector bundle over $\M$ and  $\N, \N':\Gamma(\E)\to\Gamma(\E)$	
 normally hyperbolic operators associated  to $g$ and $g'$ respectively.\\
If  $g \simeq g'$,
 then there exist (infinitely many) vector space isomorphisms,
called {\bf M\o ller maps} of $g,g'$ (with this order),
 \begin{equation*}
 \S  :\Sol^{g}_{sc,c}(\N)\to \Sol^{g'}_{sc,c}(\N')
  \end{equation*}
 such that 
\begin{itemize}
\item[(1)] referring to the said domains,
$$\mu \N'\S = \N $$
for some smooth $\mu: \M \to (0,+\infty)$ (which can be always taken $\mu=1$ constantly in particular),  
\item[(2)]   the restriction $\S^0 := \S|_{\Ker^{g}_{sc}(\N)}$ (also called {\bf M\o ller map} of $g',g'$) defines a vector space isomorphism
$$\S^0 : \Ker^{g}_{sc}(\N) \to \Ker^{g'}_{sc}(\N')\:.$$
\end{itemize}
\end{theorem}

\begin{proof} First of all we notice that there always exists a normally hyperbolic operator  $\N$ on $\E$ associated to every $g\in \mathcal{GM}_M$: For instance the {\em connection-d’Alembert operator} in  \cite[Example 2.1.5]{Ba-lect} referred to a generic connection $\nabla$ on $\E$, which always exists, and the Levi-Civita connection on $(\M,g)$.
 Let us consider a sequence 
$g_0=g, g_1,\ldots, g_N=g'$
of globally hyperbolic metrics on $\M$ satisfying Definition \ref{def:paracausal def} and a corresponding sequence 
of formally selfadjoint normally hyperbolic operators $\N_k$ with $\N_0:= \N$ and $\N_{N}:= \N'$. We can apply Proposition \ref{theoremS1} for each pair $g_k, g_{k+1}$ for $k=0,1,\ldots, N-1$. It turns immediately out that, with an obvious notation,
\begin{equation*}
\S := \S_0\S_1\cdots \S_{N-1} \:, \quad \mu := \mu_0\cdots \mu_{N-1}\:, \quad \mbox{where} \quad \mu_k \N_k \S_k = \S_{k-1}\quad k=0,\ldots N-1\:.
\end{equation*}
satisfy the thesis of the theorem, where
either $S_k := \Omega_{k-} \Omega_{k+}$, $\mu_k := \rho_k$ or  $S_k := (\Omega_{k+})^{-1}(\Omega_{k-})^{-1}$, 
$\mu_k := \rho^{-1}_k$ according to $g_k \preceq g_{k+1}$ or $g_{k+1} \preceq g_{k}$ respectively.  With the same 
convention it results that $\S^0 = \S^0_0\S'_1\cdots \S^0_{N-1} $ where  either
$S^0_k = \Omega^0_{k-} \Omega^0_{k+}$ or  $S_k = (\Omega^0_{k+})^{-1}(\Omega^0_{k-})^{-1}$ according to the discussed cases.
\end{proof}

\subsection{Preservation of symplectic forms}\label{subsec:moller sympl form}

The M{\o}ller maps   $\S^0$ as  in Theorem  \ref{remchain} preserve the symplectic forms of the normal operators they relate when these operators are formally selfadjoint.

\begin{theorem} \label{thm:pres sympl}
Consider $g,g' \in \mathcal{GH}_M$ with respectively  associated normally hyperbolic operators $\N$, $\N'$ on the $\KK$-vector bundle $\E$ over $\M$.
 If $g'\simeq g$ and  $\N$ and $\N'$ are formally selfadjoint with respect to a non-degenerate,  Hermitian fiber metric $\fiber{\cdot}{\cdot}$, then  there are M\o ller maps $S^0$ satisfying the thesis  of Theorem  \ref{remchain} such that
$$\sigma^{\N'}_{g'}(\S^0\Psi,\S^0 \Phi) = \sigma^{\N}_{g}(\Psi,\Phi) \quad \mbox{for every $\Psi,\Phi \in \Ker^{g}_{sc}(\N)$,}$$
where we used  the notation  $\sigma_g^\N$ in place of $\sigma^\N_{(\M,g)}$.
\end{theorem}
\begin{proof} It is sufficient to prove the thesis for the maps $\Omega^0_\pm$ referred to two metrics $g_0\preceq g_1$, which immediately implies the thesis also for the inverse maps
$(\Omega^0_\pm)^{-1}$  they being isomorphisms. Indeed,  according  the proof of Theorem \ref{remchain},
 the isomorphisms $\S^0$ are compositions of various copies of $\Omega^0_\pm$  and their inverses. Let us consider $\Omega^0_+ :\Ker_{sc}(\N_0) \to \Ker_{sc}(\N_\chi)$ and we prove the thesis for it, the other case being very similar.  Consider a smooth  Cauchy time function $t$  for $g_1$ and the associated foliation made of spacelike Cauchy hypersurfaces $\Sigma_t$ in common for $g_0,g_1$, and $g_\chi$.  If the  smoothing function $\chi$ used to build up $g_\chi$ and $\N_\chi$ vanishes before $t_0$ and we use $\Sigma_t$ with $t<t_0$ to compute the relevant symplectic forms,   due to 
(\ref{pastequal2}),
$$\sigma^{\N_\chi}_{g_\chi}(\Omega_+^0\Psi,\Omega_+^0 \Phi)   = \sigma^{\N_0}_{g_0}(\Psi,\Phi) \quad \mbox{for every $\Psi,\Phi \in \Ker^{g_0}_{sc}(\N_0)$.}$$
Above, we have used the definition of the symplectic form, we have noticed that $g_\chi=g_0$ around $\Sigma_t$  and that  the $\N_0$ and $\N_\chi$ compatible connections must coincide there as they are locally defined and uniquely  determined  by $\N_0 \Psi = \N_\chi \Psi =(-\mbox{tr}_g(\nabla\nabla)+ c)\Psi$ for every smooth 
$\Psi$ compactly supported around a point $p$ with $t(p)< t_0$.  Thinking of $\sigma^{\N_\chi}_{g_\chi}(\Omega_+^0\Psi,\Omega_+^0 \Phi)$ as defined in $(\M,g_\chi)$ and of  $\sigma^{\N_0}_{g_0}(\Psi,\Phi) $ as defined in $(\M,g_0)$, though both computed on $\Sigma_t$ with $t< t_0$,  Proposition \ref{lem:indip Sigma} concludes the proof.
\end{proof}

\begin{definition}\label{def:sympl Moll}
We call {\bf symplectic M\o ller map} any linear isomorphism defined in accordance with Theorem~\ref{thm:pres sympl}.
\end{definition}

\subsection{Causal propagators and paracausally related metrics}\label{subsec:moller oper}

In this section, we prove how is possible to choose the functions $\rho$ and $\rho'$ affecting the definitions (\ref{defR+})-(\ref{defR-}) of $\R_\pm$ in order to satisfy a further requirement with some crucial implications in QFT: the preservation of the causal propagator of  two operators $\N$ and $\N'$ when the associated metrics are paracausally related. Essentially speaking, a M\o ller map satisfying this further requirement will be named {\em M\o ller operator}.

\subsubsection{Adjoint operators}
To study the relation between M{\o}ller maps and the causal propagator of normally hyperbolic operators defined on a vector bundle equipped with a  non-degenerate (Hermitian) fiber metric, we need a suitable notion of {\em adjoint operator} which generalizes the notion of formal adjoint of differential operators. 

Let $\E$ be a  $\KK$-vector bundle on the manifold $\M$ equipped with a non-degenerate, symmetric if $\KK=\RR$ or  Hermitian if $\KK=\CC$,  fiber metric $\fiber{\cdot}{\cdot}$. Suppose that $g$ and $g'$ (possibly $g\neq g'$) are Lorentzian metrics on $\M$.
Consider a $\KK$-linear operator 
$$\T : \D(\T) \to \Gamma(\E)\:,$$
where $\D(\T) \subset \Gamma(\E)$ is a $\KK$-linear subspace and $\D(\T) \supset \Gamma_c(\E)$. 
\begin{definition} \label{defadjoint}
An operator $$\T^{\dagger_{gg'}} : \Gamma_c(\E) \to \Gamma_c(\E)$$ is said to be the {\bf adjoint of $\T$ with respect to  $g, g'$} (with the said order) if it satisfies 
\begin{equation*} 
\int_\M \fiber{\fh(x)}{(\T\ff)(x)}  \vol_{g'}(x) = \int_\M \fiber{\left(\T^{\dagger_{gg'}} \fh\right)(x)}{\ff(x)}  \vol_{g}(x)\quad 
\forall \ff \in \D(\T)\:,\: \forall \fh \in \Gamma_c(\E).
\end{equation*}
\end{definition}
\begin{notation}
If $g=g'$ then we shall denote the adjoint of $\T$ with respect to $g$ simply as $\T^{\dagger_g}$.
\end{notation}
We prove below that $\T^{\dagger_{gg'}}$ is unique if exists so that calling it ``the'' adjoint operator of $\T$ is appropriate.

\begin{remark} If $\T: \D(\T) \to \Gamma(\E)$ is defined as in Definition \ref{defadjoint} and 
$\T^{\dagger_{gg'}}$ exists, then  
 $$\int_\M \fiber{\fh}{\T\ff_n}\vol_{g'} \to 0\quad  \mbox{$\forall \fh \in \Gamma_c(\E)$ as $\Gamma_c(\E) \ni \ff_n \to 0$ for $n\to +\infty$ in the {\em topology of test sections} \cite{Ba-lect}}\:.$$
{\em Vice versa}, this only condition is not sufficient to guarantee the existence of $\T^{\dagger_{gg'}}$ as a $\Gamma_c(\E)$-valued operator. 
Using a  straightforward extension of the Schwartz kernel theorem, the condition above just implies the existence of a weaker version of $\T^{\dagger_{gg'}}$ which is distribution-valued.  
\end{remark}

In the rest of the paper if $\T : \D(\T) \to \Gamma(\E)$ and $\T' : \D(\T') \to \Gamma(\E)$, we define the {\bf standard domains} of their compositions as follows, where $a\in \KK$.
\begin{itemize}
\item[(a)] $\D(a\T) := \D(\T)$ -- or $\D(a\T) := \Gamma(\E)$ if $a=0$ -- is the domain of $a\T$ defined pointwise;
\item[(b)]  
$\D(\T+ \T') := \D(\T) \cap \D(\T')$  is the domain of $a\T+b\T'$
defined pointwise;
\item[(c)] $\D(\T'\circ \T) := \{\ff \in \D(\T) \:|\:  \T(\ff) \in \D(\T')\}$ is the domain of $\T' \circ \T$.
\end{itemize}

\begin{proposition}\label{propadjoint}
Referring to the notion of adjoint in Definition \ref{defadjoint}, the following facts are valid.
\begin{itemize}
\item[(1)] If the adjoint $\T^{\dagger_{gg'}}$ of $\T$  exists, then it is unique.
\item[(2)] If $\T : \Gamma(\E) \to \Gamma(\E)$ is a differential operator and $g=g'$, then $\T^{\dagger_{gg}}$ exists and is the restriction of the formal adjoint to $\Gamma_c(\E)$. (In turn, the formal adjoint of $\T^\dagger$ is  the unique extension to $\Gamma(\E)$ of the differential operator $\T^\dagger$ as a differential operator)
\item[(3)] 
Consider a pair of $\KK$-linear operators  $\T : \D(\T) \to \Gamma(\E)$,  $\T' : \D(\T') \to \Gamma(\E)$ 
 and $a,b \in \KK$.  Then 
$$(a\T+b\T')^{\dagger_{gg'}} = \overline{a} \T^{\dagger_{gg'}} + \overline{b} \T'^{\dagger_{gg'}}$$
provided 
$\T^{\dagger_{gg'}}$ and $\T'^{\dagger_{gg'}}$ exist.
\item[(4)] Consider a pair of $\KK$-linear operators  $\T : \D(\T) \to \Gamma(\E)$ and $\T' : \D(\T') \to \Gamma(\E)$  such that 
\begin{itemize}
\item[(i)]  $\D(\T'\circ \T) \supset \Gamma_c(\E)$,
\item[(ii)] $\T^{\dagger_{gg'}}$ and $\T'^{\dagger_{g'g''}}$ exist,
\end{itemize} 
then $(\T'\circ \T)^{\dagger_{gg''}}$  exists and
$$(\T'\circ \T)^{\dagger_{gg''}} = \T^{\dagger_{gg'}} \circ \T'^{\dagger_{g'g''}}\:.$$
\item[(5)] If $\T^{\dagger_{gg'}}$ exists, then $(\T^{\dagger_{gg'}})^{\dagger_{g'g}} = \T|_{\Gamma_c(\E)}$.
\item[(6)] If $\T : \D(T)= \Gamma(\E) \to \Gamma(\E)$ is bijective,  admits $\T^{\dagger_{gg'}}$, and $\T^{-1}$ admits $(\T^{-1})^{\dagger_{g'g}}$, then $\T^{\dagger_{gg'}}$ is bijective and 
$(\T^{-1})^{\dagger_{g'g}} = (\T^{\dagger_{gg'}})^{-1}$.
\end{itemize}
\end{proposition}

\begin{proof}  We write below $\dagger$ in place of $\dagger_{gg'}$ if it is not strictly necessary to specify the metrics. 
To prove (1) let's assume that, fixed an operator $\T:\D(\T)\to\Gamma(\E)$ there exist two different adjoints $\T^\dagger_1,\T^\dagger_2:\Gamma_c(\E)\to\Gamma_c(\E)$ both satisfying definition \ref{defadjoint}, i.e.
$$
\int_\M \fiber{\T^\dagger_1\fh}{\ff} \: \vol_{g} = \int_\M \fiber{\T^\dagger_2 \fh}{\ff} \: \vol_{g}\quad
$$
 for all  $\ff \in \D(\T)$ and all $\fh \in \Gamma_c(\E)$. Then by linearity of the integration and (anti) linearity of the product, the former identity is equivalent to
$$
 \int_\M \fiber{\T^\dagger_1\fh-\T^\dagger_2\fh}{\ff} \: \vol_{g} =0.
$$
 Since  $\Gamma_c(\E)\subset\D(\T)$, the thesis follows by reducing to every fixed local trivialization over every arbitrarily fixed coordinate patch $U$ on $\M$. Restricting to $U$, the equation above can be recast to
$$
 \int_U \sum_{a=1}^N(\T^\dagger_1\fh-\T^\dagger_2\fh)^a(p)  \ff_{a}(p)\: \vol_{g}(p) =0.
$$
where $\ff_{a}(p)$ is a fiber component of $\fiber{\cdot}{\ff}_p \in \E_p^*$ with $p\in U$. Since $U \ni p\mapsto (\T^\dagger_1\fh-\T^\dagger_2\fh)^a(p)$ is continuous and $U \ni p \mapsto  \ff_{a}(p)$ is smooth, compactly supported (with support in $U$) and arbitrary (because $\fiber{\cdot}{\cdot}$ is non-degenerate), the fundamental lemma of calculus of variations implies that $U \ni p\mapsto (\T^\dagger_1\fh-\T^\dagger_2\fh)^a(p)$ is the zero function for $a=1,\ldots, N$. Since $U$ can be fixed as a neighborhood of every point of $\M$, (1) follows.

 The proof of (2) and (3) is obvious: (2) follows by comparing definitions \ref{defadjoint} and \ref{defadjointdiff}, while (3) follows by direct computation checking that $ \overline{a} \T^\dagger + \overline{b} \T'^\dagger$ satisfies the definition of  $(\overline{a} \T + \overline{b} \T')^\dagger$ (notice that $\Gamma_c(\E) \subset \D(\overline{a} \T^\dagger + \overline{b} \T'^\dagger)$ if $\T^\dagger$ and $\T'^\dagger$ exist).
 
 To prove (4), since the composition is well defined on a suitable domain, we can just use twice the definition \ref{defadjoint}
 
 $$\int_\M \fiber{\fh}{\T^\prime\circ\T\ff} \: \vol_{g''} = \int_\M \fiber{\T'^{\dagger_{g'g''}} \fh}{\T\ff} \: \vol_{g'}=\int_\M \fiber{\T^{\dagger_{gg'}}\circ\T'^{\dagger_{g'g''}} \fh}{\ff} \: \vol_{g}
$$ 
for all  $\ff \in \D(\T'\circ\T)$ and all $\fh \in \Gamma_c(\E)$: notice that using the definition of the adjoint in the second equality is possible because $\T'^{\dagger_{g'g''}}:\Gamma_{c}(\E)\to\Gamma_{c}(\E)$. The found identity proves that $\T^{\dagger_{gg'}}\circ\T'^{\dagger_{g'g''}}$ satisfies the definition of $(\T'\circ \T)^{\dagger_{gg''}}$ ending the proof of (4).

(5) is true because, if $\T^{\dagger_{gg'}}:\Gamma_c(\E)\to\Gamma_c(\E)$ exists, then $\T|_{\Gamma_c(\E)}$ satisfies the definition of $(\T^{\dagger_{gg'}})^{\dagger_{g'g}}$.

Finally, (6) arises by taking the $^{\dagger_{gg}}$ adjoint of both sides of the identity  $T\circ T^{-1} = I$ and 
the $^{\dagger_{g'g'}}$ adjoint of both sides of the identity $T^{-1} \circ T = I$ and taking (4) into account.
\end{proof}

\subsubsection{M\o ller operators  and  causal propagators}
We are in a position to state one of the most important results of this work by specializing the isomorphisms   introduced 
in Theorem \ref{remchain} by means of a suitable choice of the function $\mu$. As a matter of fact (1) and (3)  have been already established in Theorem \ref{remchain}.

\begin{theorem} \label{teoRRR} Let  $\E$ be $\KK$-vector bundle over the smooth manifold $\M$ with a non-degenerate, 
real or Hermitian depending on $\KK$,  fiber metric $\fiber{\cdot}{\cdot}$.
 Consider $g,g' \in \mathcal{GH}_M$ with respectively associated normally hyperbolic formally-selfadjoint operators $\N$, $\N'$.\\
If $g\simeq g'$, then  it is  possible to define (in infinite ways) a $\KK$-vector space isomorphism $\R : \Gamma(\E) \to \Gamma(\E)$, called {\bf M\o ller operator} of $g,g'$ (with this order), such that the following facts are true.
\begin{itemize}
\item[(1)] The restrictions to the relevant subspaces of $\Gamma(\E)$ respectively define M\o ller maps (hence linear isomorphisms) 
as in  Theorem \ref{remchain}. $$\R|_{\Sol^{g}_{sc,c}(\N)} = \S  :\Sol^{g}_{sc,c}(\N)\to \Sol^{g'}_{sc,c}(\N')\mbox{\quad and \quad}\R|_{\Ker^{g}_{sc}(\N)} = \S^0 : \Ker^{g}_{sc}(\N) \to \Ker^{g'}_{sc}(\N')\,.$$
\item[(2)]   The causal propagators 
$\G_{\N'}$ and $\G_{\N}$, respectively
of $\N'$ and $\N$, satisfy
\begin{equation} \R \G_{\N}  \R^{\dagger_{gg'}} = \G_{\N'}\:. \label{firstE}\end{equation}
\item[(3)]  By denoting $c'$ the smooth function such that $\vol_{g'} = c' \: \vol_{g}$, we have \begin{equation}  c' \N'\R  = \N\:.\label{secondE}
\end{equation}
\item[(4)] It holds 
\begin{equation*}
\R^{\dagger_{gg'}} \N'|_{\Gamma_c(\E)}= \N|_{\Gamma_c(\E)}\: .
\end{equation*}
\item[(5)] The maps $\R^{\dagger_{gg'}} : \Gamma_c(\E) \to \Gamma_c(\E)$ and $(\R^{\dagger_{gg'}})^{-1} = (\R^{-1})^{\dagger_{g'g}}: \Gamma_c(\E) \to \Gamma_c(\E)$ are continuous with respect to the natural  topologies of $\Gamma_c(\E)$  in the domain and in the co-domain.
\end{itemize}
\end{theorem}
\begin{remarks} Before proving our claim, we want to underline the following:
\begin{itemize} \item[(1)] Any M\o ller operator defines  a symplectic M\o ller map (\cf Definition~\ref{def:sympl Moll}). Indeed, the preservation of the causal propagator (\cf (2) in Theorem~\ref{teoRRR}) implies that the symplectic forms are preserved in view of Proposition \ref{propCPsim}. However, the converse is false since the preservation of the causal propagator relies upon a suitable choice of the function $\rho$, whereas this choice is immaterial for the preservation of the symplectic forms.
\item[(2)] M\o ller operators can be explicitly constructed as follows.
If $g'\simeq  g$, and referring to  a finite  sequence of metrics $g_0 := g, g_1, \ldots, g_N:= g' \in \mathcal{GH}_\M$ as in Definition \ref{def:paracausal def},   then there exists a corresponding sequence of formally selfadjoint $g_k$-normally hyperbolic operators 
$\N_0:= \N, \N_1, \ldots, \N_N := \N' : \Gamma(\E) \to \Gamma(\E)$ such that 
\begin{equation}\label{frop1} \R= \R_0 \cdots \R_{N-1}\:,\end{equation}
is a M\o ller operator of $g,g'$ where
\begin{equation} \label{frop2} \R_k := \R^{(k)}_- \R^{(k)}_+\quad \mbox{if $g_k\preceq g_{k+1}$ \qquad or}\qquad \R_k := (\R^{(k)}_+)^{-1} (\R^{(k)}_-)^{-1}\quad \mbox{if $g_{k+1}\preceq g_{k}$.}
\end{equation}
Above, for every given $k$, $\R^{(k)}_\pm$ are defined as   $\R_\pm$ as in Equations~(\ref{defR+}) and (\ref{defR-}) where 
\begin{itemize} 
\item[(i)]   $\N_0$  is replaced by $\N_k$ and $\N_1$ is replaced by $\N_{k+1}$ if $g_k \preceq g_{k+1}$, \item[(ii)]  $\N_0$ is replaced by $\N_{k+1}$ and $\N_1$ is replaced by $\N_{k}$ if $g_{k+1} \preceq g_{k}$,
\item[(iii)]  $\rho := c_0^\chi$, and  $\rho' := c^1_0$  (assuming $\vol_{g_\chi} = c^\chi_0 \vol_{g_0}$ and $\vol_{g_1} = c^1_\chi \vol_{g_\chi}$).
\end{itemize}
The smooth Cauchy time  function $\chi$  in (\ref{defR+}) and (\ref{defR-}) can be chosen arbitrarily and depending on $k$ in general.
The final M\o ller operator $\R$ of $g,g'$ also  depends on all the made choices.
\end{itemize}
\end{remarks}

\begin{proof}[Proof of Theorem~\ref{teoRRR}] We divide the proof into several steps.

(1)-(3)  Let us first prove the thesis for the special case of $g= g_0 \preceq g_1=g'$, with $V_x^{g_0+} \subset V_x^{g_1+}$ for all $x\in \M$, and specialize the definition of the isomorphisms (\ref{defR+}) and (\ref{defR-}) to
\begin{align}
  \R_+&=\Id - \G_{c^\chi_0 \N_\chi}^+(c^\chi_0 \N_\chi -\N_0): \Gamma(\E) \to \Gamma(\E)  \label{defR+2} \\
  \R_-&=\Id - \G_{c^1_0\N_1}^-  (c^1_0\N_1 - c^\chi_0\N_\chi) : \Gamma(\E) \to \Gamma(\E) \label{defR-2}
  \end{align}
where
\begin{equation*} 
\vol_{g_\chi} = c^\chi_0 \vol_{g_0}\quad \mbox{and}\quad \vol_{g_1} = c^1_0 \vol_{g_0}
\end{equation*}
It is easy to see that 
\begin{equation}  (c^\chi_0 \N_\chi)^{\dagger_{g_0}} = c^\chi_0 \N_\chi \quad \mbox{and}\quad (c^1_0 \N_1)^{\dagger_{g_0}} = c^1_0 \N_1 \:.\label{aggpart}\end{equation} 
Our goal is to prove that the isomorphism $\R := \R_-\R_+: \Gamma(\E) \to \Gamma(\E)$ satisfies the thesis.\\
Per direct inspection, applying the definition of adjoint operator and taking advantage of (\ref{aggpart}), Proposition \ref{propprodgreen}, and (\ref{eqGstar2}), we almost immediately  have that
\begin{equation}\label{RRdagger}\R_+^{\dagger_{g_0}}=\Id - (c^\chi_0 \N_\chi -\N_0)\G_{c^\chi_0 \N_\chi}^-|_{\Gamma_c(\E)}
\quad \mbox{and}\quad  \R_-^{\dagger_{g_0}}=\Id - (c^1_0\N_1 - c^\chi_0\N_\chi) \G_{c^1_0\N_1}^+|_{\Gamma_c(\E)}\:.\end{equation}
Again per direct inspection we see that
$$c^\chi_0\N_\chi \R_+ = \N_0 \quad \mbox{and} \quad c_0^1 \N_1 \R_- = c_0^\chi \N_\chi$$
and thus
$$c_0^1 \N_1 \R = c_0^1 \N_1 \R_-\R_+ = \N_0$$
as wanted.\\
As we prove below, the following identities are valid
\begin{equation}   \R_+ \G_{\N_0} \R_+^{\dagger_{g_0}}=   \G_{c^\chi_0\N_\chi}\quad \mbox{and}\quad \R_- \G_{c^\chi_0\N_\chi} \R_-^{\dagger_{g_0}}
=\G_{c^1_0\N_1}= \G_{\N_1} (c^1_0)^{-1}\label{interm}\end{equation}  
so that 
$$\R_- \R_+ \G_{\N_0} R_+^{\dagger_{g_0}} R_-^{\dagger_{g_0}} = \G_{\N_1}(c^1_0)^{-1}$$
which is equivalent to
$$\R_- \R_+ \G_{\N_0} (\R_-\R_+)^{\dagger_{g_0}} c^1_0 = \G_{\N_1}\:.$$
On the other hand, we have
$$A^{\dagger_{g_0}} c^1_0 = A^{\dagger_{g_0g_1}} $$ so that
$$\R\G_{\N_0} \R^{\dagger_{g_0g_1}}  = \R_- \R_+ \G_{\N_0} (\R_-\R_+)^{\dagger_{g_0g_1}} = \G_{\N_1}\:.$$
To conclude the proof  of (1)-(3) for the case $g=g_0\preceq g_1= g'$ we prove (\ref{interm}).\\
Since $G_{\N_0}$ is defined as the difference of the advanced and retarded Green operators restricted to compact sections, we perform the computation separately for the two operators.\\
We start from $\R_+\G_{\N_0}|_{\Gamma_{c}(\E)}\R_+^{\dagger_{g_0}}$: the adjoint of the M{\o}ller operator is defined over $\Gamma_c(\E)$ and gives back compactly supported sections, then the advanced Green operator maps a compactly supported section $\ff\in\Gamma_{c}(\E)$ to a solution such that $\supp(\G^+_{N_0}\ff)\subset J^+_0(\supp(\ff)\subset J^+_\chi(\supp(\ff))$, where the last inclusion is due to the crucial hypothesis $g_0\preceq g_\chi\preceq g_1$. Now since $\supp(\ff)$ is compact the smooth Cauchy time function $t$ attains a minimum $t_0\in\RR$ therein, so we choose a common smooth Cauchy hypersurface $\Sigma_{t_1}$ of the foliation induced by $t$ such that $t_1<t_0$ and deduce that $\supp(\G^+_{N_0}\ff)\subset J^+_\chi(\supp(\ff))\subset J^+_\chi(\Sigma_{t_1})$ which implies by \cite[Lemma 1.2.61]{Ba-lect}  that $\G^+_{N_0}\ff\in\Gamma_{pc}^{\chi}(\E)$.\\
Omitting the restriction of the domain of the causal propagators  from the notation for sake of clarity, but having in mind that it is crucial for the validity of the argument, we obtain:
 $$\R_+\G^+_{\N_0}=\G^+_{\N_0}-G^+_{c_0^\chi\N_\chi}c_0^{\chi}\N_\chi\G^+_{\N_0}+G^+_{c_0^\chi\N_\chi}\N_0\G^+_{\N_0}=\G^+_{c_0^\chi\N_\chi}.$$ 
 A similar reasoning proves that
 $$G^-_{\N_0}\R_+^{\dagger_{g_0}}=\G^-_{c_0^\chi\N_\chi}.$$
 where now the restriction of the domains of the causal propagators to compactly supported sections is assumed from the definition of the adjoint.
 Collecting together the two identities found, we have
$$\R_+\G_{\N_0}\R_+^{\dagger_{g_0}} = (\R_+\G^+_{\N_0} - G^-_{\N_0}\R_+^{\dagger_{g_0}})
+ M = \G_{c_0^\chi\N_\chi} + M\:,$$
with, where both sides have to be computed on compactly supported sections,
$$M:= (\Id-\R_+) \G_{\N_0}^- \R_+^{\dagger_{g_0}}- \R_+ \G_{\N_0}^+ (\Id -  \R_+^{\dagger_{g_0}})\:.$$
A lengthy direct evaluation of $M$ using (\ref{defR+2}) and the former in (\ref{RRdagger}) shows that $M=0$.
All that establishes the first identity in (\ref{interm}), while the latter follows by almost identical facts.\\
Let us pass to prove (1)-(3) for  the case $g_1 \preceq g_0$, with $V_x^{g_1+} \subset V_x^{g_0+}$ for all $x\in \M$.    First of all we observe that from the previously treated case ($g_0\preceq g_1$) we have
$c^0_1 \N_0 \R^{-1} = \N_1$ where $c^0_1 = (c_1^0)^{-1}$ and $\vol_{g_0} = c^0_1 \vol_{g_1}$.  Interchanging the names of $g_0$ and $g_1$, this result implies that  (\ref{secondE}) is true for  
$g_1 \preceq g_0$ when using $\R^{-1}$ in place of $\R$.  An analogous procedure proves  (\ref{firstE}) for the case $g_1 \preceq g_0$
from the same equation, already established, valid when $g_0\preceq g_1$. Also in this case the relevant M\o ller operator is $\R^{-1}$. To this end,  we  have only to prove that $(\R^{-1})^{\dagger_{g_1g_0}}$ exists and coincides to $(\R^{\dagger_{g_0g_1}})^{-1}$.
Indeed, under these assumptions (\ref{firstE}) implies
$$\N_1 = \R^{-1} \N_0 (\R^{\dagger_{g_0g_1}})^{-1}=  \R^{-1} \N_0 (\R^{-1})^{\dagger_{g_1g_0}}$$
which is our thesis when interchanging $g_0$ and $g_1$.
 This  fact that $(\R^{-1})^{\dagger_{g_1g_0}}= (\R^{\dagger_{g_0g_1}})^{-1}$
 actually  can be established exploiting   (6) in \ref{propadjoint}: $\R$ is bijective over $\Gamma(\E)$, and admits the adjoint $\R^{\dagger_{g_0 g_1}}$, so if the inverse $\R^{-1}$ admits the adjoint $(\R^{-1})^{\dagger_{g_1 g_0}}$, then $\R^{\dagger_{g_0 g_1}}$ is bijective and its inverse is such that $(\R^{\dagger_{g_0 g_1}})^{-1}=(\R^{-1})^{\dagger_{g_1 g_0}}$.
Let us prove that $\R^{-1}$ admits adjoint (with respect to any metric among $g_0,g_\chi, g_1$ since the existence of the adjoint with respect to one of them trivially implies the existence of the adjoint with respect to the other metrics)  to end the proof for the case $g_1\preceq g_0$.
 By recalling that $\R^{-1}=\R_+^{-1}\circ\R_-^{-1}$ it suffices to show that $\R_+^{-1}$ and $\R_-^{-1}$ both admit adjoints. We explicitly give the $g_0$-adjoint of $\R_+^{-1}$ the other case being analogous,
 $$(\R_+^{-1})^{\dagger_{g_0}}=\Id+(c_0^{\chi}\N_\chi-\N_0)\G^-_{\N_0}|_{\Gamma_c(\M)}\:.$$
Let us pass to the proof of (1)-(3) for the general case $g\simeq g'$ also establishing the last part of the thesis.
In this case there is  a sequence 
$g_0=g, g_1,\ldots, g_N=g'$
of globally hyperbolic metrics on $\M$ satisfying Definition \ref{def:paracausal def} and a corresponding sequence 
of selfadjoint  normally hyperbolic operators $\N_k$ with $\N_0:= \N$ and $\N_{N}:= \N'$. (This sequence always exists because, for every globally hyperbolic metric $g$, there is a normally hyperbolic operator $\N$ as proved in the proof of Theorem \ref{remchain}. The operator $\tilde{\N}:= \frac{1}{2}(\N + \N^{\dagger_g})$ is simultaneously formally selfadjoint with respect to $\fiber{\cdot}{\cdot}$ and normally hyperbolic.)
Taking advantage of the validity of the thesis in the cases  $g\preceq g'$ and $g'\preceq g$, using in particular (4) and (6) in 
 Proposition \ref{propadjoint}, one immediately shows that we can build a M{\o}ller map for a paracausal deformation of metrics just by defining $\R$ as the composition of the various similar operators defined for each copy $g_k, g_{k+1}$ as in (\ref{frop1}) and (\ref{frop2}).

(4)  If $\ff \in \Gamma_c(\E)$, $$\R^{\dagger_{gg'}} \N'\ff = \R^{\dagger_{gg'}} \N'^{\dagger_{g'}} \ff=
 (\N' \R)^{\dagger_{gg'}} \ff = (c' \N)^{\dagger_{gg'}} \ff= \N^{\dagger_{g}} \ff = \N\ff \:.$$

(5) It is sufficient to prove the thesis for the case $g=g_0 \preceq g_1=g'$ and for $\R_+^{\dagger_{g_0}}$. The case of $\R_-^{\dagger_{g_0}}$ is analogous. In the case $g_1 \preceq g_0$ one uses the inverses of the operators above,
 and all remaining cases are proved just by observing that the considered M\o ller operators are compositions of the elementary operators $\R_\pm^{\dagger_{g_0}}$ and/or their inverses and smooth functions used as multiplicative operators.
We know that 
$$ \R_+^{\dagger_{g_0}}=\Id - (c^\chi_0 \N_\chi -\N_0)\G_{c^\chi_0 \N_\chi}^-|_{\Gamma_c(\E)}\:.$$
The identity operator has already the requested continuity property, so we have only to focus on the second addend using the fact that a linear combination of continuous maps is continuous as well.
The map $\G_{c^\chi_0 \N_\chi}^-|_{\Gamma_c(\E)} : \Gamma_c(\E) \to \Gamma(\E)$ is continuous with respect to the natural topologies  of the domain and co-domain (see {\it e.g.}~\cite[Corollary 3.6.19 ]{Ba-lect}). Since $ (c^\chi_0 \N_\chi -\N_0)$ is a smooth differential operator
$ (c^\chi_0 \N_\chi -\N_0)\G_{c^\chi_0 \N_\chi}^-|_{\Gamma_c(\E)} : \Gamma_c(\E) \to \Gamma(\E)$ is still continuous. To conclude the proof it is sufficient to prove that if $ \Gamma_c(\E) \ff_n \to 0$ in the topology of  $\Gamma_c(\E)$ and $K \supset supp(\ff_n)$ for all $n\in \mathbb{N}$ is a compact set, then there is a compact set $K'$ such that $K' \supset supp( (c^\chi_0 \N_\chi -\N_0)\G_{c^\chi_0 \N_\chi}^- \ff_n)$ for all $n\in \mathbb{N}$.  If $t : \M \to \RR$ is the Cauchy temporal function of $g_1$ used to construct $\R_+$ and $\R_-$, whose level sets $\Sigma_{\tau}:= t^{-1}(\tau)$ are Cauchy hypersurfaces for $g_0,g_\chi, g_1$ and $g_\chi =g_0$ in the past of $\Sigma_{t_0}$, then the  set  $J^{(\M,g_\chi)}_-(K) \cap D^{(\M,g_\chi)}_+(\Sigma_{t_0})$, which is compact for known properties of globally hyperbolic spacetimes,  includes all supports of $(c^\chi_0 \N_\chi -\N_0)\G_{c^\chi_0 \N_\chi}^- \ff_n$ from the very definition of retarded Green operator also using the fact that $(c^\chi_0 \N_\chi -\N_0)$ vanishes in the past of $\Sigma_{t_0}$.
\end{proof}

As a byproduct of Theorem~\ref{teoRRR} we get a technical, but  important, corollary.

\begin{corollary}\label{corMol} Consider $g,g',g'' \in \mathcal{GH}_\M$, corresponding formally selfadjoint and normally hyperbolic operators $\N,\N',\N''$ on the $\KK$-vector bundle $\E$ on $\M$ equipped with a non-degenerate, Hermitian, fiberwise  metric.  Assume that $g\simeq g'$ and $g'\simeq g''$ and suppose that $\R_{gg'}$ is a  M\o ller operator of  $g,g'$ 
and $\R_{g'g}$ is a  M\o ller operator of $g',g''$
according to 
(\ref{frop1}). The following facts are true.
\begin{itemize} 
\item[(1)]  $\R_{gg'}^{-1}$ is a M\o ller operator of $g',g$.
\item[(2)]  $\R_{gg'}\R_{g'g''}$ is a M\o  ller operator of $g',g''$.
\end{itemize} 
\end{corollary}

\begin{proof} It is immediate form the construction of $\R$ described at the end of Theorem \ref{teoRRR} relying on (\ref{frop1}).
\end{proof}

\begin{remark}\label{remaggeqR} Observe that the construction of the M\o ller operator $\R$ of $g_0,g_1$, for $g_0 \preceq g_1$,
as $\R = \R _-\R_+ $ we used several times in this work is nothing but an elementary case of (2). Indeed, in that case, 
$g_0 \preceq g_\chi \preceq g_1$ and $\R_+$, $\R_-$ are, respectively, a M\o ller operator of $g_0,g_\chi$ and $g_\chi, g_1$.
\end{remark}

\section{M{\o}ller $*$-isomorphisms in algebraic quantum field theory}\label{sec:Moller AQFT}

The aim of this section is to investigate the role of the M{\o}ller operators at the quantum level. In order to achieve our goal, we will follow the so-called algebraic approach to quantum field theory, see \textit{e.g.}~\cite{CQF1, CQF2,BD,FK,IgorValter}.
In \textit{loc. cit.} the quantization of
a free field theory on a (curved) spacetime is interpreted as a two-step procedure:
\begin{enumerate}
\item The first consists of the assignment to a physical system of a  $*$-algebra of observables which
encodes structural properties such as causality, dynamics and canonical commutation relations.
\item The second step calls for the identification of an algebraic state, which is a positive, linear
and normalized functional on the algebra of observables.
\end{enumerate}
Using this framework, in this section we shall lift the action of the M{\o}ller operators on the algebras of the free quantum fields and then we will pull-back the action of the M\o ller operators on quantum states, showing that the maps preserve the Hadamard condition with quite weak hypotheses which, in principle, permit an extension of the theory to a perturbative approach.\\
For a more detailed introduction to the algebraic approach to quantum field theory we refer to~\cite{aqft2,gerard} for textbook and to~\cite{cap2,BD,BDM,Nic,cap,capS,DHP,DMP,DMP3,DFMR,Moller,
gerard1,gerardgrav,gerard2,gerard3,gerard4,gerard5,gerardDirac,gerardYM, CSth,cospi,SDH12} for some recent applications.
\medskip

We begin first by recalling the construction of the free quantum field theories in curved spacetime for the general class of Green hyperbolic operators, which, as we have seen contains the class of the normally hyperbolic operators which are the ones under exam. 

\begin{notation}
Through this section, $\E$ will always denote an $\RR$-vector bundle over a globally hyperbolic spacetime $(\M,g)$. In particular, we denote the non-degenerate, symmetric, fiberwise  metric  by $\fiber{\cdot}{\cdot}$.   
\end{notation}

\subsection{Algebras of free quantum fields and the M{\o}ller $*$-isomorphism}

Given a formally-selfadjoint normally hyperbolic operator $\N :\Gamma(\E)\rightarrow\Gamma(\E)$ and its causal propagator $G$, we first define the unital complex $*$-algebra $\mathcal{A}_f$ as the {\em free complex unital $*$-algebra} with abstract (distinct)
generators $\phi(\ff)$ for all $\ff\in\Gamma_c(\E)$, identity  $1$, and involution $^*$ as discussed in  \cite{IgorValter}. (As a matter of fact $\mathcal{A}_f$ is made of finite linear complex combinations of $1$ and finite products of generic elements $\phi(\ff)$ and $\phi(\fh)^*$.) Then we define a refined complex unital $*$-algebras  by imposing the following relations by the quotient $\mathcal{A}=\mathcal{A}_f\slash\mathcal{I}$ where $\mathcal{I}$ is the two sided $*$-ideal generated by the following elements of $\mathcal{A}$:
\begin{itemize}
\item $\phi(a\ff+b\fh)-a\phi(\ff)-b\phi(\fh)\:, \quad \forall a,b\in\RR \quad \forall \ff,\fh\in\Gamma_c(\E)$
\item $\phi(\ff)^*-\phi(\ff)\:,\quad \forall \ff\in\Gamma_c(\E)$
\item $\phi(\ff)\phi(\fh)-\phi(\fh)\phi(\ff)-i\G_{N}(\ff,\fh)1\:,  \quad \forall \ff,\fh\in\Gamma_c(\E)$,
\end{itemize}
where we have used the notation
$$\G_{\N}(\ff,\fh) := \int_\M \fiber{\ff(x)}{(\G_\N \fh)(x)} \vol_g(x)\:.$$
We have the further possibility  to enrich the ideal with the generators:
\begin{itemize}
	\item $\phi(\N\ff)\:,  \quad \forall\ff\in\Gamma_{c}(\E)$.
\end{itemize}
\begin{notation}
The equivalence classes $[\phi(\ff)]$ will be denoted by $\Phi(\ff)$ and they will be called {\bf field operators}
({\bf on-shell} if the ideal is enlarged by including the generators $\phi(\N\ff)$), and we use the notation $\II$ for the 
identity $[1]$ of $\mathcal{A}_f/ \mathcal{I}$. 
\end{notation}
\begin{definition}\label{CCRalgebra} Given a formally-selfadjoint normally hyperbolic operator $\N :\Gamma(\E)\rightarrow\Gamma(\E)$ and its causal propagator $G$, we call  {\bf CCR algebra } 
of the quantum fields $\Phi$, the unital $*$-algebra defined by   $\mathcal{A}:= \mathcal{A}_f/ \mathcal{I}$.  The algebra is said to be {\bf on-shell} in case the ideal is enlarged by including the generators $\phi(\N\ff)$. Furthermore, we call  {\bf observables} of $\mathcal{A}$ any  Hermitian element of it. 
\end{definition}

With the above  notation, the following properties are valid
\begin{itemize}
\item\textbf{$\RR$-Linearity}. $\Phi(a\ff+b\fh)=a\Phi(\ff)+b\Phi(\fh)\:, \quad \forall a,b\in\RR \quad \forall \ff,\fh\in\Gamma_c(\E)$
\item{\textbf{Hermiticity}}. $\Phi(\ff)^*= \Phi(\ff)\:,\quad \forall \ff\in\Gamma_c(\E)$
\item{\textbf{CCR}}. $\Phi(\ff)\Phi(\fh)-\Phi(\fh)\Phi(\ff)= i\G_{N}(\ff,\fh)\II\:,  \quad \forall \ff,\fh\in\Gamma_c(\E)$.
\end{itemize}
The on-shell field operators also satisfy
\begin{itemize}
	\item \textbf{Equation of motion}. $\Phi(\N\ff)=0\:,  \quad \forall\ff\in\Gamma_{c}(\E)$.
\end{itemize}

\begin{remark}
 The idea behind the notation $\Phi(\ff)$ is a formal smearing procedure that uses the scalar product
$$\Phi(\ff) = \int_\M \fiber{\Phi(x)}{\ff(x)} \vol_g(x)\:.$$
From this perspective, since $\N$ is formally selfadjoint,
 the identity $\Phi(\N\ff)=0$ for all $\ff \in \Gamma_c(\E)$  has the distributional meaning $\N \Phi =0$. Alternatively, as explained in \cite{propSing}, one may use a different representation where $\Phi$ is viewed as a ``generalized section'' of  the dual bundle $\E^*$. In that case the formal identity $\N \Phi =0$ corresponding to the equation of motion has to be replaced by $\N^* \Phi =0$.
\end{remark}

Given different normally hyperbolic operators  $\N,\N'$  all the information about causality and dynamics is  encoded in the ideal $\mathcal{I}, \mathcal{I}'$. In that case we have two corresponding initial  unital $*$-algebras $\mathcal{A}_f$ and $\mathcal{A}'_f$ with 
respective generators $\phi(\ff)$ and $\phi'(\ff)$.
Though the freely generated algebras are canonically isomorphic, under the unique unital $*$-isomorphism such that  $\phi(\ff) \to \phi'(\ff)$ for all $\ff \in \Gamma_c(\E)$,
the quotient  algebras are intrinsically different because the CCR are different depending on the choice of the causal propagator $\G_{\N}$ or $\G_{\N'}$. However there is an isomorphism between them as soon as a M\o ller operator exists.
Indeed, the existence of the M{\o}ller operator discussed in the previous sections can be exploited to define first an isomorphism  of the free algebras $\mathcal{A}_{f}$ and $\mathcal{A}'_{f}$ since the operator $\R^{\dagger_{gg'}}:\Gamma_c(\E)\to\Gamma_c(\E)$ is an isomorphism.\medskip

\begin{definition}	
Let $\N,\N' : \Gamma(\E) \to \Gamma(\E)$ be two formally-selfadjoint (with respect to a fiber metric $\fiber{\cdot}{\cdot}$)  normally hyperbolic operators  globally hyperbolic spacetimes $(\M,g)$ and $(\M,g')$. If $g\simeq g'$, 
we define an  isomorphism  $\mathcal{R}_f:\mathcal{A}'_{f}\to\mathcal{A}_{f}$ as the unique unital $*$-algebra isomorphism between the said free unital $*$-algebras such that $\mathcal{R}_f(\phi'(\ff))=\phi(\R^{\dagger_{g g'}}\ff)\quad\forall\ff\in\Gamma_c(\E)$.
where  $\R$ is a M\o ller operator of $g,g'$ (in this order) satisfying Theorem \ref{teoRRR} and Equation~\eqref{frop1}.
\end{definition}
As we shall see in the next proposition, the isomorphism between freely generated algebras induces an isomorphism of the quotient algebras.

\begin{proposition}\label{algebraicMoller}
Let $\N$ and $\N'$ be two formally-selfadjoint  normally hyperbolic operators  acting on the sections of the
$\RR$-vector bundle 
$\E$ over $\M$, and referred 
to respective  $g,g'\in \mathcal{GM}_\M$.\\  If $g \simeq g'$ and $\R$ is a M\o ller operator of $g_,g'$ in the sense of Theorem \ref{teoRRR} and Equation~\eqref{frop1}, then the CCR algebras $\mathcal{A}$ and $\mathcal{A}'$ (possibly both on-shell) respectively associated to $\N$ and $\N'$ are isomorphic
under the quotient isomorphism $\mathcal{R} : \mathcal{A}'_{f}/\mathcal{I}' \to \mathcal{A}_{f}/\mathcal{I}$ constructed out of $\mathcal{R}_f$, the unique unital $*$-algebra isomorphism  satisfying $\mathcal{R}(\Phi'(\ff))=\Phi(\R^{\dagger_{g g'}}\ff)\quad\forall\ff\in\Gamma_c(\E)$.
\end{proposition}
\begin{proof}To prove the statement it suffices to show that the operator $\mathcal{R}_f$ maps the ideal $\mathcal{I}'$ to the ideal $\mathcal{I}$. 
Each ideal $\mathcal{I}$ and $\mathcal{I}'$ is the intersection of three (four) ideals corresponding to the requirements of linearity, Hermiticity, CCR (and equation of motion).
The fact that $\mathcal{R}_f$ preserves the ideals due to linearity and hermiticity is an immediate consequence of the fact that $\mathcal{R}_f$ is a $*$-algebra homomorphism of the involved freely generates algebras. The ideal arising from the equation of motion  condition is preserved due to the first statements of Theorem \ref{teoRRR} and item (4) therein. \\
The situation is more delicate regarding the ideal generated by the CCR. Preservation of that ideal is actually   
 an immediate consequence of $\mathcal{R}_f(\II')=\II$ ($\mathcal{R}_f$ is unital by hypothesis) and the structure of CCR together with (\ref{firstE}):
\begin{align*}
			\G_{\N}(\ff',\fh')& =\G_{\N_0}(\R^{\dagger_{g g'}}\ff,\R^{\dagger_{g g'}}\fh)\\
			&=\int_{\M}\fiber{\R^{\dagger_{g g'}}\ff}{\G_{\N}\R^{\dagger_{gg'}}\fh}\vol_{g}\\
			&=\int_{\M}\fiber{\ff}{\R \G_{\N}\R^{\dagger_{g g'}}\fh}\vol_{g'}\\
			&=\int_{\M}\fiber{\ff}{\G_{\N'}\fh}\vol_{g'} \\
			&= \G_{\N'}(\ff,\fh)\,.
\end{align*}
This concludes our proof.
	\end{proof}

\begin{definition} \label{defmollerstar}
A unital $*$-isomorphism $\mathcal{R} : \mathcal{A}' \to \mathcal{A}$ defined in Proposition \ref{algebraicMoller} out of 
the M\o ller operator   $\R$ of $g,g'$
as in Theorem \ref{teoRRR} and (\ref{frop1})
is called {\bf M\o ller $*$-isomorphism} of  the CCR algebras  $\mathcal{A}$,$ \mathcal{A}'$ (in this order)
\end{definition}

\subsection{Pull-back of algebraic states through the M\o ller $*$-isomorphism}\label{subsec:algebraic moller}

As explained in the beginning of this section, the subsequent step in the quantization of a field theory consists in identifying a distinguished state on the $*$-algebra of the quantum fields. The {\em GNS construction} then guarantees the existence of a representation of the quantum field algebra through, in general unbounded, operators defined over a common dense subspace of some Hilbert space. We will not care about the explicit representation and recall some definitions  (see \cite{DM} for a general discussion also pointing out several not completely solved standing issues). 

\begin{definition}
We call an \textbf{ (algebraic) state} over a unital $*$-algebra $\mathcal{B}$ a $\CC$-linear functional $\omega:\mathcal{B}\to\CC$ which is
\begin{itemize}
\item[(i)]\textbf{Positive} $\omega(a^*a)\geq0\quad\forall a\in\mathcal{B}$,
\item[(ii)]\textbf{Normalized} $\omega(\II)=1$
\end{itemize}
\end{definition}

A generic element of the CCR algebras $\mathcal{A}$ of a quantum field $\Phi$ associated to the normally hyperbolic operators discussed before can be written as a finite polynomial of the generators $\Phi(f)$, where the zero grade term is proportional to $\II$, to specify the action of a state it's sufficient to know its action on the monomials, i.e its \textbf{n-point functions}
\begin{equation}  \omega_n(\ff_1,..,\ff_n) :=\omega(\Phi(\ff_1)...\Phi(\ff_n)) \label{defomegan}\end{equation} 
The map $\Gamma_c(\E) \times \cdots \times \Gamma_c(\E) \ni (\ff_1,\ldots, \ff_n) \mapsto \omega_n(\ff_1,..,\ff_n)$ can be extended by linearity to the space of finite  linear combinations of sections $\ff_1 \otimes \cdots  \otimes \ff_n \in \Gamma_c(\E^{n\boxtimes})$, where $\E^{n\boxtimes}$ is $n$-times exterior tensor product of the vector bundle $\E$ with itself.
If we impose continuity with respect to the usual topology on the space of compactly supported test sections, since the said linear combinations are dense,  we can uniquely extend  the $n$-point functions to distributions 
in  $\Gamma_c'(\E^{n\boxtimes})$ we shall hereafter indicate by the same symbol $\omega_n$. 
It has a  formal integral kernel,
$$\omega_n(\ff_1,..,\ff_n)=\int_{\M^n}\tilde{\omega}_n(x_1,...,x_n)\ff_1(x_1)...\ff_n(x_n) \vol_{\M^n}(x_1,\ldots,x_n),$$
where   $$\vol_{\M^n}(x_1,\ldots,x_n) := \vol_g(x_1) \otimes \cdots \mbox{($n$ times)} \cdots \otimes \vol_{g}(x_n)$$ henceforth.
Notice that if more strongly $\omega_n\in \Gamma_c'(\E^{n\boxtimes})$, then 
$${\omega}_n(\fh)= \int_{\M^n}\tilde{\omega}_n(x_1,...,x_n)\fh(x_1, \ldots, x_n) \vol_{\M^n}(x_1,\ldots,x_n)$$
is also defined for $\fh \in \Gamma_c(\E^{n\boxtimes})$.  The case $n=2$ is the easiest one. The Schwartz kernel theorem implies $\Gamma_c(\E) \ni \ff \mapsto \omega_2(\fh, \ff)$ is (sequentially) continuous at $\ff=0$ for every fixed $\fh \in \Gamma_c(\E)$  if and only if $\omega_2$ continuously extends to a unique distribution we hereafter  indicate with the same symbol   $\omega_2 \in \Gamma_c'(\E \boxtimes \E)$.

An important fact  (see  the comment after \cite[Proposition 5.6]{propSing}) is that, if the CCR algebra $\mathcal{A}$ admits states, then  the fiberwise metric $\langle \cdot|\cdot\rangle$ must be positive. In other words,
$\langle \cdot|\cdot\rangle$ is a real  symmetric positive scalar product. We shall assume it henceforth.

Differently from a free quantum field theory on Minkowski spacetime, where the Poincaré invariant state -- known as Minkowski vacuum -- might be a natural choice, on a general curved spacetime there might be no choice of a natural state. However there is a class of states, known as  quasifree (or Gaussian) states, whose GNS representation mimics the Fock representation of Minkowski vacuum (see {\it e.g.}~ \cite{IgorValter}).
  
\begin{definition}\label{defAS} Let $\mathcal{A}$ be the CCR algebra.
A state $\omega:\mathcal{A}\to\CC$ is called {\bf quasifree}, or equivalently {\bf Gaussian}, if the following properties for its n-point functions hold
\begin{itemize}
	\item[(i)] $\omega_{n}(\ff_1,...,\ff_{n})=0$, if $n\in \mathbb{N}$ is odd,
	\item[(ii)]  $\omega_{2n}(\ff_1,...,\ff_{2n})=\sum_{partitions}  \omega(f_{i_1},f_{i_2}) \cdots  \omega(f_{i_{n-1}},f_{i_n})$,
if $n\in \mathbb{N}$ is even,
\end{itemize}
where ``partitions'' for even $n$ refers to the class of all possible decompositions of the set $\{1,2,\ldots, n\}$ into $n/2$ pairwise disjoint subsets of $2$ elements $\{i_1,i_2\}$, $\{i_3,i_4\}$, $\ldots$, $\{i_n-1,i_n\}$ with $i_{2k-1} < i_{2k}$ for 
$k=1,2,\ldots, n/2$.
\end{definition}

For these states all the information is encoded in the two-point distribution, as one can expect in a free theory.  It is not difficult to prove that,
for a quasifree state in view of the definition above,  $\omega_2\in \Gamma_c'(\E)$ entails that $\omega_n$ continuously extends to   $\omega_n \in \Gamma_c'(\E^{n\boxtimes})$ obtained, for $n=2k$, as a linear combination of tensor products of distributions $\omega_2$ and trivial if $n=2k+1$.

\begin{remark} If $\mathcal{A}$ is on-shell, then the $n$-point function satisfies trivially
\begin{equation*}
 \omega_n(\ff_1,\ldots, \N\ff_k, \ldots, \ff_n)=0\quad \mbox{for every $k=1,\ldots, n$ and $f_k \in \Gamma_c(\M)$.}
\end{equation*}
as a consequence of (\ref{defomegan}) and $\Phi(\N\ff)=0$.  However it may happen that these identities are valid (for some $n$) even if 
the algebra is not on-shell.
\end{remark}

In the next proposition, we shall see that the action of the M\o ller isomorphism $\mathcal{R}$ between CCR-algebras can be pull-backed on the quantum states. Furthermore, the pull-back of a quasifree state is again a quasifree state.

\begin{proposition}\label{pullbackstate} Let be $g,g' \in \mathcal{GH}_\M$, consider the algebras $\mathcal{A}$, $\mathcal{A}'$ respectively associated to
formally-selfadjoint 
 normally hyperbolic operators $\N,\N' :\Gamma(\E)\to \Gamma(\E)$ constructed out of $g$ and $g'$ and 
let $\omega:\mathcal{A}\to\CC$ be a state. Assuming that $g \simeq g'$, we define a functional ${\omega}':\mathcal{A}'\to\CC$ by pull-back through a M\o ller $*$-isomorphism $\mathcal{R}: \mathcal{A}'\to \mathcal{A}$ 
of $ \mathcal{A}$,$\mathcal{A}'$
as in  Definition \ref{defmollerstar}, i.e.
$${\omega}'=\omega\circ\mathcal{R}.$$ Then the following statements hold true:
\begin{itemize}
	\item[(1)] ${\omega}'$ is a state on $\mathcal{A}'$;
	\item[(2)]  $\omega'_2\in \Gamma_c'(\E\boxtimes \E)$ if and only if $\omega_2\in \Gamma_c'(\E\boxtimes \E)$;
	\item[(3)] $\omega'$ is quasifree if and only if $\omega$ is.
\end{itemize}
\end{proposition}
\begin{proof}
(1) Linearity is obvious since we are composing linear maps. Normalization follows from 1 in \ref{algebraicMoller} and from the fact that $\omega$ is normalized. Positivity follows from positivity of $\omega$ and
the fact that $\mathcal{R}$ preserves the involutions,  the products, and is surjective.
(2) Since $\omega_2 \in \Gamma_c'(\E\times \E)$, then  it is $\Gamma_c(\E)$-continuous in the right entry (taking values in $\Gamma_c'(\E)$ and with respect to the corresponding topology). As a consequence, by composition of continuous functions, if $\fh\in \Gamma_c(\E)$ is given, 
$$\Gamma_c(\E) \ni \ff \mapsto \omega'_2(\fh, \ff) = \omega_2(\R^{\dagger_{g g'}}\fh, \R^{\dagger_{g g'}} \ff )$$
is  $\Gamma_c(\E)$-continuous 
as well because 
 $\R^{\dagger_{g g'}}: \Gamma_c(\E) \to \Gamma_c(\E)$ is continuous in the $\Gamma_c(\E)$ topology in domain and co-domain for  (5) of Theorem \ref{teoRRR}. In other words $\Gamma_c(\E) \ni \ff \mapsto \omega'_2(\cdot, \ff) \in \Gamma'_c(\E)$ is continuous.  We conclude that  $\omega'_2  \in \Gamma_c'(\E\boxtimes \E)$ due to the Schwartz kernel theorem. The result can be reversed swapping the role of the states and the metrics,  
noticing that $\omega = \omega' \circ \mathcal{R}^{-1}$ where $\mathcal{R}^{-1}$ is also a M\o ller  $*$-isomorphism, the one constructed out of   $\R^{-1}$ which is, in turn,  a M\o ller  operator associated to the pair $g'$, $g$ in this order in view of Corollary \ref{corMol}. \\
(3) The proof is immediate and  follows by construction. 
\end{proof}

\subsection{M\o ller preservation of the microlocal spectrum condition for off-shell algebras}\label{subsec:Moller Hadamard}

It is widely accepted that, among all possible (quasifree) states, the physical
ones are required to satisfy the so-called Hadamard condition. The reasons for this choice are
manifold: For example, it implies the finiteness of the quantum fluctuations of the expectation
value of every observable and it allows us to construct Wick polynomials \cite{HW1,KM} and other observables, as the stress energy tensor,  relevant in semi-classical quantum gravity following a covariant
scheme \cite{Mo03,HM}, encompassing a locally covariant ultraviolet renormalization \cite{HW} (see also \cite{IgorValter} for a recent pedagogical review).
These states have been also employed, e.g. (the following list is far from being exhaustive)  in the study of the Blackhole radiation \cite{DMP3,gerardUnruh, MP,VXP}, in cosmological models  \cite{DMP,DFMR} and other applications to spacetime models \cite{Mo08,FMR1,FMR2}, and to study energy quantum inequalities \cite{FS}. 
 For later convenience, we decided to present the Hadamard condition as a microlocal condition on the wave-front set of the two-point distribution \cite{Radzikowski,Radzikowski2} instead of the 
equivalent geometric version based on the Hadamard parametrix \cite{KW,PDEcs, M}. Let's briefly sketch what they are and why they are useful. 

From now on we adopt the definitions of wave-front set $WF(\psi)$ of distribution $\psi$ on $\RR$-vector bundles equipped with a non-degenerate, symmetric, fiberwise  metric\footnote{The authors of \cite{propSing} more generally study the case of a complex Hermitian vector bundle endowed with an antilinear involution (here the identity bundle map) there indicated by $\Gamma$.} as in \cite{propSing}.

We shall use some very known definitions and results of microlocal analysis applied to distributions of $\Gamma'_c(\F)$ where $\F$ is a $\KK$-vector bundle,  $\F= \E \boxtimes \E$ for instance (see \cite{propSing} for details).  In particular,  
\begin{itemize}
\item $\psi \in \Gamma'_c(\F)$ is a smooth section of the dual bundle $\F^*$, indicated with the same symbol $\psi \in  \Gamma(\F^*)$,
if and only if $WF(\psi) = \emptyset$.
\item We say that $\psi,\psi' \in \Gamma_c'(\F)$ are {\bf equal  mod $C^\infty$}, if $\psi-\psi' \in \Gamma(\F^*)$.

\item Let us assume that $\F= \E \boxtimes \E$ where $\E$ is equipped with a non-degenerate, symmetric (Hermitian if $\KK=\CC$), fiberwise  metric and let $\P: \Gamma_c(\E) \to \Gamma_c(\E)$ be a formally selfadjoint smooth differential operator. We say that $\nu \in  \Gamma'_c(\E\boxtimes \E)$ is a {\bf bi-solution $\P\ff =0$  mod $C^\infty$}, if there exist $\varphi, \varphi' \in \Gamma(\F^*)$ such that 
$$\nu(\P \ff \otimes \fh) =  \int_\M  \langle \psi, \ff\otimes \fh\rangle\: \vol_g\otimes \vol_g\:, \quad 
\nu( \ff \otimes \P\fh) =  \int_\M  \langle \psi', \ff\otimes \fh\rangle\: \vol_g\otimes \vol_g \quad \forall \ff,\fh \in \Gamma_c(\E)\:.$$
\end{itemize}

We are in a position to state the definition of {\em micro local spectrum condition} and {\em Hadamard state}.
Below,  $\sim_{\parallel}$ is the equivalence relation in $T^*\M^2\backslash\{0\}$ such that  
$(x,k_x)\sim_{\parallel}(y,k_y)$ if there is 
 a \textit{null geodesic} passing through  $x,y\in\M$ and the geodesic parallely transports the co-tangent vector to  that geodesic $k_x \in T^*_x\M$
into the  co-tangent vector  to that geodesic $k_y\in T^*_y\M$.
Finally, $k_x\triangleright 0$ means that the covector $k_x$ is \textit{future directed}. 

\begin{definition}\label{Hadamard} With $\mathcal{A}$ as in Definition \ref{defAS},
a state $\omega:\mathcal{A}\to\CC$ is called a \textbf{Hadamard state} if
  $\omega_2 \in \Gamma_c'(\E \boxtimes\E)$ and
 the following {\bf microlocal spectrum condition} is valid
\begin{equation} \label{WHad} WF(\omega_2)=\{(x,k_x;y,-k_y)\in T^*\M^2\backslash\{0\}\:|\:(x,k_x)\sim_{\parallel}(y,k_y), k_x\triangleright0\}\:.\end{equation}
More generally, a distribution $\nu   \in \Gamma_c'(\E \boxtimes\E)$
is said to be of {\bf Hadamard type} if its wave-front set $WF(\nu)$  is  the right-hand side of (\ref{WHad}).
\end{definition}

\begin{remark} $\null$
\begin{itemize}
\item[(1)] Notice that  $(x,k_x;x,-k_x) \in WF(\nu)$ for every future directed lightlike covector $k_x\in T_x^*\M$ if $\nu\in \Gamma_c'(\E \boxtimes\E)$ is  of Hadamard type.
\item[(2)] It is possible to prove that a fiberwise scalar product $\langle\cdot|\cdot\rangle$ must be necessarily positive if $\mathcal{A}$ admits quasifree Hadamard states as proved in the comment after \cite[Proposition 5.6]{propSing}. We  henceforth assume that 
$\langle\cdot|\cdot\rangle$ is positive.
\end{itemize}
\end{remark}

The theorem below, which is the second main result of this paper, shows that the Hadamard condition is preserved under the pull-back along the  M\o ller isomorphism.

\begin{theorem}\label{thm:main Had}
Let $\E$ be an $\RR$-vector bundle over smooth manifold $\M$ and denote with $\fiber{\cdot}{\cdot}$ positive, symmetric, fiberwise  metric.
 Let be  $g,g' \in \mathcal{GH}_\M$, consider the corresponding  formally-selfadjoint normally hyperbolic operators $\N,\N' : \Gamma(\E) \to \Gamma(\E)$ and refer to the associated CCR algebras $\mathcal{A}$ and $\mathcal{A}'$ (off-shell in general). Finally, suppose that $g\simeq g'$.\\
 $\omega:\mathcal{A}\to\CC$ is  a quasifree Hadamard state,
 if and only if  $$\omega':= \omega \circ \mathcal{R} : \mathcal{A}'\to\CC\:,$$
constructed out of a M\o ller $*$-isomorphism $\mathcal{R}$ of $\mathcal{A},\mathcal{A}'$, is a quasifree  Hadamard state of $\mathcal{A}'$. 
\end{theorem}

\begin{remark} We stress that it is not required that the algebras are on-shell nor that the relevant  two-point functions satisfy the equation of motion with respect to the corresponding normally hyperbolic operators.
\end{remark}

The rest of this section is devoted to prove Theorem~\ref{thm:main Had},  a refinement  of it stated in the last Theorem \ref{thm:main intro Had}, and a proof of existence of Hadamard states based on our formalism. 

Our first observation is the following.
\begin{lemma}\label{lemmaAppedix}
Let  $\S: \Gamma(\E) \to \Gamma(\E)$ be  any of the four operators $\R_+$, $\R_-$, $\R_+^{-1}$, $\R_-^{-1}$,
defined as in (\ref{defR+}), (\ref{defR-}), (\ref{inverse}), (\ref{inverse2}), and $U\subset \RR^m$ an open set.\\
If $\{\ff_z\}_{z\in U} \subset \Gamma(\E)$ is such that 
$\M \times U \ni (x,z) \mapsto \ff_z(x)$ is jointly smooth, then $$\M\times U \ni (x,z) \mapsto (\S \ff_z)(x)$$ is jointly smooth as well.
\end{lemma}

\begin{proof}
We consider the case of $\R_+$, the remaining three instances having a similar proof. What we have to prove is that
$\M\times U \ni (x,z) \mapsto  \left(\G_{\rho \N_\chi}^+(\rho \N_\chi -\N_0) \ff_z\right)(x)$ is smooth under the said hypotheses.
Let us first consider the case where there is compact $K \subset \M$ such that $\mbox{supp}(\ff_z) \subset K$ for all $z\in U$. In this case, defining $F(x,z) := (\rho \N_\chi -\N_0)\ff_z(x)$,
 the projection $\pi: \mbox{supp}(F) \ni (x,z) \mapsto z \in U$ is proper\footnote{If $C \subset U$ is compact and thus closed, then $\pi^{-1}(C)$ is a closed set, $\pi$ being continuous,  contained in the compact $K \times C$, so that $\pi^{-1}(C)$ is compact as well.} and this fact will be used shortly. Interpreting 
$\G_{\rho \N_\chi}^+ : \Gamma_c(\E) \to \Gamma_c(\E)$ and thus as a Schwartz kernel, we can compute the wavefront set of the map $\M\times U \ni (x,z) \mapsto  \left(\G_{\rho \N_\chi}^+(\rho \N_\chi -\N_0) \ff_z\right)(x)$ viewed as the distributional kernel of the composition of the kernel  $\G_{\rho \N_\chi}^+(x,y)$ and the smooth kernel $F(y,z)$.
We know that  (see, e.g. \cite{IgorValter} for the scalar case, the vector case being analogous)
$$ WF(\G_{\rho \N_\chi}^+)=\left\{(x,k_x;y,-k_y)\in T^*\M^2\backslash\{0\}\:|\:(x,k_x)\sim_{\parallel}(y,k_y), x \in J_+(y) \:\mbox{or}\: k_x=k_y\:, x=y\right\}$$
whereas, since $F$ is jointly smooth,
$$WF(F) = \emptyset\:.$$
The known composition rules of wavefront sets of Schwartz kernels, which use  in particular the fact that the projection $\pi$ above is proper  (in \cite[Theorem 5.3.14]{IgorValter} which is valid also in the vector field case), immediately yields
$$WF(\G_{\rho \N_\chi}^+ \circ  F ) \subset \emptyset\:.$$
It being $WF(\G_{\rho \N_\chi}^+ \circ F ) = \emptyset$, we conclude that $\M\times U \ni (x,z) \mapsto  \left(\G_{\rho \N_\chi}^+(\rho \N_\chi -\N_0) \ff_z\right)(x)$ is a smooth function as desired.\\
Let us pass to consider the generic jointly smooth family $\{\ff_z\}_{z\in U} \subset \Gamma(\E)$ without restrictions on the supports.
First of all, we observe that $\ff'_z(x) :=  ((\rho \N_\chi -\N_0)\ff)(x)$ is past compact by construction for every $z\in U$, because its support is contained in the future of $\Sigma_{t_0}$ referring to the construction of $\N_\chi$.
According to the proof of \cite[Theorem 3.6.7]{Ba-lect},
if $\fh$ is past compact, $x_0 \in \M$, and $A \supset \mbox{supp}(\fh) \cap J_-(x_0)$ is an open relatively compact set, for every compactly supported smooth function $s_A \in C^\infty_c(\M; [0,1])$ such that $s_A(x)=1$ if $x\in A$, it holds
$$(\G_{\rho \N_\chi}^+ \fh)(x_0) = (\G_{\rho \N_\chi}^+ s_A \fh)(x_0) \:.$$
We want to apply this identity for $\fh = \ff_z$.
Take  $t' <t_0$. Given $x_0\in \M$ we can always define $A:=I_-(\tilde{x}_0) \cap I^+(\Sigma_{t'}) $ where $\tilde{x}_0 \in I_+(x_0)$ \footnote{Notice that since the spacetime is  globally hyperbolic, $\overline{I_\pm(x)}= J_\pm(x)$ and $\overline{I_-(\tilde{x}_0) \cap I^+(\Sigma_{t'})}  
=J_-(\tilde{x}_0) \cap J^+(\Sigma_{t'}) $ which is  compact because $\Sigma_{t'}$ is a smooth spacelike Cauchy surface.}.
With this choice, $A$ does not depend on $z\in U$ and the same $A$ can be used for $x$ varying in an open neighborhood $A'$ of $x_0$, since $I_-(\tilde{x}_0)$ is open. We conclude that, if $(x,z) \in A'\times U$, then 
\begin{equation} \left(\G_{\rho \N_\chi}^+(\rho \N_\chi -\N_0) \ff_z\right)(x) = (\G_{\rho \N_\chi}^+\circ F)(x,z) \quad \mbox{where $F(x,z) = s_A(x) 
(\rho \N_\chi -\N_0)\ff_z(x)  $} \:.\label{comparison}\end{equation}
In this case $K:= \mbox{supp}(s_A)$ includes all the supports of the maps $\M \ni x \mapsto F(x,z)$
for every $z\in U$. The first part of the proof is therefore  valid for the map  $\M\times U \ni (x,z) \mapsto (\G_{\rho \N_\chi}^+\circ F)$ which must be  jointly smooth as a consequence.
In particular, its restriction  $A'\times U \ni (x,z) \mapsto \left(\G_{\rho \N_\chi}^+(\rho \N_\chi -\N_0) \ff_z\right)(x)$ is   jointly smooth as well. Since $A'$ can be taken as a neighborhood of every point in $\M$ and $z\in U$ is arbitrary, from (\ref{comparison}) the whole function 
$\M\times U \ni (x,z) \mapsto \left(\G_{\rho \N_\chi}^+(\rho \N_\chi -\N_0) \ff_z\right)(x)$
is jointly smooth.
\end{proof}

Relying  on Lemma~\ref{lemmaAppedix}, we can notice the following.
\begin{lemma}\label{lemmaBI}  Consider a pair of globally hyperbolic metrics $g_0$ and $g_\chi$ on $\M$ as in Proposition \ref{propositionA} and 
corresponding  normally hyperbolic operators $\N_0,\N_\chi :\Gamma (\E) \to \Gamma(\E)$ for the $\RR$-vector bundle  on $\M$ equipped with the positive symmetric fiberwise metric $\fiber{\cdot}{\cdot}$. \\
Then, $\nu_0 \in \Gamma'_c(\E \boxtimes \E)$ is a bisolution of $\N_0 \ff=0$ mod $C^\infty$ if and only if $\nu_\chi:= \nu \circ \R_+^{\dagger_{g_0g_\chi}} \otimes \R_+^{\dagger_{g_0g_\chi}}$ is a bisolution of $\N_\chi \ff =0$ mod $C^\infty$, where $\R_+$ is defined in (\ref{defR+2}).
\end{lemma}

\begin{proof} We start by stressing that, as already noticed, in view of the known continuity properties of $\R_+^{\dagger_{g_0g_\chi}}$ 
and its inverse and using Schwartz' kernel theorem, $\nu _0\in  \Gamma'_c(\E \boxtimes \E)$ if and only if  $\nu_\chi \in \Gamma'_c(\E \boxtimes \E)$. \\ We pass to prove that if $\nu_0$ is a bisolution mod $C^\infty$, then $\nu_\chi$ is a bisolution mod $C^\infty$, referring to the corresponding operators.
Let us hence suppose that $\nu_0(\N_0\ff, \fh) = \psi(\ff \otimes \fh)$ and $\nu_0(\ff, \N_0\fh) = \psi'(\ff \otimes \fh)$ 
for some smooth sections $\psi,\psi' \in \Gamma((\E\boxtimes \E)^*)$ and all $\ff,\fh \in \Gamma_c(\E)$. 
The identity
$$\R_+^{\dagger_{g_0g_\chi}} \N_\chi|_{\Gamma_c(\E)}= \N_0|_{\Gamma_c(\E)}\:,$$
immediately implies that, if $\varphi(x,y) := c_0^\chi(x)c_0^\chi(y)\psi(x,y)$, $\varphi'(x,y) := c_0^\chi(x)c_0^\chi(y)\psi'(x,y)$, 
$$\nu_\chi(\N_\chi \ff,\fh) = \int_{\M\times \M} \langle \varphi(x,y), (\Id \otimes \R_+^{\dagger_{g_0g_\chi}} (\ff \otimes \fh))(x,y)\rangle \vol_{g_\chi}(x) \otimes \vol_{g_\chi}(y)$$
and
$$\nu_\chi( \ff,\N_\chi\fh) = \int_{\M\times \M} \langle \varphi'(x,y), (\R_+^{\dagger_{g_0g_\chi}}\otimes\Id (\ff \otimes \fh))(x,y)\rangle \vol_{g_\chi}(x) \otimes \vol_{g_\chi}(y)$$
The proof ends if proving that  there are smooth sections $\varphi_1,\varphi'_1 \in \Gamma((\E\boxtimes \E)^*)$, such that 
$$ \int_{\M\times \M} \langle \varphi, \Id \otimes \R_+^{\dagger_{g_0g_\chi}} (\ff \otimes \fh)\rangle \vol_{g_\chi}\otimes \vol_{g_\chi}
= \int_{\M\times \M}\langle \varphi_1(x,y), \ff(x) \fh(y) \rangle\vol_{g_\chi}(x) \otimes \vol_{g_\chi}(y)$$
and
$$ \int_{\M\times \M} \langle \varphi, \R_+^{\dagger_{g_0g_\chi}}\otimes \Id  (\ff \otimes \fh)\rangle \vol_{g_\chi}\otimes \vol_{g_\chi}
= \int_{\M\times \M}\langle \varphi'_1(x,y), \ff (x) \fh(y) \rangle \vol_{g_\chi}(x) \otimes \vol_{g_\chi}(y)$$
for every pair $\ff,\fh \in \Gamma_c(\E)$. We prove the former identity only, the second one having an identical proof.
To this end we pass to the index notation (also assuming Einstein's summing convention), the indices being referred to the fiber in the local trivialization,
$$ \int_{\M\times \M} \langle \varphi, \Id \otimes \R_+^{\dagger_{g_0g_\chi}} (\ff \otimes \fh)\rangle \vol_{g_\chi}\otimes \vol_{g_\chi}
$$ $$= \sum_{j,k} \int_{\M\times \M}\chi_j(x) \chi'_k(y) \varphi_{ab}(x,y) (\R_+^{\dagger_{g_0g_\chi}} \ff)^a(x)  \fh^b(y) \vol_{g_\chi}(x)\otimes \vol_{g_\chi}(y) $$
Above $\{\chi_j\}_{j\in J}$ and $\{\chi'_k\}_{k\in K}$ are partitions of the unity of $\M$ subordinated to corresponding locally finite coverings of $\M$ supporting  local trivializations,  whose fiber coordinates  are labeled by $^a$ and $^b$. Moreover, only a finite number of indices $(j,k)\in J\times K$ give a contribution to the sum, uniformly in $x,y$, in view of the compactness of the supports of $\ff$ and $\fh$ and the local finiteness of the used coverings.
The right-hand side can be rearranged to
$$
=   \sum_{k\in K}\int_{\M} \chi'_k(y) \left(\sum_{j\in J} \int_{\M}\chi_j(x) \varphi_{ab}(x,y) (\R_+^{\dagger_{g_0g_\chi}} \ff)^a(x)   \right)\fh^b(y)\vol_{g_\chi}(y) $$
$$
=   \sum_{k\in K}\int_{\M} \chi'_k(y) \left( \int_{\M} \fiber{\varphi'_{yb}(x)}{(\R_+^{\dagger_{g_0g_\chi}} \ff)(x)} \vol_{g_\chi}(x)\right)   \fh^b(y)\vol_{g_\chi}(y) $$
$$=  \int_{\M} \sum_{k\in K} \chi'_k(y)  \left( \int_{\M} \fiber{(\R_+\varphi'_{yb})(x)}{\ff(x)}  \vol_{g_0}(x)\right) \fh^b(y)  \vol_{g_\chi}(y) $$
$$ =  \sum_{j,k} \int_{\M\times \M}  \chi_j(x) \chi'_k(y) c_0^\chi(x) (\R_+\varphi'_{yb})_a(x)\ff^a(x) \fh^b(y)   \vol_{g_\chi}(x)\otimes \vol_{g_\chi}(y)$$
$$=  \int_{\M\times \M}\langle \varphi_1(x,y), \ff\otimes \fh(x,y) \rangle\vol_{g_\chi}(x) \otimes \vol_{g_\chi}(y)\:,$$
where  we have locally defined ${\varphi'}_{yb}^a(x) := \xi^{ac}(x) \varphi_{cb}(x,y)$, with $\xi^{ab}(x)$ being the inverse fiber metric at $x\in M$ in any  considered local trivialization.
Above,
$\varphi_{1\:ab}(x,y) := c_0^\chi(x) (\R_+\varphi'_{yb})_a(x)$
is the candidate  section of $(\E \boxtimes \E)^*$  we were looking for, represented  in local coordinates of the atlas of the said trivialization.
 That section is smooth, i.e.,  $\varphi_1 \in \Gamma((\E \boxtimes \E)^*)$ as desired. Indeed,  the maps $\M \times U_k \ni (x,y) \mapsto \varphi'_{yb}(x)$ define  a family of sections of $\Gamma(\E)$ parametrized by $y\in U_k$ for every given  $b\in \{1,\ldots, N\}$, where $U_k\subset  \M$ is the projection onto $\M$ of the domain of the considered local trivialization. This family is jointly smooth in $x,y$   as established in Lemma \ref{lemmaAppedix}.  \\
The converse statement, that $\nu_0$ is a bisolution mod $C^\infty$
if $\nu_\chi$ is, can be proved with the same procedure simply replacing $\R_+$ with $(\R_+)^{-1}$ and using Lemma \ref{lemmaAppedix} again.
\end{proof}

Before giving the proof of Theorem~\ref{thm:main Had}, we need a final  lemma, which shows that any Hadamard distribution whose antisymmetric part is given by the causal propagator of a normally hyperbolic system $\N$ is actually a bisolution of $\N$ itself modulo smooth errors.

\begin{lemma}\label{lemmaCOM} Let $\N: \Gamma(\E) \to \Gamma(\E)$  be a  formally selfadjoint normally hyperbolic operator and suppose that  $\nu \in \Gamma_c'(\E\boxtimes \E)$ is of Hadamard type and satisfies 
$$\nu(x,y) -\nu(y,x) = i \G_\N(x,y)\quad \mbox{mod}\quad C^\infty$$
where $G_\N(x,y)$ is the distributional kernel of the causal propagator $\G_\N$. In this case $\nu$ is a bisolution of $\N \ff =0$ mod $C^\infty$.
\end{lemma}

\begin{proof} The proof  is a straightforward re-adaptation of the proof appearing in the {\em Note added in proof} of \cite{Radzikowski}.
\end{proof}

We are finally in a position to prove Theorem~\ref{thm:main Had}.

\begin{proof}[Proof of Theorem~\ref{thm:main Had}] We have only to prove that $\omega'$ is Hadamard 
if and only if $\omega$ is,
since the other preservation property has been already proved in (4) of Proposition \ref{pullbackstate}.
If  $g_0\simeq g_1$, there is a sequence 
of globally hyperbolic metrics $g'_0=g_0, g'_1,\ldots, g'_N=g_1$ such that either $g'_k\preceq g'_{k+1}$ or $g'_{k+1}\preceq g'_{k}$ and the future cones satisfy a corresponding inclusion.   The M\o ller operator $\mathcal{R}$  of $\mathcal{A}, \mathcal{A}'$ is obtained as the composition of the M\o ller operators $\mathcal{R}_k$ of the formally-selfadjoint normally hyperbolic operators $\N'_k, \N'_{k+1}$ associated to  the pairs $g'_k$, $g'_{k+1}$:
$$\mathcal{R} := \mathcal{R}'_0\mathcal{R}'_1\cdots\mathcal{R}'_{N-1}$$
as in the proof of Theorems \ref{remchain}, \ref{teoRRR} and (\ref{frop1}).  The thesis is demonstrated  if  we prove that, with obvious notation, $\omega^{k+1}$ is Hadamard if and only if $\omega^k$ is.  So in principle we have to prove the thesis only for
a pair of metrics $g_0,g_1$ with 
 the two cases $g_0 \preceq g_1$ and $g_1 \preceq g_0$. Actually the latter  is a consequence of the former, using the fact that a M\o ller $*$-isomorphisms are bijective and that a M\o ller operator of the second case is the inverse operator of a M\o ller operator of the first case in accordance to Corollary \ref{corMol}. In summary, the proof is over if establishing the thesis for the case $g=g_0 \preceq g_1=g'$ and we shall concentrate on that case only in the rest of the proof.

Recalling by (\ref{teoRRR}) and (\ref{propadjoint}) that $\R^{\dagger_{g_0 g_1}}=\R_{+}^{\dagger_{g_0 g_	\chi}}\R_{-}^{\dagger_{g_{\chi} g_1}}$, we write  
$$\omega^1_2(\ff_1,\ff_2)=\omega_2^0(\R^{\dagger_{g_0 g_1}}\ff_1,\R^{\dagger_{g_0 g_1}}\ff_2)=\omega_2^0(\R_{+}^{\dagger_{g_0 g_	\chi}}\R_{-}^{\dagger_{g_\chi g_1}}\ff_1,\R_{+}^{\dagger_{g_0 g_	\chi}}\R_{-}^{\dagger_{g_\chi g_1}}\ff_2).$$
To analyze the wave-front set of this bidistribution, we split again the operation in two steps. First we define a pull-back state on the algebra $\mathcal{A}_{\chi}$ of quantum fields defined  for the formally-selfadjoint normally hyperbolic operator $\N_{\chi}$, i.e a normally hyperbolic operator on $(\M, g_{\chi})$. This intermediate pull-back states reads
\begin{equation}\label{omomchi}\omega^{\chi}_2(\ff_1,\ff_2)=\omega^0_2(\R_{+}^{\dagger_{g_0 g_\chi}}\ff_1,\R_{+}^{\dagger_{g_0 g_\chi}}\ff_2). \end{equation}  
We intend to prove that $\omega_2^{\chi} \in \Gamma_c'(\E)$ is of Hadamard type if and only $\omega^0_2$ is. 
 Notice that both two-point functions have antisymmetric parts that coincide with $i\G_{\N_\chi}$ and $i\G_{\N_0}$, respectively, in view of the CCRs of the respective algebras. 
If $\omega_2^0 \in \Gamma_c'(\E)$ is of Hadamard type, then  it is a a bisolution of $\N_0\ff=0$ mod $C^\infty$ in view of Lemma \ref{lemmaCOM}.
The same argument proves that, if  $\omega_2^\chi \in \Gamma_c'(\E)$ is of Hadamard type, then   it is a bisolution of $\N_\chi\ff=0$ mod $C^\infty$ due to  \ref{lemmaCOM}.
Applying  Lemma \ref{lemmaBI} to both cases we  have that, 
\begin{itemize}
\item[(a)] $\omega_2^0 \in \Gamma_c'(\E)$ of Hadamard type implies that 
$\omega_2^0$ is a bisolution $\N_0\ff=0$ mod $C^\infty$ and  $\omega^\chi_2$ is a a bisolution of $\N_\chi\ff=0$ mod $C^\infty$; \item[(b)]   $\omega_2^\chi  \in \Gamma_c'(\E)$ of Hadamard type implies that 
$\omega_2^\chi$ is a bisolution $\N_\chi\ff=0$ mod $C^\infty$ and  $\omega^0_2$ is a a bisolution of $\N_0\ff=0$ mod $C^\infty$.
\end{itemize}

We are now in a position to apply the Hadamard singularity propagation theorem.  Consider the smooth Cauchy time function $t$ in common with $g_0$ and $g_\chi$, such that $\chi(x)=0$ if $t(x) < t_0$. As a preparatory remark we 
notice  that $\R_+^{\dagger_{g_0g_\chi}}\ff = \ff$ from (\ref{RRdagger}) when the support of $\ff$ stays in the past of the Cauchy surface  $\Sigma_{t_0}= t^{-1}(t_0)$. In that region $g_0=g_\chi$ by definition of $g_\chi$. Finally  due to (\ref{omomchi}),
$$\omega^\chi_2(\ff,\fh)= \omega^0_2(\ff,\fh) \quad \mbox{if $t(\mbox{supp}(\ff))< t_0$, $t(\mbox{supp}(\fh)) < t_0$}$$
Hence, in particular,  $\omega_2^\chi$ is of Hadamard type when the supports of the test functions are taken in that  region if and only if 
 $\omega_2^0$ is of Hadamard type when the supports of the test functions are taken there.
More precisely, it happens when the supports of the arguments $\ff,\fh$
 are taken in a (globally hyperbolic)  neighborhood of a Cauchy surface (for both metrics!) $\Sigma_\tau:= t^{-1}(\tau)$ with $\tau < t_0$ between two similar slices.
Since both distributions are bisolutions of the respective equation of motion mod $C^\infty$ and the operators are normally hyperbolic,  the theorem of propagation of Hadamard singularity (see, e.g., Theorem 5.3.17 in \cite{IgorValter}\footnote{The proof which appears  there is valid for the  on-shell algebra of the scalar real Klein-Gordon field, 
but the passage to  normally hyperbolic operators also  weakening the bisolution requirement to bisolution mod $C^\infty$ is immediate, since it is based on standard H\"ormander theorems about  singularity propagation which works mod $C^\infty$.
See the comments in Remark 5.3.18 in \cite{IgorValter}}) implies that $\omega_2^\chi$ and $\omega_2^0$ are of Hadamard type everywhere in $(\M,g_\chi)$ and $(\M,g_0)$, respectively.\\
A  similar reasoning shows that $\omega_2^1 \in \Gamma'_c(\E\boxtimes \E)$, with 
$$\omega^1_2(\ff_1,\ff_2)=\omega_2^{\chi}(\R_{-}^{\dagger_{g_\chi g_1}}\ff_1,\R_{-}^{\dagger_{g_\chi g_1}}\ff_2)\:, $$
is Hadamard on $(\M,g_1)$ if and only if  $\omega_2^\chi$ is on $(\M,g_\chi)$. Combining the two results we have that 
$\omega'=\omega^1$ is Hadamard on $(\M,g'=g_1)$ if and only if $\omega= \omega^0$ is Hadamard on $(\M,g=g_0)$
concluding the proof.
\end{proof}

We are now in the position to prove our last result. 

\begin{theorem}\label{thm:main intro Had} Let $\E$ be an $\RR$-vector bundle on a  smooth manifold $\M$ equipped with a  positive, symmetric, fiberwise  metric  $\fiber{\cdot}{\cdot}$.
 Let  $g,g' \in \mathcal{GH}_\M$, consider the corresponding  formally-selfadjoint normally hyperbolic operators $\N,\N' : \Gamma(\E) \to \Gamma(\E)$ and refer to the associated CCR algebras $\mathcal{A}$ and $\mathcal{A}'$.\\
Let   $\nu   \in \Gamma_c'(\E \boxtimes\E)$ be of Hadamard type and satisfy
$$\nu(x,y) - \nu(y,x) = i\G_{\N}(x,y)\quad mod \quad C^\infty\:,$$
$\G_{\N}(x,y)$ being the distributional Kernel of $\G_{\N}$. \\
\label{key} Assuming  $g\simeq g'$,  let us define $$\nu' := \nu \circ \R^{\dagger_{gg'}} \otimes \R^{\dagger_{gg'}}\:,$$
for a M\o ller operator $\R: \Gamma(\E) \to \Gamma(\E)$ of $g,g'$.
Then the following facts are true.
\begin{itemize} 
\item[(i)]   $\nu$ and $\nu'$ are bisolutions mod $C^\infty$ of the field equations defined by $\N$ and $\N'$ respectively,
\item[(ii)] $\nu' \in  \Gamma_c'(\E \boxtimes\E)$,
\item[(iii)] $\nu'(x,y) - \nu'(y,x) = i\G_{\N'}(x,y)$ mod $C^\infty$,
\item[(iv)] $\nu'$ is of Hadamard type.
\end{itemize}
\end{theorem}

\begin{proof}
Since we never exploited the fact that $\omega$ is positive, nor the fact that the antisymmetric part of its two points function is {\em exactly} the causal propagator, nor the fact that the relevant algebras of fields are {\em on-shell} (i.e.,  the equation of motion are satisfied by the two-point functions), we can use the same arguments as in the proof of the previous theorem to conclude.
\end{proof}

We conclude this section with the following straightforward result of existence of Hadamard quasifree states which apparently does not use the Hadamard singularity propagation argument (actually this argument was used in the proof of Theorem \ref{thm:main Had}).
\begin{corollary}
Let $(\M,g)$ be a globally hyperbolic spacetime, $\N$ be a formally-selfadjoint normally hyperbolic operator   acting on the sections of the
$\RR$-vector bundle 
$\E$ over $\M$ and refer to the associated CCR algebras $\mathcal{A}$. Then there exists a Hadamard state on $\mathcal{A}$.
\end{corollary}
\begin{proof}
It is well-known \cite{Fulling} that, in a globally hyperbolic ultrastatic spacetime,  the (unique) CCR quasifree ground state which is invariant under the preferred Killing time is Hadamard. Hence, combining Corollary~\ref{cor:ultr para} with Theorem~\ref{thm:main Had} we can conclude.
\end{proof}

\section{Conclusion and future outlook}\label{sec:concl}
We conclude this paper by discussing some open issues which are raised in this paper and we leave for future works.
\paragraph*{Paracausally related metrics.}
One of the key ingredients in the realization of the M\o ller operator is the introduction of the new geometric notion that we have called {\em paracausal deformation}. In particular, we have seen in Section~\ref{sec:Moller}, that for any given two paracausally deformed metrics, there exists a M\o ller operator which realizes an algebraic  equivalence between corresponding free quantum field theories defined on different curved spacetimes (on a given manifold).
Therefore it seems very natural and important to classify the class of metrics which are not paracausally related in relation to the existence  of inequivalent quantum field theories.
We have already seen in Section~\ref{subsec:paracausal}, that if $(M,g)$ and $(M,g')$ admit a common foliation of Cauchy hypersurfaces, then $g$ and $g'$ are paracausally related.
 As suggested by part (3) of the Example~\ref{examples}, the claim that if $(M,g)$ and $(M,g')$ are Cauchy-compact globally hyperbolic spacetimes then $g$ and $g'$ are paracausally related, does not sound reasonable. The idea behind is that $g$ and $g'$ in part (3) of the Example~\ref{examples} have somehow `different time-orientation'. Since the time-orientation depends on the metric on $\M$, we have to provide a criteria to translate the requirement that $g$ and $g'$ are both `future-directed' in some sense.
Keeping in mind what said above, a conjecture which urges to be proved or disproved is the following one (maybe adding some further hypothesis).
\begin{conjecture}\label{conj1}
Let $t$ and $t'$ be  Cauchy temporal functions for globally hyperbolic spacetimes $(\M,g)$ and $(\M,g')$.
Finally denote with $\bra\cdot,\cdot\ket$ the natural paring between $\T^*\M$ and  $\T\M$. Then
$$g\simeq g'\qquad \text{if  and only if} \qquad \bra \partial_t, dt'\ket >0 \; \text{ and }\; \bra \partial_{t'}, dt \ket >0\,,$$
where $\partial_t$ (resp. $\partial_{t'}$) is the dual of $dt$ (resp. $dt'$) with respect to $g$ (resp. $g'$). 
\end{conjecture}
\begin{remark} The requirement $\bra \partial_t, dt'\ket >0$  implies that the integral curve $\gamma=\gamma(t)$ of $\partial_t$ on $(\M,g')$ satisfies $t'(\gamma(t_2))> t'(\gamma(t_1))$ if $t_2>t_1$. This requirement is weaker than assuming the $\partial_t$ is timelike and future-directed for $g'$. 
 The reason why we also impose $\bra \partial_{t'}, dt\ket >0$  is that  being paracausally related is an equivalence relation in $\mathcal{GH}_\M$. 
\end{remark}
A similar conjecture has to be established or disproved in the class of asymptotically flat spacetimes. 
\begin{conjecture}
If $(\M,g)$ and $(\M,g'$)  are asymptotically flat globally hyperbolic spacetimes then $g$ and $g'$ are paracausally related if  $\bra \partial_t, dt'\ket >0 $ and $\bra \partial_{t'}, dt \ket >0$.
\end{conjecture}
\begin{remark}
Differently from Conjecture~\ref{conj1}, the conjecture above would not provide a characterization of asymptotically flat spacetimes, since part (3) of the Example~\ref{examples} provide a counterexample to the `necessary' part of the statement. Indeed, the Minkowski metrics presented in Example~\ref{examples}  do not satisfy $\bra dt', \partial_t\ket >0$.
\end{remark}

If the conjectures are proved to hold, the results would suggest that free quantum field theories on globally hyperbolic spacetimes are more sensitive to the topology of the manifold $\M$ with respect to the metric $g$ endowing $\M$.  In particular,  the physics encoded in the quasifree states $\omega$ for a quantum field propagating on curved  spacetime ($\RR^4,g)$ can be found in the pull-backed state $\omega'=\omega\circ \cal R$ on Minkowski spacetime. Loosely speaking, the quantum effects due to the interaction between the quantum field and the classical gravitational potential can be thought as special observable in spacetimes without gravitational interaction.

\paragraph*{Homotopical properties of M\o ller operators}
Another issue which deserves to be investigated is the dependence of the M\o ller operator $\R$ of $g,g'$ and the associated $*$-isomorphism from the finite sequence of globally hyperbolic  metrics joining $g,g'$ in the sense of paracausal relationship. It is clear from the construction of $\R$  that the natural ``composition'' of sequences corresponds to the composition of operators.  
\begin{question}
Is there some homotopical notion in the space of globally hyperbolic metrics which is reflected in the space of M\o ller operators?
\end{question}

\paragraph*{Pull-back of ground and KMS states through the M\o ller $*$-isomorphism.}
Another issue we have faced is the lack of control on the action of the group of $*$-automorphism induced by the isometry group of the spacetime $\M$ on $\omega$. Let us remark, that the study of invariant states is a well-established research topic (\cf \cite{NCTori,NCspinoff}). Indeed, the type of factor can be inferred by analyzing which and how many states are invariant. From a more physical perspective instead, invariant states can represent equilibrium states in statistical mechanics {\it e.g.} KMS-states or ground states.
The previous remark leads us to the following open question:
\begin{question} Under which conditions it is possible to perform an adiabatic limit, namely when is $ \lim\limits_{\chi\to 1} \omega_1$ well-defined? 
\end{question}

A priori we expect that there is no positive answer in all possible scenarios, since
it is known that certain free-field theories, {\it e.g.}, the massless and minimally coupled (scalar or Dirac) field on four-dimensional de Sitter spacetime, do not possess a ground state, even though their massive counterpart does.
(Notice that this is not a no-go Theorem, but at least an indication that, in these situations, the map $\omega\to\omega\circ\mathcal R$ cannot be expected to preserve the ground state property.)

A partial investigation in this direction has been carried on in \cite{Moller,Moller2} for the case of a scalar field theory on globally hyperbolic spacetimes with empty boundary. In this situation it has been shown that, under suitable hypotheses the adiabatic limit can be performed preserving the invariance property under time translation but spoiling in general the ground state or KMS property.\medskip

\paragraph*{M\o ller $*$-isomorphism in perturbative AQFT.}
We conclude with the following observation. The results of Section~\ref{sec:Moller AQFT} are valid also for off-shell algebras as well as for distribution of Hadamard type. Therefore, it could be possible to extend the action of the M\o ller operator also on the algebra of extended observables in a perspective of {\em deformation product quantization} (see for instance Section 2 of \cite{DHP}), which include, e.g., the Wick polynomials of the underlying fields.
Wick polynomials and time-ordered products of Wick polynomials are the building blocks for
perturbative renormalization of quantum fields, both in Minkowski spacetime and in curved
spacetime, where the metric is considered as a given external classical field. Although of utmost physical relevance, these formal operators  as the {\em stress-energy operator} do not belong to the algebra of observables generated by the
smoothly smeared field operators (operator-valued distributions). This is because they correspond to products of distributions at a given point and this notion is not well-defined in general. The popular and perhaps most
effective procedure to eliminate the short-distance divergences consists of simply subtracting a suitable Hadamard distribution.
This procedure is systematically embodied  in a product deformation quantization procedure which relies on a suitable set of functionals with a specific wavefront set.  
  The following observation leads to the following conjecture:
\begin{conjecture}
Let ${\cal A}_0, {\cal A}'_0$ be the algebra of observables 
of the  globally hyperbolic spacetimes $(\M,g)$ and $(\M,g')$
and ${\cal R}_0$ a M\o ller  $*$-isomorphism of them. 
If  ${\cal A}, {\cal A}'$ are corresponding 
 extended algebras of observable (which include the Wick polynomials etc.) and $g\simeq g'$, then ${\cal R}_0$ extends to  a (M\o ller) $*$-isomorphism $\cal R: \cal A \to \cal A'$.
\end{conjecture}

\paragraph*{M\o ller operators and gauge theory.}

Last but not least, all our results concern vector-valued fields of a vector bundle $\E$ over $\M$. The fiber metric on the  bundle does not depend on the globally hyperbolic metrics $g$ chosen on $\M$. A natural extension of the formalism would be  the inclusion of that $g$-dependence in the fiber metric. This extension would allow to encompass the case of a {\em Proca field} and, possibly the case of the electromagnetic field, though  issues  with gauge invariance and gauge  fixing are expected to pop out.

\end{document}